%% file: main.tex
\documentclass[11pt]{article}

\usepackage[letterpaper,margin=1in]{geometry}
\usepackage[T1]{fontenc}
\usepackage[utf8]{inputenc}
\usepackage{lmodern}
\usepackage{microtype}
\usepackage{amsmath,amssymb,amsthm,mathtools,mathrsfs}
\usepackage{graphicx}
\usepackage[dvipsnames]{xcolor}
\usepackage[numbers,sort&compress]{natbib}
\usepackage[
  colorlinks=true,
  linkcolor=BrickRed,
  citecolor=PineGreen,
  urlcolor=Blue
]{hyperref}
\usepackage[capitalise,noabbrev]{cleveref}

\newtheorem{theorem}{Theorem}[section]
\newtheorem{proposition}[theorem]{Proposition}
\newtheorem{lemma}[theorem]{Lemma}
\newtheorem{corollary}[theorem]{Corollary}
\theoremstyle{definition}
\newtheorem{definition}[theorem]{Definition}
\theoremstyle{remark}

\newcommand{\CCZ}{\mathrm{C}\mathrm{C}Z}
\newcommand{\CNOT}{\mathrm{CNOT}}
\newcommand{\CZ}{\mathrm{C}Z}
\newcommand{\SWAP}{\mathrm{SWAP}}
\newcommand{\supp}{\operatorname{supp}}
\newcommand{\ket}[1]{\lvert #1\rangle}
\newcommand{\bra}[1]{\langle #1\rvert}
\newcommand{\im}{\operatorname{im}}

\allowdisplaybreaks[2]
\title{Purely-logarithmic-time- and constant-space-overhead\\ fault-tolerant quantum computation}
\author{
  Zhengyi Han \qquad Zi-Wen Liu\\[0.5em]
  {\small Yau Mathematical Sciences Center, Tsinghua University}
}
\date{\today}

\begin{document}

\maketitle

\input{sections/abstract}

\tableofcontents

\input{sections/introduction}
\input{sections/model}
\input{sections/main-proof}
\input{sections/state-preparation}
\input{sections/addressable-logic}
\input{sections/qltc-decoder}
\input{sections/injection-route}
\input{sections/conclusion}

\section*{Acknowledgments} 
This work is supported in part by NSFC under Grant No.~12475023, Dushi Program, and a startup funding from YMSC.

\section*{AI disclosure}
The authors set the overall strategy, and generative AI tools provided substantial assistance in developing the proof details and presentation. The authors take full responsibility for the results.

\bibliographystyle{unsrt}
\bibliography{references}

\end{document}

%% file: sections/abstract.tex
\begin{abstract}
We prove that constant-space-overhead fault-tolerant quantum computation can be achieved with provably strictly logarithmic time overhead, improving over the best known results with additional subpolylogarithmic factors.
Our main construction uses polynomial-subrank transversal logical $\CCZ$ gates on good quantum locally testable codes to implement addressable universal computation by transferring batches of logical qubits between dense storage and active logical subspaces while reusing the same ancillary workspace.
Logical $\CCZ$ gates are implemented directly by the transversal operation, so only stabilizer resource states require separate preparation.
Furthermore, we give an alternative construction that also achieves purely logarithmic time overhead based on modifying the quantum Reed--Solomon magic-state distillation scheme of Nguyen and Pattison.  Recursively applying a fixed distillation circuit protected by qLTCs of increasing block length eliminates the  subpolylogarithmic time factor.
\end{abstract}

%% file: sections/introduction.tex
\section{Introduction}
\label{sec:introduction}

The fault tolerance threshold theorem establishes that arbitrarily long quantum computations
can be protected against noise below a constant physical error rate
\cite{AharonovBenOr1999,Kitaev1997,KnillLaflammeZurek1998}.  Conventional
concatenated code constructions achieve this reliability with
polylogarithmic overhead in both space and time.  
Gottesman showed that suitable qLDPC codes can be used to achieve fault-tolerant quantum computation with constant space overhead and polynomial time overhead~\cite{Gottesman2013}, and Fawzi, Grospellier, and Leverrier subsequently realized this approach using quantum expander codes~\cite{Fawzi2018PolyTime}.  The time overhead was subsequently reduced in a series of works.
Yamasaki and Koashi achieved quasi-polylogarithmic time overhead~\cite{YamasakiKoashi2024}, and Tamiya, Koashi, and Yamasaki achieved polylogarithmic time overhead~\cite{TamiyaKoashiYamasaki2024}, both accounting for the runtime of classical processing.
Nguyen and Pattison obtained $O(\log^{1+o(1)}(|C|/\varepsilon))$ time overhead~\cite{NguyenPattison2025}.

Here we achieve a
strictly logarithmic time overhead, eliminating the remaining subpolylogarithmic factors.
Our main construction is based on the recent good qLTC with polynomial-subrank transversal logical $\CCZ$ result of Li, Li, and Liu~\cite{GoodQLTCTransversal2026}, which enables the computational code to 
perform a sufficiently large number of independent non-Clifford gates directly.   More precisely, for code blocks of size $\Theta(n)$, a single transversal
layer implements $r=n^{\Omega(1)}$ independent logical $\CCZ$ gates
on suitable active logical subspaces.
Logical routing and coordinate selection allow these gates to act on the desired logical triples.
We process the logical qubits of each dense storage block in $O(n/r)$ successive batches, with different block groups processed in parallel and the same ancillary workspace reused between batches.
Each batch takes $O(r+\operatorname{polylog} n)$ time, including transfers between storage and active subspaces and resource state preparation.
The resulting physical time per ideal layer is $O(n+\log W)$,
including routing and classical control.
Choosing $n=\Theta(\log(WD/\varepsilon))$ therefore gives
$O(n)$ physical time per ideal layer and strictly logarithmic time overhead.

For an adaptive Clifford+$\CCZ$ circuit of width $W$, depth $D$, and target
error $\varepsilon$, we use code blocks of length
$n=\Theta(\log(WD/\varepsilon))$.
Below a constant physical noise threshold, the resulting fault-tolerant
circuit uses
\[
  O\bigl(W+F(n)\bigr)
  \quad\text{physical qubits and has depth}\quad
  O\!\left(D\log\frac{WD}{\varepsilon}\right),
\]
while reproducing the ideal classical output distribution within error
$\varepsilon$.
Here $F(n)$ is the additional space required for state preparation,
switching, and routing, with
$\log F(n)=O(\log n\log\log n)$.
The space overhead is constant whenever $W\ge C_F\,F(n)$.
The depth bound includes the time used for classical processing,
and the ideal circuit may have arbitrary connectivity.

We make use of the avoiding-set and weight-enumerator calculus developed
by Nguyen and Pattison~\cite{NguyenPattison2025} to compose our
fault-tolerant gadgets.
We adapt their single-shot decoding argument to the qLTC
family used here.
To prepare the canonical logical $Y$ states needed for phase gates,
we combine their qLTC-protected prepare--unencode--test--distill approach
with constant-yield Hamming distillation.

We also present an independent construction based on the $\ket{\CCZ}$-state injection scheme of Nguyen and
Pattison~\cite{NguyenPattison2025}.
We recursively apply a fixed finite instance of their punctured quantum
Reed--Solomon distillation circuit, protected by qLTCs of increasing
block length. The resulting fault-tolerant circuit has the same
$O(D\log(WD/\varepsilon))$ depth bound, with a larger additive space term.

\subsection{Organization}

Section~\ref{sec:model} defines the adaptive circuit, noise model, and
weight enumerator composition.  Section~\ref{sec:main-proof} states the
theorem based on transversal $\CCZ$ and derives it from results proved in the later sections.
Section~\ref{sec:state-preparation} constructs encoded stabilizer resources,
Section~\ref{sec:addressable-logic} develops direct block code switching and
addressable transversal gates, and Section~\ref{sec:qltc-decoder} establishes
the code, decoder, memory, and length selection statements.  The independent
$\CCZ$ state injection construction is proved in
Section~\ref{sec:injection-route}, followed by a discussion of the two
constructions.

%% file: sections/model.tex
\section{Circuit and noise model}
\label{sec:model}

This section defines the circuit, noise, correctness, and composition notions
used throughout the paper.  They are independent of the code family.
Correctness is proved conditionally on a fault support that avoids a specified
bad family.  The locally stochastic condition then bounds the probability of
that family.  Weight enumerators keep these deterministic and probabilistic
parts separate under gadget composition.  This is the order of analysis used
in Ref.~\cite{NguyenPattison2025}, with the physical noise rate and target
output error kept as distinct parameters.

The ideal computation is specified first.  Width always means quantum width.
For a positive integer $m$, write $[m]=\{1,\ldots,m\}$.

\begin{definition}[Adaptive quantum circuit]
\label{def:adaptive-circuit}
An adaptive Clifford+$\CCZ$ circuit $C$ of width $W$ and depth $D$ is
specified by an initial classical map $C_0$ and, for every $t\in[D]$, a list
$(\mathsf O_{t,1},\ldots,\mathsf O_{t,w_t})$ of quantum operations with
disjoint supports and by a noiseless classical circuit
$C_t$.  The allowed operations fall into three classes.  A unitary location
applies either a one- or two-qubit Clifford gate or a $\CCZ$ gate; a
measurement location performs a destructive Pauli measurement; and an
initialization location prepares a state.

The computation maintains a classical register of polynomial size
$\vec x_t$.  A unitary operation $\mathsf O_{t,j}$ may be controlled by a bit of
$\vec x_t$; a nonunitary operation occurs unconditionally at its designated
location.  If $\vec m_t$ is the measurement record produced at time $t$,
then the register is updated once at the end of the layer according to
\begin{equation}
  \vec x_{t+1}=C_t(\vec x_t,\vec m_t).
  \label{eq:classical-update}
\end{equation}
The map $C_0$ initializes the control register from the classical input.
After the prescribed readouts in the last quantum layer have completed,
$C_D$ is evaluated once and the prescribed part of its output is reported.  The
quantum input is $\ket0^{\otimes W}$.  We put
\begin{equation}
  N=WD.
  \label{eq:padded-volume}
\end{equation}
\end{definition}

The circuits $C_t$ have constant depth.  The fault-tolerant simulation will
also use noiseless classical circuits of polynomial size.  Their hardware is
not included in the quantum width, but their depth is included in physical
time.  Thus a classical computation of depth $T$ makes the quantum data wait
for $O(T)$ physical time steps, with error correction applied throughout the
wait.  Both the ideal circuit and its simulation may use arbitrary physical
connectivity.

\subsection{Codes and reduced Pauli weight}

Let
\begin{equation}
  \mathcal Q=\operatorname{CSS}(H_X,H_Z),\qquad H_XH_Z^T=0,
  \label{eq:css-code}
\end{equation}
be a binary CSS code.  It is an $[[n,k,d]]$ code when it has $n$ physical
qubits, $k$ logical qubits, and distance $d$.  We use a chosen encoding
isometry $\mathcal E$ whenever a logical basis is required.

\begin{definition}[Quantum LDPC and locally testable codes]
\label{def:qltc}
Let $H\in\mathbb F_2^{m_H\times n}$ and let
$\mathcal C=\ker H\subseteq\mathbb F_2^n$.  The classical code
$\mathcal C$ is locally testable with soundness $\rho$ and locality
$\Delta$ if the row and column weights of $H$ are at most $\Delta$ and
\begin{equation}
  \frac{|Hx|}{m_H}\ge
  \rho\,\frac{\min_{c\in\ker H}|x+c|}{n}
  \label{eq:classical-local-testability}
\end{equation}
for every $x\in\mathbb F_2^n$.  The CSS code
$\mathcal Q=\operatorname{CSS}(H_X,H_Z)$ is $\Delta$-qLDPC if both check matrices have row
and column weights at most $\Delta$.  It is a $(\rho,\Delta)$-qLTC if the
two classical codes defined by $H_X$ and $H_Z$ satisfy
\eqref{eq:classical-local-testability}, with their respective numbers of
listed checks in the denominator.
\end{definition}

The listed checks may be dependent; their number is part of the tester
normalization.  The distance to $\ker H$ in
\eqref{eq:classical-local-testability} is not the same as distance modulo
the stabilizer space of the same Pauli type.  Both notions occur below.

For a Pauli $P$ on a stabilizer code, its weight reduced modulo stabilizers is
\begin{equation}
  |P|_R=\min_{S\in\operatorname{Stab}(\mathcal Q)}|SP|.
  \label{eq:reduced-weight}
\end{equation}
When a binary vector represents one Pauli type, the same notation denotes
minimum Hamming weight modulo the corresponding binary stabilizer space.

\subsection{Noise model}

Fault supports are defined independently of a probability distribution.  A
location is one of the quantum operations above together with its physical time
step.  Identity operations on idle qubits are locations, and every qubit
belongs to exactly one location at each physical time step.  Let $\Omega$
denote the resulting set of locations.

\begin{definition}[Fault and fault path]
\label{def:fault-path}
A fault $\mathbf f$ is a sequence of superoperators inserted between the
ideal layers of the physical circuit.  It is physical when every inserted
map is completely positive and trace preserving.  The fault path
$\supp\mathbf f\subseteq\Omega$ is the set of locations on which these maps
act nontrivially.
\end{definition}

The values of the fault maps are not required to factor across locations.
Only their support enters the probability hypothesis.

\begin{definition}[Locally stochastic noise]
\label{def:locally-stochastic}
A random physical fault $\mathbf f$ is locally stochastic with rate $p$ if
\begin{equation}
  \Pr[S\subseteq\supp\mathbf f]\le p^{|S|}
  \label{eq:locally-stochastic}
\end{equation}
for every $S\subseteq\Omega$.  Conditioned on its fault path, faulty
locations may apply arbitrary CPTP noise channels.
\end{definition}

No independence assumption is made in \eqref{eq:locally-stochastic}.  The
fault values may be correlated with one another and with all adaptive
measurement records.  We will repeatedly use the following reduction
conditioned on the measurement history.

\begin{lemma}[Pauli expansion with the measurement history]
\label{lem:recordwise-Pauli}
Condition on a physical fault support and an actual measurement history of a Clifford
instrument.  Every supported fault map can be expanded into superoperators
of the form $\rho\mapsto A\rho B$, where $A$ and $B$ are Paulis supported in
the same faulty locations.  If a support statement and a logical identity
hold for every nonzero left and right Pauli term with that measurement
history, then they hold for the original physical instrument on the same
support, including on an arbitrary reference system.
\end{lemma}

\begin{proof}
Pauli operators form a basis for operators on the support of each fault.
Expand every Kraus operator on that basis and retain the classical record in
the output.  Projectors associated with incompatible records annihilate the
corresponding terms.  The surviving terms therefore have the stated
left--right Pauli form, and the classical correction computed from the
record is the same physical operation for all of them.  Linearity gives the
logical identity once it has been checked term by term.  Complete positivity
and normalization are properties of the original instrument; the expansion
coefficients are never interpreted as probabilities and never enter the
bound on bad fault paths.
\end{proof}

\begin{lemma}[Syndrome measurement]
\label{lem:syndrome-measurement}
Let $\Pi_s$ be the projector onto syndrome sector $s$ of an
$[[n,k,d]]$ stabilizer code.  Let $\mathcal E_B$ be a superoperator
supported on a set $B$ that is correctable as an erasure, and let $\rho$ be
a code state, possibly entangled with a reference.  If all stabilizer checks
are measured and the outcome is retained, then the branch with outcome $s$
can be written
\begin{equation}
  \lvert s\rangle\!\langle s\rvert\otimes
  \alpha_s E_s\rho E_s^\dagger
  \label{eq:syndrome-measurement}
\end{equation}
for a Pauli $E_s$ supported on $B$ and a scalar $\alpha_s$.  For a physical
map, summing the branches gives the original completely positive
trace-preserving measurement instrument.
\end{lemma}

\begin{proof}
Expand $\mathcal E_B$ as
\[
  \mathcal E_B(\rho)
  =\sum_{\mu,\nu}\alpha_{\mu\nu}A_\mu\rho B_\nu
\]
with Paulis $A_\mu,B_\nu$ supported on $B$.  Since
$\rho=\Pi_0\rho\Pi_0$, the branch projectors annihilate every term unless
$A_\mu$ and $B_\nu$ both have syndrome $s$.  Any two Paulis supported on
$B$ with that syndrome have a product of trivial syndrome supported on
$B$.  Erasure correctability implies that this product is a stabilizer, so
all surviving left actions agree on the code space up to a scalar, and the
same holds on the right.  Choosing one representative $E_s$ gives
\eqref{eq:syndrome-measurement}.  The reference system is untouched
throughout the calculation.
\end{proof}

This is the precise reason that the later proofs may analyze Pauli
representatives after check measurement while retaining arbitrary CPTP
faults in the physical model.  The statement is not a stochastic
decomposition of the physical channel.

\subsection{Weight enumerators}

The deterministic part of the proof describes failure through families of
bad subsets of locations.

\begin{definition}[Avoiding sets and weight enumerators]
\label{def:weight-enumerator}
Let $\mathcal F\subseteq\mathcal P(\Omega)$ be a family of bad sets.  A set
$X\subseteq\Omega$ is $\mathcal F$-avoiding if no $F\in\mathcal F$ is
contained in $X$.  The weight enumerator of $\mathcal F$ is
\begin{equation}
  \mathcal W(\mathcal F;x)=\sum_{F\in\mathcal F}x^{|F|}.
  \label{eq:weight-enumerator}
\end{equation}
\end{definition}

If $\mathbf f$ is locally stochastic with rate $p$, then
\begin{equation}
  \Pr[\supp\mathbf f\text{ is not }\mathcal F\text{-avoiding}]
  \le \mathcal W(\mathcal F;p).
  \label{eq:enumerator-probability}
\end{equation}
This is a union bound over the members of $\mathcal F$ and does not use
independence.

Suppose $\Omega=\Omega_1\cup\Omega_2$, and let
$\mathcal F_i\subseteq\mathcal P(\Omega_i)$.  Under the natural inclusions,
define
\begin{align}
  \mathcal F_1\boxplus\mathcal F_2
    &=\mathcal F_1\cup\mathcal F_2,
    \label{eq:bad-family-sum}\\
  \mathcal F_1\circledast\mathcal F_2
    &=\{F_1\cup F_2:F_i\in\mathcal F_i\}
    \quad\text{when }\Omega_1\cap\Omega_2=\varnothing.
    \label{eq:bad-family-product}
\end{align}
The corresponding enumerators satisfy
\begin{align}
  \mathcal W(\mathcal F_1\boxplus\mathcal F_2;x)
    &\le \mathcal W(\mathcal F_1;x)+
       \mathcal W(\mathcal F_2;x),
    \label{eq:enumerator-sum}\\
  \mathcal W(\mathcal F_1\circledast\mathcal F_2;x)
    &=\mathcal W(\mathcal F_1;x)
      \mathcal W(\mathcal F_2;x).
    \label{eq:enumerator-product}
\end{align}
The second equality follows from disjoint spacetime locations, not from
probabilistic independence.  It is the estimate used for simultaneous bad
raw slots in the canonical-$Y$ preparation of
\Cref{sec:state-preparation}.

For nested gadgets, let a finite set $I$ label outer locations.  For each
$i\in I$, let $\Omega_i$ be the disjoint location set of the inner gadget
substituted at $i$, and let
$\mathcal S_i\subseteq\mathcal P(\Omega_i)$ be its family of bad sets.  If
$\mathcal F\subseteq\mathcal P(I)$ is the outer family of bad sets, define
\begin{equation}
  \mathcal F\bullet\{\mathcal S_i\}_{i\in I}
  =\mathop{\boxplus}_{F\in\mathcal F}
     \mathop{\circledast}_{i\in F}\mathcal S_i.
  \label{eq:bad-family-composition}
\end{equation}
Here an empty product contributes the family containing the empty set.  If
$\mathcal W(\mathcal S_i;x)\le q(x)$ for every $i$, then
\begin{equation}
  \mathcal W\!\left(
      \mathcal F\bullet\{\mathcal S_i\}_{i\in I};x\right)
  \le \mathcal W(\mathcal F;q(x)).
  \label{eq:enumerator-composition}
\end{equation}
Indeed, for each $F$, the product rule gives at most
$q(x)^{|F|}$, and summing over $F\in\mathcal F$ yields the right-hand side.
The Hamming factory applies this proposition at each level of its
predetermined grouping of inputs.

\subsection{Fault-tolerant gadgets and composition}

We now connect families of bad fault paths to the quantum input and output
conditions used in the proof.  The definition below allows any constant
number of input and output blocks, including no input blocks for preparation
and no output blocks for destructive measurement.  Its classical output
always includes the complete measurement history.

For the concrete gadgets below we use one common boundary condition stated in
terms of weight reduced modulo stabilizers.
Let $t_*>0$.  On block $b$, let $\mathcal E_b$ be its encoding
isometry and let $n_b$ and $d_b$ be its length and distance.  The permitted
left and right operators belong to the span of
\begin{equation}
  P_b\mathcal E_b,\qquad |P_b|_R\le t_*n_b,
  \label{eq:correctable-span}
\end{equation}
where $P_b$ is a Pauli.  For a finite list of blocks, an $m$-qubit reference
system, and a logical--reference state $\rho$, let $I_m$ denote the identity
on the reference system.  We say that a physical state satisfies the
boundary condition about $\rho$ if it belongs to the linear span of
operators of the form
\begin{equation}
 \left(\left[\bigotimes_b P_b\mathcal E_b\right]\otimes I_m\right)
 \rho
 \left(\left[\bigotimes_b Q_b\mathcal E_b\right]^\dagger\otimes I_m\right),
 \qquad |P_b|_R,|Q_b|_R\le t_*n_b,
 \label{eq:typed-boundary-map}
\end{equation}
where $P_b$ and $Q_b$ are Paulis.  The bound is imposed separately on every
block.  Empty input and output lists give the preparation and destructive
measurement cases, respectively.  If
$2t_*n_b<d_b$ for every block, the usual error correction condition applies
to the whole span: the product of two permitted representatives is either
detectable or a stabilizer, and is never a nontrivial logical Pauli.

\begin{definition}[Fault-tolerant gadget]
\label{def:ft-gadget}
Let $\mathsf{Gad}$ be a physical instrument with any finite list of
encoded input blocks and output blocks, and let $\mathsf{Op}$ be the intended
logical instrument.  Let $\mathcal G$ be a family of bad subsets of the
locations of $\mathsf{Gad}$.  The gadget is fault tolerant if, whenever the
fault path is $\mathcal G$-avoiding, the following holds for every $m\ge0$
and every logical--reference input state $\rho$.  Let $\mathcal I_m$ be the
identity channel on the $m$-qubit reference system.  If the physical input
state satisfies \eqref{eq:typed-boundary-map} about $\rho$, then every
output branch satisfies the same condition, with the same reference system
and the same bound on every output block, about the corresponding branch of
$(\mathsf{Op}\otimes\mathcal I_m)(\rho)$.  The complete classical record is
retained.
\end{definition}

The definition includes preparation gadgets with no input, gates on several blocks,
gadgets mapping one code to another, destructive measurements, and parallel calls.
It is deliberately stated for instruments rather than only their averaged
quantum channels, because later classical corrections depend on the actual
record.

\begin{proposition}[Sequential and parallel composition]
\label{prop:gadget-composition}
Suppose two fault-tolerant gadgets have matching boundary conditions.
Their sequential composition is fault tolerant outside the sum of their
families of bad fault paths.  A parallel tensor product of finitely many
gadgets on disjoint blocks is fault tolerant outside the corresponding sum
of these families.  Both statements remain valid in the presence of an arbitrary
reference system.
\end{proposition}

\begin{proof}
For sequential composition, the output of the first gadget satisfies the
boundary condition required by the second.  Avoidance of the sum means that both
individual bad families are avoided.  For parallel composition, apply the
definition to one factor while treating all other logical blocks and the
external system as its reference, and repeat for the remaining factors.
The record registers are retained throughout, so the same argument applies
to instruments.
\end{proof}

\begin{lemma}
\label{lem:instrument-normalization}
Suppose that, for every complete record of a physical instrument and every
left--right Pauli term supported on an allowed fault set, the corrected
branch implements the prescribed logical map on the span defined in
\eqref{eq:typed-boundary-map}.  For a logical unitary, assume each branch is
a scalar multiple of the same logical isometry.  For a logical Pauli
measurement, assume each branch is a scalar multiple of the ideal outcome
map with projector $P_s$.  Then the physical instrument implements the
complete ideal logical instrument, including the outcome probabilities and
its action on arbitrary reference systems.
\end{lemma}

\begin{proof}
For a unitary branch, trace preservation of the actual physical instrument
forces the sum of its positive scalar weights to one.  For a measurement,
sum over all internal records that produce the same logical outcome $s$.
The branch calculation gives nonnegative coefficients $c_s$, and trace
preservation gives
\begin{equation}
  \sum_s c_sP_s=I.
  \label{eq:measurement-normalization}
\end{equation}
Orthogonality of the nonzero outcome projectors implies $c_s=1$ on every
outcome sector.  Thus the proof determines the full instrument, not only its
normalized states after measurement.  The coefficients used in the Pauli expansion in
the branch calculation need not be positive and are never used as
probabilities.
\end{proof}

\begin{proposition}
\label{prop:nested-gadget-composition}
Suppose an outer gadget is correct provided the set of failed inner
locations avoids $\mathcal F$, and the inner gadget at $i$ is correct when
its fault path avoids $\mathcal S_i$.  Then the substituted physical
gadget is correct outside
$\mathcal F\bullet\{\mathcal S_i\}_{i\in I}$.  If every inner enumerator is
at most $q(p)$, the probability that it is bad is at most
$\mathcal W(\mathcal F;q(p))$.
\end{proposition}

\begin{proof}
If the substituted fault path avoids the composed family, the set of
inner gadgets whose support contains a bad set is $\mathcal F$-avoiding.
All other inner gadgets satisfy their logical guarantees.  The failed inner
locations therefore form an allowed set for the outer gadget, which is correct.  The
probability estimate is \eqref{eq:enumerator-composition} followed by
\eqref{eq:enumerator-probability}.
\end{proof}

\subsection{Corrected gadget boundaries}

The preceding definitions are independent of the particular code family.
All later gadgets will use the same $t_*$ in
\eqref{eq:typed-boundary-map}.  After a gate with bounded spread, error
correction must reduce the output error to this radius.  Preparations and
switching must produce outputs within the same radius, starting either with
no logical input or with another encoded block.  The decoder specific to
this construction is stated in \Cref{lem:decoder-interface}
and proved in \Cref{sec:single-shot-decoder}; the circuit level error
correction argument is proved in \Cref{lem:ec-gadget}.

\begin{definition}[Corrected gadget boundary]
\label{def:normalized-boundary}
A gadget output is at the corrected boundary when it uses the canonical
encoding required by the compiled schedule, every complementary logical
coordinate is in its prescribed zero or plus state, and all Pauli corrections
determined by syndrome or logical measurement outcomes have been computed and
physically applied.  The waiting time for these operations is included in the
gadget.  After the final error correction cycle, the output satisfies
\eqref{eq:typed-boundary-map}.
\end{definition}

All preparation, switching, and active logical operations used below produce
outputs at this corrected boundary.  In particular, all three operands of a
transversal $\CCZ$ satisfy the boundary condition before the gate.  Deferred
Clifford Pauli frames may toggle later Clifford measurement records, but no
Pauli frame of large weight is commuted symbolically through a non-Clifford layer.

Two elementary support facts will be used repeatedly.  First, a physical
unitary of bounded depth maps an operator supported on $S$ into an
operator supported on a forward light cone of constant size around $S$.  Expanding the
result in Paulis does not enlarge that light cone.  Second, measuring the
stabilizer checks of a correctable error map while retaining the outcome
separates Pauli terms by syndrome.  Pauli terms with the same syndrome differ
by a stabilizer on the code space.  These facts, together with
\Cref{lem:recordwise-Pauli}, are what permit the calculations with weight
reduced modulo stabilizers in later sections to cover arbitrary CPTP noise channels at the faulty
locations.

No conclusion specific to the code construction is used in this section.  In
particular, qLTC soundness by itself does not give an
efficient single-shot decoder, and a small syndrome does not rule out a
logical Pauli.  The additional geometric reductions needed for the decoder
and memory estimates are proved separately in
\Cref{sec:single-shot-decoder}.

%% file: sections/main-proof.tex
\section{Fault-tolerant quantum computation with transversal $\CCZ$}
\label{sec:main-proof}

This section derives the main theorem from the state preparation, logical
gate, and decoding results proved in the following sections.  The
fault-tolerance argument uses the stated code properties and gadget bounds;
the arithmetic construction of the code family enters only through the
results collected below.

\subsection{The computational code}

\subsubsection{Code parameters and error correction}

For every sufficiently large $d$, let $P_j(d)$, $j\in[3]$, be the three
primal CSS codes associated with the three arguments of the transversal
$\CCZ$ action, and let $D_j(d)$ be their physical-$H$ dual codes.  Set
\begin{equation}
  \mathfrak C_d
  =\{P_1(d),P_2(d),P_3(d),D_1(d),D_2(d),D_3(d)\}.
  \label{eq:computational-code-set}
\end{equation}
The usable combined lengths are
\begin{equation}
  n_d=n_*(Q^{2d}-1),
  \qquad
  \frac{n_{d+1}}{n_d}\le Q^2+1
  \label{eq:dense-lengths-main}
\end{equation}
for every sufficiently large $d$.

For a fixed sufficiently large $d$, set
\begin{equation}
  \mathfrak C:=\mathfrak C_d,
  \qquad n:=n_d,
  \qquad \tau:=\tau_d,
  \qquad k_{\CCZ}:=k_{\CCZ}(\tau_d).
  \label{eq:fixed-code-notation}
\end{equation}
We suppress $d$ on the six codes and call a physical block encoded by
$\mathcal Q\in\mathfrak C$ a $\mathcal Q$ block.  Write $H_j$ for the
logical space of $P_j$.  Every block has length $\Theta(n)$, bounded check
weight and incidence, and linear distance.  The primary block has dimension
\begin{equation}
  k\ge \kappa n
  \label{eq:primary-rate-main}
\end{equation}
for a constant $\kappa>0$; the auxiliary blocks are used at their actual
dimensions.

By \Cref{thm:qltc-family}, there are constants
$\sigma\in(0,1]$ and $c_\sigma>0$ and injections
\[
  i_j:\mathbb F_2^{k_{\CCZ}}\longrightarrow H_j,
  \qquad j\in[3],
\]
satisfying
\begin{equation}
  k_{\CCZ}\ge c_\sigma n^\sigma,
  \qquad
  \tau(i_1x,i_2y,i_3z)
    =\sum_{v=1}^{k_{\CCZ}}x_vy_vz_v
  \quad\text{for every }x,y,z\in\mathbb F_2^{k_{\CCZ}}.
  \label{eq:main-subrank}
\end{equation}
The injections and their extensions to canonical logical bases are
determined during compilation.  The estimate in \eqref{eq:main-subrank}
controls the number of active bands and is used in the final time bound.

\begin{lemma}[Decoder contraction and fault-tolerant memory]
\label{lem:decoder-interface}
For each $\mathcal Q\in\mathfrak C$ and every $\eta>0$, there
are constants $\alpha,\beta,\gamma_\mathrm{dec}>0$ and a decoder of constant depth
$\mathcal D$ such that
\begin{equation}
  \bigl|e+\mathcal D(He+m)\bigr|_R
  \le \eta|e|_R+\gamma_\mathrm{dec}|m|
  \label{eq:main-decoder-interface}
\end{equation}
whenever
$|e|_R+\alpha|m|\le\beta n$.  For an exact syndrome, an
$O(\log n)$ depth version returns a correction of absolute weight $O(|e|_R)$
leaving a stabilizer.  There is a common radius $t_*n$ and a constant
physical threshold for the finite list of physical gate segments used in the
simulation.  Each segment preserves the codespace and acts on at most a constant
number of code blocks.  After error correction, the output error again has
reduced weight at most $t_*n$
unless a constant
positive fraction of its $O(n)$ locations are faulty.  The corresponding
event induced by the corresponding bad support family has probability
$e^{-\Omega(n)}$, uniformly over arbitrary supported CPTP
faults and arbitrary reference systems.
\end{lemma}

\Cref{lem:decoder-interface} is proved in
\Cref{sec:single-shot-decoder}; its version relative to one stored syndrome
baseline is \Cref{lem:sector-storage}.

\begin{lemma}[Fault-tolerant error correction]
\label{lem:ec-gadget}
The syndrome extraction circuit followed by the decoder of
\Cref{lem:decoder-interface} is an error correction gadget compatible with
the boundary condition.  If $t_\mathrm{in}$ is the incoming reduced error weight
and $f$ is the number of fresh faults in one physical window of bounded depth, then
\begin{equation}
  t_\mathrm{out}
  \le 2\eta t_\mathrm{in}+(2\gamma_\mathrm{dec}+1)h f
  \label{eq:physical-memory-interface}
\end{equation}
for a light cone constant $h$.  The decoder contraction, boundary radius,
and allowed fault fraction can be chosen so that this map sends every input
within the common radius $t_*n$ to an output within the same radius.  A bad window has probability
$e^{-\Omega(n)}$ below a constant physical threshold.  The statement holds
with arbitrary references and arbitrary supported CPTP faults.
\end{lemma}

The common choice of constants and the version using a stored syndrome
baseline are proved in \Cref{sec:single-shot-decoder}.

\subsubsection{Logical operations}

Let $\lambda_Y\ge1$ and set
\begin{equation}
  \ell(n)=\left\lceil\log_2(\lambda_Y n)\right\rceil,
  \qquad
  h_j=\left\lfloor2\log_2(4j)\right\rfloor,
  \qquad
  u_j=2^{h_j}-1,
  \label{eq:global-batch-functions}
\end{equation}
and
\begin{equation}
  M(n)=\prod_{j=1}^{\ell(n)}u_j,
  \qquad
  K(n)=\prod_{j=1}^{\ell(n)}(u_j-2h_j).
  \label{eq:y-batch-yield}
\end{equation}
The Hamming lengths increase with the purification level, and there is a
constant $\eta_Y>0$ such that
\begin{equation}
  K(n)\ge\eta_YM(n),
  \qquad
  \log M(n)=\Theta(\log n\log\log n).
  \label{eq:y-constant-yield}
\end{equation}
Choose a constant $b>0$ large enough for all preparation, readout, and local
classical computations, and define
\begin{equation}
  g(n)=(\log(n+2))^b,
  \qquad
  F(n)=n^2+nM(n).
  \label{eq:global-resource-functions}
\end{equation}

\begin{lemma}[Canonical stabilizer resources]
\label{lem:primary-y-resource-interface}
Full encoded zero and plus blocks, together with their complete destructive
readouts, are available with linear quantum space, polylogarithmic time,
and failure probability $e^{-\Omega(n)}$ per complete call.  A
complete primary code canonical-$Y$ batch produces $K(n)$ jointly protected
logical $Y$ rows with
\begin{equation}
  S_Y(n)\le CnM(n),
  \qquad
  T_Y(n)\le Cg(n),
  \qquad
  \mathcal W(\mathcal B_Y;p)\le Ce^{-cn}.
  \label{eq:resource-failure-main}
\end{equation}
If $q_Y$ rows are requested, complete batches are run with
$u=\lceil q_Y/K(n)\rceil$; constant yield gives
\begin{equation}
  uM(n)\le \eta_Y^{-1}q_Y+M(n).
  \label{eq:y-complete-batches}
\end{equation}
All outputs jointly satisfy the corrected boundary condition of
\Cref{def:normalized-boundary}, including in the presence of arbitrary
reference systems.
\end{lemma}

The proof is given in \Cref{thm:primary-y-factory}.  The canonical $Y$ rows
implement logical $S$ gates, whereas the logical $\CCZ$ below is the
transversal operation of the computational code.

\begin{lemma}[Canonical block code switching]
\label{lem:canonical-bcs}
Let $\mathcal Q_1,\mathcal Q_2\in\mathfrak C$, with dimensions $k_1,k_2$ and block
lengths $O(n)$.  A direct canonical code switch maps $k_2$ blocks of
$\mathcal Q_1$ to $k_1$ blocks of $\mathcal Q_2$ by the logical transpose
\begin{equation}
  (b,a)\longmapsto(a,b).
  \label{eq:canonical-transpose-main}
\end{equation}
For $s$ disjoint grids run in parallel, the quantum space is $O(sn^2)$, the
time is $O(\log^2(n+2))$, and the joint probability of the bad supports is
$s\operatorname{poly}(n)e^{-\Omega(n)}$.
The instrument is exact on inputs entangled with arbitrary reference systems outside that
event.  Only the actual dimensions $k_1,k_2$ enter the grid, so no rate
condition is imposed on the auxiliary codes.
\end{lemma}

Section~\ref{sec:switching} proves the lemma by a direct code switch at the
actual code dimensions, including auxiliary codes of low rate.

Write $A_1:\mathbb F_2^{k_{\CCZ}}\to\mathbb F_2^k$ for the matrix of the embedding
$i_1$ in the chosen canonical logical bases, and choose a left inverse $L_1$
with $L_1A_1=I_{k_{\CCZ}}$.  If $\mathcal E_{P_1}^\mathrm{can}$ is the canonical primary code
encoding, define the active encoding by
\begin{equation}
  \mathcal E_{P_1}^\mathrm{act}\ket{x}
  =\mathcal E_{P_1}^\mathrm{can}\ket{A_1x}.
  \label{eq:active-encoding-main}
\end{equation}
The reversible CNOT circuit proved in
\Cref{lem:clean-linear-embedding} implements
\begin{equation}
  \ket{x}\ket0\longmapsto\ket0\ket{A_1x}
  \label{eq:clean-linear-embedding-main}
\end{equation}
and, on the image of $A_1$, its inverse uses $L_1$ to recover $x$ and clear
the encoded bank.  All work registers return to zero, and both directions are
valid with arbitrary references.  For one group, including preparation of
the zero bank and error correction between matching layers, the width is
$O(n^2+M(n)n)$, the time is
$O(k_{\CCZ}+g(n))$, and the failure
probability is $\operatorname{poly}(n)e^{-cn}$.  Short inputs are padded by
independent zeros, which are cleared again by the inverse.

\begin{lemma}[Active logical operations]
\label{lem:addressable-logic}
For $q$ active blocks, mixed zero, plus, and $Y$ resources, logical state transfers
between the active encodings,
logical Paulis, addressable $H$, $S$, $\CNOT$, $\CZ$, and
$\CCZ$, together with transfers between blocks that preserve coordinate labels, are
available with width
$O(qn+F(n))$, time
$O(k_{\CCZ}+g(n))$, and aggregate probability of the bad supports
$\operatorname{poly}(q,n)e^{-cn}$.  All statements use the actual auxiliary
dimensions and hold with arbitrary references.
\end{lemma}

\Cref{prop:specified-logical-gates,lem:mixed-stabilizer-resources,lem:aligned-transfer}
prove these
operations and bounds.  The CNOT circuits implement the linear embeddings,
and \eqref{eq:main-subrank} determines the number of bands.  The logical
$\CCZ$ is the transversal operation of the qLTC family.

\subsection{Compiling one adaptive layer}

All computational data are encoded in full canonical blocks of the primary
code, whose dimension $k=\Theta(n)$ follows from
\eqref{eq:primary-rate-main}.  For each role $\mathsf R$, let
$q_{\mathsf R}$ be the number of dense primary blocks assigned to it.  The
constant number of roles is chosen so that
\begin{equation}
  \sum_{\mathsf R}q_{\mathsf R}=O(W/k+1).
  \label{eq:q-blocks}
\end{equation}
When a bank is transposed, pad it to
\begin{equation}
  q'_{\mathsf R}=k\left\lceil\frac{q_{\mathsf R}}k\right\rceil.
  \label{eq:dense-group-padding}
\end{equation}
Thus fewer than $k$ full zero blocks are added for each role, and
\begin{equation}
  \sum_{\mathsf R}q'_{\mathsf R}=O(W/k+k)
  \label{eq:total-role-padding}
\end{equation}
because the number of roles is constant.  Partition the canonical
coordinates into
\begin{equation}
  B=\left\lceil\frac{k}{k_{\CCZ}}\right\rceil,
  \qquad
  J_s=\{(s-1)k_{\CCZ}+1,\ldots,\min(sk_{\CCZ},k)\}.
  \label{eq:active-bands}
\end{equation}
The final band is padded by independent zero coordinates to size $k_{\CCZ}$.  Every
enable bit on these padding coordinates is set to zero.  Thus
all operations that do not preserve zero are disabled there; padded operands of a
$\CNOT$ or $\CCZ$ are zero, and $S$ leaves zero unchanged.

\subsubsection{Active bands}

Consider a group of $k$ dense blocks and write its logical array as $X_{b,j}$.
For a state entangled with an arbitrary reference system, write
\begin{equation}
  \sum_X \ket{X}_\mathrm{dense}\ket{\phi_X}_R,
  \qquad
  X=(X_{b,j})\in\mathbb F_2^{k\times k}.
  \label{eq:dense-array-state}
\end{equation}
The canonical transpose for the primary code in \Cref{lem:canonical-bcs}
maps this state to
\begin{equation}
  \sum_X \ket{X^T}_\mathrm{dense}\ket{\phi_X}_R.
  \label{eq:dense-array-transpose}
\end{equation}
Retain that orientation while processing all bands.

\begin{lemma}[Processing one coordinate band]
\label{lem:band-exposure}
Let $J_s$ have size $h\le k_{\CCZ}$.  After adjoining $k_{\CCZ}-h$ full zero rows, the
tensor factors in the rows indexed by $J_s$ can be mapped isometrically from
the transposed dense bank to $k$ blocks in the image of
$\mathcal E_{P_1}^\mathrm{act}$, processed by a call whose enable
mask is zero on the padding coordinates, and inserted into the corresponding
destination rows.  When the call transfers logical states between blocks,
the matching is extended by the identity on the padding coordinates.  The operation is the identity on
every unexposed row and is valid with arbitrary reference entanglement.  The
padding coordinates remain independent zero and are cleared before they are
discarded.
\end{lemma}

\begin{proof}
Take the $h$ transposed rows indexed by $J_s$ and append
$k_{\CCZ}-h$ full primary
zero rows.  For every dense block label $b$, these rows define
\begin{equation}
  x_b=(X_{b,j})_{j\in J_s}\oplus0^{k_{\CCZ}-h}.
  \label{eq:short-band-vector}
\end{equation}
Apply the reversible CNOT circuit for the injective map $A_1$ to the block
index, followed by the
canonical transpose for the primary code.  The resulting $k$ active blocks contain
\begin{equation}
  \mathcal E_{P_1}^\mathrm{act}(x_b).
  \label{eq:band-exposure}
\end{equation}
Every unexposed row remains a full primary code block under ordinary error
correction.  Thus no unknown complementary logical coordinate is silently
initialized or discarded.  During the active call, all enable bits on the
$k_{\CCZ}-h$ padding coordinates are zero, and any aligned matching leaves those
coordinates unchanged.

After the active computation, apply logical state transfers so that every
output is in the active $P_1$ encoding, apply the canonical transpose, and
run the embedding circuit backward.  The first $h$ output rows occupy the
corresponding destination slots; the remaining $k_{\CCZ}-h$ rows
are again independent zero and may be discarded.  This is an exact isometry
on the selected tensor factors and the identity on their complement,
including for data entangled with a reference system.  A destination row is replaced only
if it was exposed in the same call or was a known zero row; no unknown
logical row is overwritten.
\end{proof}

\begin{lemma}
\label{lem:dense-band-interface}
For $q\ge1$ dense blocks, put $q'=k\lceil q/k\rceil$.  One band of the
resulting padded bank can be exposed, processed by any logical operation from
\Cref{lem:addressable-logic}, and reinserted using width
$O(q'n+F(n))$ and time
$O(k_{\CCZ}+g(n))$.  Only one band and a constant number of
dense, active, resource, and scratch banks are live at a time.
\end{lemma}

\begin{proof}
The $q'/k$ groups run in parallel.  The canonical transposes contribute
$O(q'n+F(n))$ width and $O(g(n))$ time.  The two
clean circuits for the linear embeddings and the active operation contribute
$O(k_{\CCZ}+g(n))$ time.  All work registers are
either cleared or unconditionally measured and discarded after their
independence or consumption has been established.  No such register remains
live into the next band, and memory error correction is applied to completed
outputs while later calls run.
\end{proof}

\subsubsection{Register permutations}

Dense register permutations are assembled from two operations that preserve
different parts of the indexing by blocks and coordinates.

\begin{lemma}[Routing while preserving coordinate labels]
\label{lem:dense-aligned-transport}
For each $j\in[k]$, let $\pi_j$ be a compiled permutation of $q$ primary
blocks.  The map
\begin{equation}
  (b,j)\longmapsto(\pi_j(b),j)
  \label{eq:dense-aligned-map}
\end{equation}
on the $qk$ logical coordinates can be implemented exactly with width
$O(qn+F(n))$ and time
\begin{equation}
  O\!\left(n+\left(1+\frac n{k_{\CCZ}}\right)g(n)\right).
  \label{eq:dense-aligned-time}
\end{equation}
The construction is valid on inputs entangled with an arbitrary reference
system.
\end{lemma}

\begin{proof}
The source bank is padded to $q'=k\lceil q/k\rceil$ blocks, each $\pi_j$ is
extended by the identity on the added blocks, and the destination consists of
$q'$ full primary code zero blocks.  The canonical transpose is applied to
each group of $k$ source blocks, after which the transposed orientation is
retained.  The zero
logical array of the destination is unchanged by transpose, so no
destination transpose is needed at this boundary.

The $B$ bands are processed serially.  For a band $J_s$,
\Cref{lem:band-exposure} exposes the corresponding source coordinates.  One
active zero target is prepared for each source block, and
\Cref{lem:aligned-transfer} is applied along the restrictions of the
permutations $\pi_j$, $j\in J_s$.  The complete source
readout in this routing procedure measures each exposed source block once.
Logical state transfers produce every output in the active $P_1$ encoding,
and the corresponding rows are inserted into the destination bank.  For a
short final band, the
additional active coordinates and their destination rows remain independent
zero.  After the last band, the measured source rows are discarded and the
full destination bank is transposed back to the dense orientation.  Equation
\eqref{eq:dense-aligned-map} has then been applied once to every real
coordinate.

One band uses only a constant number of padded dense, active, target, work, and
resource banks.  Memory error correction is applied to all other source and
destination blocks.  Hence the peak width is $O(qn+F(n))$.  The $k_{\CCZ}$ matching
rounds inside each logical state transfer and the preparation and switching calls
give
\begin{equation}
  O\!\left(
    B\left[k_{\CCZ}+g(n)\right]+g(n)
  \right).
  \label{eq:dense-aligned-count}
\end{equation}
Since $B=\lceil k/k_{\CCZ}\rceil$, $k=\Theta(n)$, and
$Bk_{\CCZ}<k+k_{\CCZ}\le2k$, this is \eqref{eq:dense-aligned-time}.  Every step is an exact
instrument on the selected tensor factors, so the statement holds with an
arbitrary reference system.
\end{proof}

\begin{lemma}[Permutations within each block]
\label{lem:dense-local-permutation}
For $q\ge1$ dense blocks, put $q'=k\lceil q/k\rceil$.  For every real dense
block $b\in[q]$, let $P_b$ be a compiled permutation of its $k$ canonical
logical coordinates.  The product of the permutations chosen independently
within each block has width $O(qn+F(n))$ and the time bound in
\eqref{eq:dense-aligned-time}.
\end{lemma}

\begin{proof}
The bank is padded to $q'$ blocks, and $P_b$ is extended by the identity on
every added block.  Partitioning the result into $q'/k$ groups of $k$ blocks,
write a block label as $b=(g,t)$, where $g\in[q'/k]$ is the group and
$t\in[k]$ is its position inside the group.  The canonical transpose is
applied to every group.
Under the array transpose, an original coordinate $(g,t,j)$ becomes
$(g,j,t)$.  The desired map becomes
\begin{equation}
  (g,j,t)\longmapsto(g,P_{g,t}(j),t).
  \label{eq:transposed-local-permutation}
\end{equation}
For each coordinate label $t$, the map
$(g,j)\mapsto(g,P_{g,t}(j))$ is the direct sum, over all groups, of
permutations of the block labels.  Hence all groups are handled
simultaneously by one application of \Cref{lem:dense-aligned-transport},
followed by the inverse transpose on each group.  The two
canonical transposes and this single routing call have the stated bounds; in
particular, the additive preparation workspace is not multiplied by the
number of groups.  This construction does not invoke dense register
permutation recursively.
\end{proof}

A prescribed permutation of the $qk$ dense slots defines a bipartite
multigraph whose left vertices are source blocks, whose right vertices are
destination blocks, and whose edges are the routed logical qubits.  Every
vertex has degree $k$.  For a set of source vertices $S$, its $k|S|$ incident
edges meet at most $k|N(S)|$ edges at the neighboring destination vertices,
so Hall's condition gives a perfect matching.  Removing that matching leaves
a regular bipartite multigraph.  Repeating gives a decomposition into $k$
perfect matchings determined during compilation.

\begin{lemma}[Dense register permutation]
\label{lem:register-permutation}
Any compiled permutation of the real and padded logical slots in $q$ dense
primary blocks can be implemented exactly with width $O(qn+F(n))$ and time
\begin{equation}
  O\!\left(n+\left(1+\frac n{k_{\CCZ}}\right)g(n)\right).
  \label{eq:dense-routing-time}
\end{equation}
The statement holds with arbitrary references.  The route graph, its edge
coloring, and all physical masks are compilation data.
\end{lemma}

\begin{proof}
For each source block, first apply
\Cref{lem:dense-local-permutation} so that the color of the route edge
becomes its coordinate label.  The $j$th perfect matching then gives the
permutation $\pi_j$ in \eqref{eq:dense-aligned-map}, and all $k$ matchings are
implemented simultaneously by one call to
\Cref{lem:dense-aligned-transport}.  A second independently chosen
permutation within each block sends each color coordinate to its final
destination coordinate.  All padded slots participate as independent zero
qubits, so the composition is a permutation of the complete padded logical
space.  It consists of two calls that permute coordinates within blocks and
one call that preserves coordinate labels,
not a recursive invocation of the present lemma.

Each constituent call processes the $B$ bands serially.  The $k_{\CCZ}$ masked
matching rounds contribute $Bk_{\CCZ}$, while resource preparation and switching
contribute $Bg(n)$.  Since
\begin{equation}
  Bk_{\CCZ}<k+k_{\CCZ}\le2k=O(n),
  \label{eq:band-linear-work}
\end{equation}
the total is \eqref{eq:dense-routing-time}.  At every time, the live quantum
registers consist of a constant number of padded source and destination banks,
one exposed band, the work blocks for the linear maps, the targets for the current transfer,
and the reused preparation workspace.  Their width is $O(qn+F(n))$.
Memory error correction is applied to blocks outside the current matching.  The
exact array identities and \Cref{lem:aligned-transfer} preserve arbitrary
references, so their composition is the prescribed global permutation.
\end{proof}

\subsubsection{Bounds for one layer}

Each ideal gate is represented by a circuit of constant depth over
$H$, $S$, $\CNOT$, Paulis, and $\CCZ$.  The construction
reserves $W$ logical slots at every ideal layer.  A slot whose wire has been
measured and is no longer live is immediately replaced by an independent
logical zero.  At the beginning of one ideal layer, every slot is assigned to one
of a constant number of roles:
the operand positions of each gate type, each measurement basis,
initialization or reset, and idle storage.  Operands of a multi-qubit gate are
placed at the same block and coordinate indices in their respective banks.
Unused slots are independent logical zeros.  Before a dense register
permutation is invoked, the smaller source or destination layout is padded by
complete zero blocks until the two layouts have the same total capacity.
The number of additional blocks is $O(k)$ because the number of roles is
constant.

The real and dummy slots are first routed into this role layout by
\Cref{lem:register-permutation}; the ideal gate locations determine both the
route and its edge colorings.  For each gate type, the padded operand banks
are transposed once and remain in that orientation while their $B$ bands are
processed.  A call on one band exposes the same coordinate interval from all
operands, transfers the data to the active encodings required by the gate,
applies the stored enable mask, and maps the outputs into the active $P_1$
encoding before reinsertion.  Pauli corrections determined by measurements are
completed before the next non-Clifford operation, and the dense orientation
is restored after the last band.

Measurement banks are likewise exposed band by band.  After the appropriate
Clifford basis rotation and complete destructive readout, a canonical outcome
$m_Z$ contributes the first $h$ entries of
$(B_{\mathcal Q}^{-1}m_Z)_{[k_{\CCZ}]}$ in the logical $Z$ basis; in the logical $X$
basis the retained outcome is the first $h$ entries of
$(B_{\mathcal Q}^{T}m_X)_{[k_{\CCZ}]}$.  Both linear maps have $O(\log n)$ depth with
bounded fan-in.  Complementary and dummy outcomes are discarded.  Measured
bands are not reinserted, while all rows awaiting measurement continue to
receive error correction.  The role assignment places no surviving real wire
in a destructive bank.  Once its final band has been measured, the bank is
discarded and its slots are replaced by fresh zero blocks.

Fresh canonical zero blocks are assigned to initialization roles.  A reset traces
out the old dense block and replaces it by an independently prepared zero
block, which gives the ideal reset channel even for a wire entangled with a
reference system.  Clifford images of zero in the ideal alphabet are prepared
by the same band procedure.  Finally, the surviving and newly initialized
banks are padded by complete zero blocks to the capacity chosen for the next
layer and routed to its layout.

Error correction is applied to every surviving dense or active block while
another band is processed, a resource is prepared, or a classical record is
computed.  The original ideal enable bits are distributed once per ideal
layer in $O(\log W+\log n)$ depth with bounded fan-out.  All masks for that layer
are obtained from precomputed binary maps in parallel and stored classically until
their band is processed.  After every logical measurement outcome is decoded,
the ideal classical update is evaluated exactly once:
\begin{equation}
  \vec x_{t+1}=C_t(\vec x_t,\vec m_t).
  \label{eq:serialized-classical-update}
\end{equation}
No route graph is solved online, and the global fanout cost is not repeated
inside the band loop.  All measurements and resets in the physical schedule
are prescribed unconditionally; only the allowed unitary gates are enabled
by the stored classical bits.

\begin{proposition}[Compilation of one layer]
\label{prop:one-layer}
There is a constant $C_{\mathrm{layer}}>0$, depending only on the code family and
gadget data, such that one ideal adaptive layer is implemented by the schedule
above with
\begin{align}
  S_{\mathrm{layer}}(W,n)&\le C_{\mathrm{layer}}\bigl[W+F(n)\bigr],
  \label{eq:one-layer-width}\\
  T_\mathrm{layer}(W,n,k_{\CCZ})
  &\le C_{\mathrm{layer}}\left[n+\left(1+\frac n{k_{\CCZ}}\right)g(n)
  +\log(W+1)\right].
  \label{eq:one-layer-latency}
\end{align}
Outside the union of the constituent bad support events, the physical
instrument is exactly the ideal adaptive layer and every surviving output
block satisfies the corrected boundary condition, including on arbitrary reference
systems and for arbitrary CPTP values at the faulty locations.
\end{proposition}

\begin{proof}
By \eqref{eq:total-role-padding} and $k=\Theta(n)$, the dense data and all
gate, measurement, initialization, and storage banks use $O(W+n^2)$ qubits.
The reusable workspace consists of one exposed band, the linear embedding and
code switching registers, and the scratch for boundary correction, readout,
and error correction.  The direct switches reuse the same $O(n^2)$ grid.
Finally, \eqref{eq:y-complete-batches} bounds the complete canonical-$Y$
factory bank for the current gate type by $O(q_Yn+nM(n))$.  The bounds in
\Cref{lem:primary-y-resource-interface,lem:canonical-bcs,lem:addressable-logic}
therefore give
\begin{equation}
 C\left[
   \left(\sum_{\mathsf R}q'_{\mathsf R}\right)n+n^2+F(n)
 \right]
 \le C'\bigl[W+F(n)\bigr].
 \label{eq:one-layer-bank-count}
\end{equation}
Here $F(n)\ge n^2$.  Completed resources remain protected in their output
banks, and the preparation workspace is reused only after the current scratch
has been cleared or measured.  Thus serial bands multiply time rather than
peak space, and \eqref{eq:one-layer-bank-count} proves
\eqref{eq:one-layer-width}.

Each band contains a constant number of gadget calls for routing, active
logical operations, and resource preparation.  Including the initial and
final dense routes, the physical time is at
most
\begin{equation}
 C\left[
   n+B\left(k_{\CCZ}+g(n)\right)+g(n)+\log(W+1)
 \right].
 \label{eq:one-layer-call-count}
\end{equation}
The $Bk_{\CCZ}$ term counts the masked matching rounds, and $Bg(n)$ accounts for the
resource calls made for the serial bands.  The last $g(n)$ term includes the
enclosing switches and preparations of complete code blocks.  Since
$B\le1+n/k_{\CCZ}$ and $Bk_{\CCZ}=O(n)$ by
\eqref{eq:band-linear-work}, \eqref{eq:one-layer-call-count} gives
\eqref{eq:one-layer-latency}.  The original enable bits are distributed once
and the ideal classical update is evaluated once, so no $\log W$ computation
is repeated inside the band loop.

For correctness, condition on the good support event for every preparation
batch and block window used in the layer.  At each band boundary, every
surviving logical register occupies its prescribed block and coordinate, all
temporary registers have been cleared or measured, and each live block
satisfies the corrected boundary condition.  The routing lemmas establish
this invariant at the role layout.  Lemma~\ref{lem:band-exposure} and the
lemmas for active operations preserve it while implementing the required logical
instrument: an output replaces only the row exposed in the same call or a
known zero row, so no unknown logical data are overwritten.  Measurement
outcomes are decoded and copied to the classical record before their banks
are discarded, and the update in \eqref{eq:serialized-classical-update} is
then evaluated once.  Sequential and parallel composition in
\Cref{prop:gadget-composition} preserves the invariant through the final
route, including for inputs entangled with arbitrary reference systems.
\end{proof}

\subsection{Main theorem}

\begin{theorem}[Main result]
\label{thm:main-result}
For the qLTC family of \Cref{thm:qltc-family}, there are
constants $p_*,A,C_\mathrm{add},C_{\mathrm{layer}},C_F,C_S,C_T,C_B,b_B,c_B>0$,
depending only on
the family and the compilation data, such that the following holds.

Let $C$ be an adaptive circuit over the chosen local
Clifford+$\CCZ$ gate set, with classical input and output, width $W\ge1$,
depth $D\ge1$, circuit size
$N=WD$, and
target error $0<\varepsilon\le1/2$.  Put
\begin{equation}
  L=\log(WD/\varepsilon)
  \label{eq:L-def}
\end{equation}
Fix $d_0$ so that the properties above hold for every $d\ge d_0$, and let
$d=d(L)$ be the least $d\ge d_0$ for which $n_d\ge A L$.
Set $n=n_d$ and $k_{\CCZ}=k_{\CCZ}(\tau_d)$.  For every
locally stochastic physical noise rate $p<p_*$, the compiled physical
schedule obeys
\begin{align}
  S_\mathrm{FT}(W,n)&\le C_\mathrm{add}\bigl[W+F(n)\bigr],
  \label{eq:main-preabsorption-space}\\
  D_\mathrm{FT}(W,D,n,k_{\CCZ})&\le C_{\mathrm{layer}}(D+1)\left[
    n+\left(1+\frac{n}{k_{\CCZ}}\right)g(n)
    +\log(W+1)\right].
  \label{eq:main-preabsorption-time}
\end{align}
and its bad event satisfies
\begin{equation}
  \Pr[\mathsf{Bad}]
  \le C_B(D+1)[W+F(n)](n+2)^{b_B}e^{-c_Bn}.
  \label{eq:main-theorem-bad}
\end{equation}
Outside this event the complete physical instrument has the ideal adaptive
logical action, including on arbitrary reference systems.  If in addition
\begin{equation}
  W\ge C_F\,F(n),
  \label{eq:main-width-condition}
\end{equation}
then the schedule is a fault-tolerant simulation in the circuit model of
\Cref{def:adaptive-circuit}, and
\begin{equation}
  S_\mathrm{FT}\le C_S W,\qquad
  D_\mathrm{FT}\le C_T D\log(WD/\varepsilon),\qquad
  \|P_\mathrm{FT}-P_\mathrm{ideal}\|_\mathrm{TV}\le\varepsilon.
  \label{eq:main-final-bounds}
\end{equation}
Faulty locations may apply arbitrary supported CPTP channels.  Noiseless
classical hardware of polynomial size in the regime
\eqref{eq:main-width-condition} is not charged to quantum width, while all
of its running time is charged to physical time.
\end{theorem}

\begin{proof}
Apply \Cref{prop:one-layer} to each ideal layer and include preparation and
readout costs of the same order.
Ancillary banks are reused, so
space is not multiplied by $D$.  This gives
\eqref{eq:main-preabsorption-space}--\eqref{eq:main-preabsorption-time}.

The polynomial subrank bound gives
\begin{equation}
  \frac{(1+n/k_{\CCZ})g(n)}{n}
  \le \frac{g(n)}n+c_\sigma^{-1}\frac{g(n)}{n^\sigma}
  \longrightarrow0.
  \label{eq:subrank-time-absorption}
\end{equation}
Hence the physical time is $O((D+1)[n+\log(W+1)])$.

Choose the code, decoder, preparation, switching, and radii that control bounded spread
before $n,W,D,\varepsilon$, and then choose $p_*$ below the corresponding
finite collection of thresholds.  Let $\mathfrak Y$ index complete
  canonical-$Y$ batches, $\mathfrak S$ complete code switching grids, and
  $\mathfrak W$ the direct preparations, readouts, and other elementary
gate or memory windows.  For each call, translate its bad support family to
the global set of physical locations.  Let $\mathsf{Bad}_Y$,
$\mathsf{Bad}_S$, and $\mathsf{Bad}_\omega$ denote the events that the actual
fault support contains a member of the corresponding translated family.
Define
\begin{equation}
  \mathsf{Bad}
  =\bigcup_{Y\in\mathfrak Y}\mathsf{Bad}_{Y}
   \;\cup\!\bigcup_{S\in\mathfrak S}\mathsf{Bad}_{S}
   \;\cup\!
   \bigcup_{\omega\in\mathfrak W}\mathsf{Bad}_{\omega}.
  \label{eq:global-bad-union}
\end{equation}
The family associated with each batch or switch already includes the
witnesses for its internal
raw slots, tests, memory intervals, and feedback.  Defects at lower levels are
therefore not counted again at the top level.  The number of such calls
is bounded by the physical spacetime volume, which gives
\begin{equation}
  \Pr[\mathsf{Bad}]
  \le C_B(D+1)[W+F(n)](n+2)^{b_B}e^{-c_Bn}.
  \label{eq:global-bad-event}
\end{equation}
This is a union bound on the original fault supports.  It uses neither independence
between windows nor a probabilistic interpretation of the coefficients in the
Pauli expansion.

Outside $\mathsf{Bad}$, induction through the gadget identities for arbitrary
reference systems gives exactly the ideal adaptive logical channel.  Each decoded
measurement record determines the ideal next classical update, and the final
destructive readout produces the ideal classical output.  The
total variation error is therefore at most the probability of the bad event.

Let $n_\mathrm{first}$ exceed all cutoffs imposed above and increase $A$ so that
$A\log 2\ge n_\mathrm{first}$.  By \eqref{eq:dense-lengths-main}, the first
usable $n\ge A L$ obeys
\begin{equation}
  A L\le n\le (Q^2+1)A L.
  \label{eq:usable-length-bound}
\end{equation}
Assume now \eqref{eq:main-width-condition}.  The additive term is absorbed
into $W$, so \eqref{eq:global-bad-event} is at most
$C'_BWD(n+2)^{b_B}e^{-c_Bn}$.  Increase the minimum length so that
$(n+2)^{b_B}\le e^{c_Bn/2}$, and choose $A$ with
$c_BA\ge4$ and
$C'_Be^{-c_BA\log2/4}\le1$.  Since
$WD=\varepsilon e^L$ and $n\ge A L$, these choices give
\begin{equation}
  \Pr[\mathsf{Bad}]\le\varepsilon.
  \label{eq:failure-absorption}
\end{equation}
Since $D\ge1$ and
$\log(W+1)=O(L)$, \eqref{eq:usable-length-bound} and
\eqref{eq:subrank-time-absorption} give the depth bound in
\eqref{eq:main-final-bounds}.  The same width condition bounds $n$, $M(n)$,
all stored records, and the compiled classical circuits by a polynomial in
$N=WD$.  It also absorbs $F(n)$ in the quantum space bound and completes
\eqref{eq:main-final-bounds}.
\end{proof}

\begin{corollary}
\label{cor:polynomial-range}
Let $c,c'>0$.  If $D\le W^c$ and $\varepsilon\ge W^{-c'}$, then
\eqref{eq:main-width-condition} holds for all sufficiently large $W$.
\end{corollary}

\begin{proof}
The hypotheses give
$\log W\le L\le(1+c+c')\log W$, so
$n=\Theta(\log W)$.  Equations~\eqref{eq:y-constant-yield} and
\eqref{eq:global-resource-functions} give
$\log F(n)=O(\log n\log\log n)=o(\log W)$, and hence
$F(n)=W^{o(1)}$.
\end{proof}

%% file: sections/state-preparation.tex
\section{State preparation gadgets}
\label{sec:state-preparation}

The dense compiler uses direct preparations of encoded zero and plus states,
complete logical readout, and canonical logical $Y$ states for phase gates.
The first two resources follow from a measured Clifford compiler whose
leading correction maps an arbitrary input into the required error span.
The $Y$ resource requires an
additional factory because an error in a raw encoded row may change all of
its logical coordinates.  The construction below treats such a row as one
bad input position, applies a sequence of Hamming corrections of increasing
length, and keeps the whole batch protected while the classical records are
evaluated.

\subsection{A measured Clifford compiler}

A CSS code $\mathcal Q_0$ of constant size is used throughout the compiler.  Puncturing at one coordinate
a binary evaluation code of length $512$ and degree at most four gives a
$[[511,1,d_0]]$ CSS code with $d_0\ge31$.  The code has $255$ independent pure
$X$ checks and $255$ independent pure $Z$ checks, and the all ones word may
be used for both logical Pauli representatives.  All circuits associated
with $\mathcal Q_0$ are chosen once and are independent of the later qLTC block
length.

For each complete syndrome $s$ of $\mathcal Q_0$, choose a Pauli of minimum weight
$D(s)$ with that syndrome, breaking ties lexicographically.  The resulting rule is a
total syndrome section: $D(s)$ is applied also on records outside the
decoding promise for errors of bounded weight.  Since $\mathcal Q_0$ is a constant, the table and the
Boolean circuit evaluating it have constant size and depth.

A check $S$ of weight $w$ is measured with a length-$w$ cat state whose
adjacent $Z_iZ_{i+1}$ parities are verified by fresh ancillas.  The cat--data
couplings are enabled only when all verification outcomes vanish, whereas the
cat readouts and resets are unconditional.  Each cat wire couples to one data
coordinate.  A fault after successful verification can therefore introduce a
Pauli on at most one data coordinate, while a fault that changes a
verification or readout record only relabels the measured projector.  The
recordwise statement extends to arbitrary operators at the faulty location by
Pauli expansion.

Four complete rounds of the independent checks are performed, and the
lexicographically first equal pair of valid syndrome words is used.  With at
most one fault among the four rounds, this pair exists and contains a clean round.
For every nonzero branch, the data at the end of the measurement therefore
has the form
\begin{equation}
  P\Pi_s A_\omega,
  \qquad \operatorname{wt}(P)\le1,
  \label{eq:four-round-sector}
\end{equation}
where $\Pi_s$ projects onto the complete syndrome sector and $A_\omega$ is a
linear map on the arbitrary input and its reference system.  Applying the
physical Pauli $D(s)$ maps that entire sector into the codespace and does not
spread $P$.

\begin{lemma}[Finite measured Clifford compiler]
\label{lem:fixed-clifford-compiler}
There is a finite collection of measured Clifford procedures on $\mathcal Q_0$ with
the following properties.

If the input belongs to the error span generated by recursively sparse
errors with at most four exceptional children at each concatenation level,
then a compiled Clifford rectangle with at most one bad immediate child
implements the intended logical operation and maps the input error span into
itself.  A destructive measurement implements the complete ideal logical
instrument.  On an arbitrary input, the leading correction maps every
nonzero branch into that error span around some encoded state; no claim about
the original logical state is made outside the correctable input domain.
Both conclusions hold with arbitrary reference systems and arbitrary
supported CPTP fault values.
\end{lemma}

\begin{proof}
The calculation for the verified cat states preceding the lemma gives a common form for
every branch of a complete correction procedure.  For each possible faulty
base location there is a single physical coordinate $q$, independent of the
fault value and measurement record, such that
\begin{equation}
  K_\omega=\sum_{P\in\{I,X_q,Y_q,Z_q\}}P\mathcal E_0L_{\omega,P},
  \label{eq:arbitrary-input-projection}
\end{equation}
where $\mathcal E_0$ is the $\mathcal Q_0$ encoder.  On a promised input containing a Pauli of
weight at most four, distance $31$ gives the stronger correction identity
when there is one fault.  The same recovery rule is used for all Pauli terms, records, and
reference states.

Let $\mathcal E_h$ be the $h$-fold concatenated encoder.  A recursively sparse error
at level $h$ has at most four exceptional level-$(h-1)$ blocks and is sparse
on every other child.  Products of two such errors have at most eight
exceptional children.  Compressing through $\mathcal E_h$ from the leaves maps each
ordinary child to a scalar and each exceptional child to a logical Pauli.
The outer Pauli consequently has weight at most eight, still below $d_0$.
The Knill--Laflamme condition thus holds for the whole recursively sparse
span, including superpositions and reference entanglement.

Every compiled rectangle begins with a complete child correction on each
input block before any multiblock coupling.  Its quantum locations, the
classical computation of the syndrome section, and all waits form a disjoint
partition into child rectangles.  A location at level zero is bad when it
is faulty, and a higher rectangle is bad when at least two of its immediate
children are bad.  If there is a unique bad child, its logical action is
propagated backward to that child's input location after exposing the child
logical--syndrome factorization.  The good children implement the desired
outer operation, and the bad child contributes at most two output blocks for
which no recursive sparsity guarantee is available.  Together with the single coordinate in
\eqref{eq:arbitrary-input-projection}, the deviation at the next level has at most
three exceptional children, which closes the induction without an
increasing radius.

The same argument applies to preparation and destructive readout.  On a
promised input it preserves the intended logical factor, whereas on an
arbitrary input the leading correction first produces a code state with a
recursively sparse deviation.  All equalities are equalities of complete
branch maps before normalization.  Summing over the actual records therefore
gives a trace-preserving instrument rather than a collection of normalized
postselected states.
\end{proof}

Let $A_0$ bound the number of immediate child rectangles in any member of
the finite compiler.  If $w_h^\mathrm{rec}(q)$ denotes the weight enumerator on the
original fault supports of bad rectangles at level $h$, when a bad location
at level zero has
weight at most $q$, then
\begin{equation}
  w_0^\mathrm{rec}(q)=q,
  \qquad
  w_h^\mathrm{rec}(q)
  \le \binom{A_0}{2}\bigl(w_{h-1}^\mathrm{rec}(q)\bigr)^2
  \le(C_\mathrm{rec}q)^{2^h}
  \label{eq:fixed-compiler-recursion}
\end{equation}
for a constant $C_\mathrm{rec}$.  The product in this recurrence comes from disjoint
original spacetime domains of distinct children, not from independence of
their output errors.

Coherent unencoding is obtained by normalizing every child, applying the
inverse encoded Clifford circuit, retaining its designated output, and
discarding the other outputs at prescribed locations.  For this peeling
step all immediate children are required to be good.  If $A_U$ bounds their
number, then the weight of the exceptional family through $h$ levels is at most
\begin{equation}
  w_h^\mathrm{unenc}(p)
  \le A_Up+A_U\sum_{j=1}^{h-1}(C_\mathrm{rec}p)^{2^j}
  \le C_Up
  \label{eq:coherent-unencoding-bound}
\end{equation}
whenever $C_\mathrm{rec}p\le1/2$.  Distinct physical output unencoders have disjoint
original locations, so the witness weight for any prescribed set of $t$
bad outputs is at most $(C_Up)^t$.

\subsection{Canonical encoding and raw rows}

Consider a CSS code with check matrices $H_X,H_Z$.  Choose independent row
bases and disjoint pivot sets.  After row operations and a physical
coordinate order $(L,Z,X)$, write
\begin{equation}
  H_X^\mathrm{can}=[C_X\ I_{r_X}\ D_X],
  \qquad
  H_Z^\mathrm{can}=[C_Z\ D_Z\ I_{r_Z}],
  \label{eq:canonical-checks}
\end{equation}
and choose canonical logical representatives
\begin{equation}
  L_Z=[I_k\ C_X^T\ 0],
  \qquad
  L_X=[I_k\ 0\ C_Z^T].
  \label{eq:canonical-logicals}
\end{equation}

\begin{lemma}[Canonical CSS encoding]
\label{lem:canonical-encoding}
The pivot sets in \eqref{eq:canonical-checks} can be chosen disjoint, and the
rows in \eqref{eq:canonical-logicals} commute with all checks and have the
standard symplectic pairing.  Starting from
$\ket\eta_L\ket0_Z\ket+_X$, measuring the checks of bounded weight, computing
their independent syndromes, and applying Pauli corrections on the pivot
zones gives an encoding instrument.  Every branch is a scalar multiple of
the same canonical encoding isometry $\mathcal E\ket\eta$, including when
$\ket\eta$ is entangled with a reference.  The quantum depth is constant,
the classical depth is $O(\log n)$, and the width is $O(n)$.
\end{lemma}

\begin{proof}
Choose pivot columns for $H_Z$.  Projection of the $H_X$ row space onto their
complement is injective: an $H_X$ row supported only on the $H_Z$ pivots is
orthogonal to every $H_Z$ row, and invertibility of the pivot submatrix
forces that row to vanish.  The $H_X$ pivots can therefore be chosen in the
complement.  Row reduction gives \eqref{eq:canonical-checks}, and
$H_XH_Z^T=0$ verifies \eqref{eq:canonical-logicals} directly.

Initialize the $L$ coordinates in $\ket\eta$, the $Z$ pivots in $\ket0$,
and the $X$ pivots in $\ket+$.  Binary maps determined by the encoder convert the
check outcomes into the independent syndromes.  Corrections on the pivot
zones set all check eigenvalues to $+1$ and commute with the canonical
logical action.  Each branch therefore intertwines the complete logical
Pauli algebra in the same way and differs from one encoding isometry only by
a scalar independent of $\ket\eta$.  Bounded check incidence gives constant
quantum depth, and the syndrome maps have $O(\log n)$ depth with bounded
fan-in and bounded fan-out.
\end{proof}

The ordinary canonical encoder is implemented with the concatenated compiler
of \Cref{lem:fixed-clifford-compiler} at level $h(n)=O(\log\log n)$.
After coherent unencoding, each raw primary code row $u$ is assigned a family
$\mathcal B_u$ of bad fault supports such that
\begin{equation}
  \mathcal W(\mathcal B_u;p)\le p_\mathrm{raw}
  \label{eq:raw-y-row-bound}
\end{equation}
for a chosen constant $p_\mathrm{raw}>0$.  If the fault support contains no
member of $\mathcal B_u$, the row lies in a
neighborhood of the full canonical logical state with small errors
$\ket{+i}^{\otimes k}$.  The family refers to the original locations of
that slot.  Thus for every prescribed set $J$ of raw slots,
\begin{equation}
  \mathcal W\!\left(\mathop{\circledast}_{u\in J}\mathcal B_u;p\right)
  \le p_\mathrm{raw}^{|J|}.
  \label{eq:raw-y-product-bound}
\end{equation}
Again, this is a product of witness sums on disjoint spacetime sets and does
not assert independence of the output states.

Each raw row is tested by measuring the checks of bounded weight of the primary
qLTC.  An observed syndrome of low weight certifies a small error around some code
state, and completeness accepts rows near the desired canonical $Y$ state.
Every candidate is accompanied by an encoded zero spare satisfying the
corrected boundary condition.  Once the acceptance bit has been computed and
copied, a complete conditional physical $\SWAP$ selects the candidate on
acceptance and the spare on rejection.  No correction is inserted inside the
decomposition of this swap into three CNOTs.  Terminal error correction returns the
selected row to the corrected boundary condition.  A row selected from the
zero spare is subsequently treated as a possibly bad logical input to the
Hamming circuit, not as a $Y$ output.

\subsection{Protection during classical processing}

The acceptance calculation takes polylogarithmic time.  Memory error
correction protects the candidate and its zero spare throughout this
calculation, relative to one stored syndrome baseline as in
\Cref{lem:decoder-interface}.

\begin{lemma}[Storage relative to a syndrome sector]
\label{lem:sector-storage}
Measure one syndrome baseline and retain it throughout a sequence of memory
cycles.  Let $v$ be its offset from the ideal syndrome, $t_a$ the reduced
error weight after cycle $a$, and $\mathbf f_a$ the fresh faults in that
cycle.  The cycles can be chosen with constants
$0<\theta<1$ and $C_\mathrm{sec},C_\mathrm{mem}>0$ such that
\begin{equation}
  t_{a+1}\le\theta t_a+C_\mathrm{sec}|v|+C_\mathrm{mem}|\supp\mathbf f_a|
  \label{eq:sector-storage}
\end{equation}
whenever the input and syndrome errors satisfy the weight condition in
\Cref{lem:decoder-interface} at each cycle.  The conclusion is uniform in the
underlying syndrome sector.
\end{lemma}

\begin{proof}
Subtract the same measured baseline from every later syndrome record before
applying the contracting decoder.  The ideal sector cancels, and the decoder
input is the syndrome of the current error with the same baseline offset $v$
at every cycle.
The decoder contraction and the syndrome extraction circuit, whose spread is bounded, give
\eqref{eq:sector-storage}.  Iteration yields
\begin{equation}
  t_a\le\theta^at_0+
  \frac{C_\mathrm{sec}|v|+C_\mathrm{mem}\max_{j<a}|\supp\mathbf f_j|}{1-\theta}.
  \label{eq:sector-storage-sum}
\end{equation}
Thus the radius is independent of the number of classical clock steps.  The
baseline is physically measured once and stored as a classical record; no
unknown syndrome is filled or silently set to zero.
\end{proof}

The tester, its acceptance computation, the encoded zero spare, the complete
conditional swap, memory until the tested row is used, and terminal error
correction are assigned to one bad support family.  Below a constant
physical threshold, the weight of the union of these families over all $M(n)$ slots
and their polylogarithmic lifetimes is at most
\begin{equation}
  CM(n)g(n)e^{-cn}.
  \label{eq:tested-slot-catastrophe}
\end{equation}
If the fault support contains no member of this family, every retained slot
satisfies the boundary condition
\eqref{eq:typed-boundary-map}, irrespective of its acceptance bit.

\subsection{Hamming purification of \texorpdfstring{$Y$}{Y} states}

For $h\ge4$, let $H_h$ be the $h\times(2^h-1)$ matrix whose columns are the
nonzero vectors of $\mathbb F_2^h$.  Put
\begin{equation}
  m=2^h-1,
  \qquad
  k_h=m-2h.
  \label{eq:hamming-parameters}
\end{equation}
The row space of $H_h$ is self-orthogonal.  On
$\ker H_h/\operatorname{row}H_h$ the dot product is nondegenerate and
nonalternating, and therefore has an orthonormal basis.  Let $G_h$ be a lift
of that basis, so
\begin{equation}
  H_hG_h^T=0,
  \qquad
  G_hG_h^T=I_{k_h}.
  \label{eq:hamming-logical-basis}
\end{equation}
Completing $G_h$ by check rows and dual partners gives an encoder consisting only of CNOTs
$U_h$ of size $O(m^2)$.

\begin{lemma}[One Hamming $Y$ step]
\label{lem:hamming-y-step}
The circuit applies $U_h^{-1}$ to $m$ input wires and measures the check wires
in the prescribed $X$ and $Z$ bases.  The sum of the two syndrome words
determines a Hamming $Z$ correction, followed by the known output sign
corrections $(-1)^{(|g|-1)/2}$ for the rows $g$ of $G_h$.  If at most one
input wire is in an arbitrary state and every other input is $\ket{+i}$, then
all $k_h$ retained outputs are exactly $\ket{+i}^{\otimes k_h}$, jointly and
with an arbitrary reference system.
\end{lemma}

\begin{proof}
Expand the arbitrary input operator in Paulis.  Relative to the $Y$
eigenbasis, one bad input produces a word of eigenvalue flips of weight at most
one.  The Hamming syndrome identifies that coordinate and the $Z$ correction
removes the flip.  The logical operators $Y^g$ for the rows $g$ of $G_h$
commute with the check measurement, and \eqref{eq:hamming-logical-basis}
shows that $U_h^{-1}$ maps them to the logical $Y$ operators on individual wires.
The displayed sign corrections select their common $+1$ eigenspace.  That
eigenspace is the pure product $\ket{+i}^{\otimes k_h}$ and therefore factors from
the reference on every nonzero record.
\end{proof}

Each wire of this circuit is encoded as a complete primary code row, and the
CNOT circuit acts transversally on all $k$ canonical logical coordinates.  A bad
row may contain arbitrary correlated errors on all $k$ coordinates, but it
still occupies only one Hamming input position.  Applying the Pauli expansion
from the proof of \Cref{lem:hamming-y-step} coordinatewise therefore gives one joint output identity;
there is no union bound over logical coordinates.

The raw slots are indexed by
$[u_{\ell}]\times\cdots\times[u_1]$, using the lengths
$u_j=2^{h_j}-1$ from \eqref{eq:global-batch-functions}.
A level-$j$ macroblock consists
of $u_j$ disjoint children.  For each output label of a child, one row from every
child enters one Hamming circuit.  Hence each bad child contributes at most
one bad input row to each circuit, and \Cref{lem:hamming-y-step} corrects all output
labels jointly.

If $w_j^Y(z)$ is the witness weight that a level-$j$ macroblock is bad when a
raw slot has weight at most $z$, then
\begin{equation}
  w_0^Y(z)=z,
  \qquad
  w_j^Y(z)\le \binom{u_j}{2}\bigl(w_{j-1}^Y(z)\bigr)^2
  \le(2^{12}z)^{2^j}.
  \label{eq:y-failure-recursion}
\end{equation}
The last estimate follows from
$\log_2u_j\le4+2\log_2j$ and the convergence of the associated weighted
sum.  Taking $p_\mathrm{raw}\le2^{-13}$ and
$\ell(n)=\lceil\log_2(\lambda_Yn)\rceil$ makes the failure probability for a logical row
exponentially small in $n$.

The same sequence has constant yield.  The elementary bounds
$7j^2\le u_j\le16j^2$ give
\begin{equation}
  \log M(n)=\Theta(\ell(n)\log\ell(n)).
  \label{eq:y-product-growth}
\end{equation}
Moreover $x_j=2h_j/u_j$ is at most $1/2$ for $j\ge2$ and
$\sum_jx_j<\infty$.  Consequently
\begin{equation}
  \frac{K(n)}{M(n)}
  =\prod_{j=1}^{\ell(n)}(1-x_j)
  \ge\frac7{15}\exp\!\left(-2\sum_{j=2}^{\infty}x_j\right)
  =:\eta_Y>0,
  \label{eq:y-yield-proof}
\end{equation}
which proves \eqref{eq:y-constant-yield} without a numerical estimate of the
infinite product.

\subsection{The canonical \texorpdfstring{$Y$}{Y} factory}

Let $c(n)$ and $t(n)$ be integer polylogarithmic upper bounds for the width
and time of one raw row preparation through testing and zero replacement.
Increase the minimum usable length so that $M(n)\ge2c(n)$ and set
\begin{equation}
  P_\mathrm{raw}=\left\lfloor\frac{M(n)}{c(n)}\right\rfloor,
  \qquad
  V_\mathrm{raw}=\left\lceil\frac{M(n)}{P_\mathrm{raw}}\right\rceil
  \le2c(n)+1.
  \label{eq:raw-wave-count}
\end{equation}
The $M(n)$ raw slots are assigned before execution and produced in
$V_\mathrm{raw}$ complete parallel rounds.  The locations assigned to unused
preparations in the last round are idle.  All zero spares are prepared before
production, and memory error correction is applied to completed rows until
they are used as inputs to the Hamming circuits.

At level $j$, every disjoint Hamming group runs the same inverse encoder in
parallel.  A sequential implementation within one group has $O(u_j^2)$
protected operations, and
\begin{equation}
  \sum_{j=1}^{\ell(n)}u_j^2=O(\ell(n)^5).
  \label{eq:hamming-total-work}
\end{equation}
Charging $O(\log n)$ time for each operation, including readout, binary
feedback, boundary correction, and memory on waiting rows, gives
\begin{equation}
  T_Y(n)\le C\bigl[(2c(n)+1)t(n)+\log^6(n+2)+\log n\bigr]
  \le Cg(n).
  \label{eq:y-factory-latency}
\end{equation}
The peak quantum space is $O(nM(n))$: processors are reused between rounds,
while the raw slots, spares, and completed outputs occupy only a constant
number of full banks.

\begin{theorem}[Preparation of canonical $Y$ states]
\label{thm:primary-y-factory}
Below a constant physical noise threshold, one complete batch produces
$K(n)$ full primary code rows whose logical content is
$\ket{+i}^{\otimes k}$.  Its space, time, and bad support family
$\mathcal B_Y$ satisfy
\begin{equation}
  S_Y(n)\le CnM(n),
  \qquad
  T_Y(n)\le Cg(n),
  \qquad
  \mathcal W(\mathcal B_Y;p)\le Ce^{-cn}.
  \label{eq:y-factory-theorem}
\end{equation}
For every fault support containing no member of $\mathcal B_Y$, all $K(n)$
outputs are jointly exact at the
corrected boundary of \Cref{def:normalized-boundary}, including with
arbitrary reference systems and arbitrary supported CPTP faults.  For $q_Y$
requested rows, let $\mathcal B_Y(q_Y)$ be the union of the translated bad
support families of the complete batches run on disjoint locations.  Then
\begin{align}
  S_Y(q_Y,n)&\le C(q_Yn+nM(n)),
  &T_Y(q_Y,n)&\le Cg(n),
  \label{eq:y-request-service}\\
  \mathcal W(\mathcal B_Y(q_Y);p)
  &\le C\left\lceil\frac{q_Y}{K(n)}\right\rceil e^{-cn}.
  \label{eq:y-request-failure}
\end{align}
\end{theorem}

\begin{proof}
The part of the bad support family concerning logical rows is
\eqref{eq:y-failure-recursion}.  The remaining part includes the physical
locations for raw preparation, unencoding, testing and replacement, the
Hamming circuits, intervening storage and classical waits, and terminal
correction.  By
\eqref{eq:tested-slot-catastrophe} its weight is at most a
quasipolynomial factor times $e^{-cn}$; since
$\log M(n)=o(n)$, this is $e^{-\Omega(n)}$ after increasing the minimum
length.  Avoiding the union of these two families gives the joint exactness
of all outputs by \Cref{lem:hamming-y-step} and the recordwise composition
rule of \Cref{prop:gadget-composition}.

For $q_Y$ requested rows, take
$u=\lceil q_Y/K(n)\rceil$.  Equation~\eqref{eq:y-yield-proof} gives
$uM(n)\le\eta_Y^{-1}q_Y+M(n)$.  Run the $u$ complete batches on disjoint
locations in parallel, retain predetermined output labels, and discard the
surplus unconditionally.  Summing the witness weights of the complete batches gives
the stated failure and resource bounds without an independence assumption.
\end{proof}

Full zero and plus rows are prepared directly by the canonical encoder and
measured Clifford compiler.  Their complete preparation has linear output
  space, reusable polylogarithmic workspace, polylogarithmic time, and a bad
  support family of weight $e^{-\Omega(n)}$.  Bell states for the secondary
construction are obtained by applying two canonical encoders to physical
Bell pairs.  The zero, plus, and Bell preparations apply to the primal codes and their
physical-$\mathsf H$ dual codes; positive rate is used only for the primary
code $Y$ factory.

\begin{corollary}[Encoded stabilizer resources]
\label{cor:encoded-stabilizer-resources}
Canonical encoded zero and plus blocks, encoded Bell states, and canonical
primary code $Y$ rows can be prepared with the corrected boundary guarantees
above, including for arbitrary reference systems.  The $Y$ demand satisfies
\eqref{eq:y-request-service}; the other resources have their data volume plus
polylogarithmic reusable workspace.
\end{corollary}

\subsection{Destructive logical measurements}

\begin{lemma}[Destructive logical readout]
\label{lem:destructive-logical-readout}
For any $\mathcal Q\in\mathfrak C$, physical $Z$ measurement followed by the full
$O(\log n)$ decoder returns the complete canonical logical $Z$ word;
physical $X$ measurement gives the analogous logical $X$ word.  If the
input obeys the corrected boundary condition and the readout window has at most
the allowed fraction of faults, the decoded outcome is the ideal logical
outcome.  The instrument identity retains the complete record and remains
valid for an arbitrary reference system.
\end{lemma}

\begin{proof}
For physical $Z$ measurement, write the observed word as
\begin{equation}
  v=v_0+e,
  \qquad H_Zv_0=0,
  \label{eq:readout-word}
\end{equation}
where $v_0$ is the ideal codeword and $e$ is the combined data and readout
error.  Under the corrected boundary promise, the exact decoder returns
$f=\mathcal D(H_Zv)$ such that $e+f$ is an $X$-check word.  Therefore
\begin{equation}
  L_Z(v+f)=L_Zv_0,
  \label{eq:logical-Z-readout}
\end{equation}
which is the complete ideal logical $Z$ word.  Physical $X$ measurement is
the same argument on the dual code.

The decoder and the logical readout map have $O(\log n)$ classical depth.
Memory correction is applied to the other live blocks while these records are
computed.  Conditioned on a physical fault support and a complete measurement
record, the same decoded word is used on every surviving left and right
Pauli term.  The recordwise Pauli reduction of
\Cref{lem:recordwise-Pauli} then proves the statement for arbitrary CPTP
faults and arbitrary reference systems.
\end{proof}

%% file: sections/addressable-logic.tex
\section{Addressable logical gates from transversal \texorpdfstring{$\CCZ$}{CCZ}}
\label{sec:addressable-logic}

The transversal physical $\CCZ$ acts first on the complete logical spaces of
the three primal codes.  To use that action inside an arbitrary circuit, the
active logical coordinates supplied by the trilinear form must be embedded
and addressed without treating a change of logical basis as a free operation.
Direct block code switching and reversible CNOT circuits implement the
required linear maps, while logical state transfers connect the active
encodings.  The code switch
uses only the code properties, decoder, and Clifford operations; the trilinear form enters when
the physical transversal gate is restricted to the active coordinates.

\subsection{Active logical encodings}

Choose the canonical encoders and logical representatives of
\Cref{lem:canonical-encoding}, and write
$\mathcal E_{P_j}^\mathrm{can}$ and $\mathcal E_{D_j}^\mathrm{can}$ for the
encoders of the primal and dual codes.  Only the primary code used for dense storage
has dimension $\Theta(n)$; for every auxiliary code we retain
its actual dimension and use only $k_{\mathrm{CCZ}}\le k_j\le n_j$.

For each primal code $P_j$, let
$i_j:\mathbb F_2^{k_{\mathrm{CCZ}}}\to H_j$ be the injection in
\eqref{eq:main-subrank}, and write $A_j$
for its matrix in the canonical logical basis, and extend $A_j$ to an
invertible binary matrix $B_j$ on the complete logical space.  If
$U_B\ket{x}=\ket{Bx}$, define the active isometries
\begin{equation}
  \mathcal E_{P_j}^\mathrm{act}(\psi)
  =\mathcal E_{P_j}^\mathrm{can}U_{B_j}
    \bigl(\psi\otimes\ket0^{\otimes(k_j-k_{\mathrm{CCZ}})}\bigr),
  \label{eq:primal-active-encoding}
\end{equation}
so that
$\mathcal E_{P_j}^\mathrm{act}\ket{x}
=\mathcal E_{P_j}^\mathrm{can}\ket{A_jx}$ on computational basis
vectors.  Write $H_\mathrm{phys}=H^{\otimes n_j}$ and
$H_\mathrm{act}=H^{\otimes k_{\mathrm{CCZ}}}$.  For the physical-$H$ dual
$D_j$, set
\begin{equation}
  \mathcal E_{D_j}^\mathrm{act}(\psi)
  =H_\mathrm{phys}\mathcal E_{P_j}^\mathrm{act}(H_\mathrm{act}\psi)
  =\mathcal E_{D_j}^\mathrm{can}U_{B_j^{-T}}
    \bigl(\psi\otimes\ket+^{\otimes(k_j-k_{\mathrm{CCZ}})}\bigr).
  \label{eq:dual-active-encoding}
\end{equation}
Thus both cases have the form
\begin{equation}
  \mathcal E_{\mathcal Q}^\mathrm{act}(\psi)
  =\mathcal E_{\mathcal Q}^\mathrm{can}U_{B_{\mathcal Q}}
    \bigl(\psi\otimes
    \ket{\chi_{\mathcal Q}}^{\otimes(k_{\mathcal Q}-k_{\mathrm{CCZ}})}\bigr),
  \label{eq:typed-active-encoding}
\end{equation}
with $B_{P_j}=B_j$, $B_{D_j}=B_j^{-T}$,
$\chi_{P_j}=0$, and $\chi_{D_j}=+$.  Equations
\eqref{eq:primal-active-encoding}--\eqref{eq:typed-active-encoding} specify
logical coordinates; physical circuits for the required linear maps are
constructed below.

In this basis the canonical logical labels of active Paulis are
\begin{equation}
  X(v):B_{\mathcal Q}(v,0),
  \qquad
  Z(w):B_{\mathcal Q}^{-T}(w,0).
  \label{eq:specified-Pauli-labels}
\end{equation}
The corresponding physical Pauli representatives are obtained from the maps
from logical Pauli labels to physical representatives in
\Cref{eq:canonical-logicals}.  If $m_Z$ is
a complete canonical logical $Z$-readout word, the active outcome is
$(B_{\mathcal Q}^{-1}m_Z)_{[k_{\mathrm{CCZ}}]}$; for a complete logical $X$-readout word it
is $(B_{\mathcal Q}^Tm_X)_{[k_{\mathrm{CCZ}}]}$.  The two readout maps have
$O(\log n)$ depth with bounded fan-in.

\subsection{Canonical code switching}
\label{sec:switching}

Let $\mathcal Q_1,\mathcal Q_2\in\mathfrak C$ have lengths
$n_i=\Theta(n)$ and dimensions $1\le k_i\le n_i$.  The
canonical transpose is the logical isometry
\begin{equation}
  U_\mathrm{tr}^{\mathcal Q_1\to\mathcal Q_2}:
  \bigl(\mathbb C^{2^{k_1}}\bigr)^{\otimes k_2}
  \longrightarrow
  \bigl(\mathbb C^{2^{k_2}}\bigr)^{\otimes k_1},
  \qquad
  (\beta,\alpha)\longmapsto(\alpha,\beta).
  \label{eq:canonical-transpose}
\end{equation}
We implement this map by a direct code switching grid rather than by
preparing its Choi state.

Arrange an $n_1\times n_2$ physical grid with old $\mathcal Q_1$ blocks in
columns and new $\mathcal Q_2$ blocks in rows.  The input occupies $k_2$
columns.  The other columns are complete zero or plus blocks placed according
to the pivot decomposition of $\mathcal Q_2$, so the grid has exactly $n_2$
columns even when $k_2$ is small.  Complete syndrome extraction on the columns is
followed by the row checks and the prescribed terminal readouts.  Exactly
$k_1$ rows survive and form the output $\mathcal Q_2$ blocks.

For one Pauli type, let $H_{\mathcal Q_i}$ and $G_{\mathcal Q_i}$ be the
syndrome and stabilizer maps of $\mathcal Q_i$, respectively, and let $E$
denote the effective error matrix during the transition.
If $S_\mathrm{col}$ and $S_\mathrm{row}$ are the repaired column and row syndrome
matrices, the parent detector is
\begin{equation}
  P_\mathrm{parent}
  =H_{\mathcal Q_1}S_\mathrm{row}+S_\mathrm{col}H_{\mathcal Q_2}^T
  =H_{\mathcal Q_1}EH_{\mathcal Q_2}^T.
  \label{eq:bcs-parent-detector}
\end{equation}
The actual decoder first finds $B$ with
$H_{\mathcal Q_1}B=P_\mathrm{parent}$ and the prescribed
columnwise representative property.  The decoder then solves the required linear
relations on each connected component of the syndrome graph, whose degree is bounded.
For a component $K$, the computed matrix $F_K$ obeys
\begin{equation}
  H_{\mathcal Q_1}F_KH_{\mathcal Q_2}^T=P_K,
  \qquad
  \operatorname{rowSupp}(F_K)\subseteq R_K,
  \qquad
  |\operatorname{colSupp}(F_K)|\le C_H|J_K|.
  \label{eq:bcs-local-filling}
\end{equation}
Put $F=\sum_KF_K$.  Although $F$ may be dense, correctness depends on its
two projection cardinalities, not on its total Hamming weight.

The algebraic step is the following rectangular cleaning identity.  If a
matrix $R$ is supported on $A\times B$, with $|A|<d_1$, $|B|<d_2$, and
$H_{\mathcal Q_1}RH_{\mathcal Q_2}^T=0$, then
\begin{equation}
  R\in \operatorname{im}G_{\mathcal Q_1}\otimes\mathbb F_2^{n_2}
      +\mathbb F_2^{n_1}\otimes\operatorname{im}G_{\mathcal Q_2}.
  \label{eq:rectangular-cleaning}
\end{equation}
Choose complements to the kernels of $H_{\mathcal Q_1}|_A$ and
$H_{\mathcal Q_2}|_B$.  The kernel of
their tensor product is the sum of the two kernel tensor spaces.  Any
normalizer supported inside $A$ or $B$ is below the corresponding code
distance and is therefore a stabilizer, proving
\eqref{eq:rectangular-cleaning}.

The faulty locations, their bounded neighborhoods, and the locations used by
the componentwise correction form a graph of bounded degree on $O(n^2)$ vertices.  A component of size
$u$ has a witness containing at least $u/C$ original faults.  Hence, below a
constant physical threshold,
\begin{equation}
  O(n^2)\sum_{u\ge a_0n}A^u2^up^{u/C}
  \le Ce^{-cn}.
  \label{eq:bcs-large-component}
\end{equation}
The same tail estimate over row and column subsets establishes every local
decoding and readout promise.  On the complementary event,
\eqref{eq:rectangular-cleaning} shows that $E+F$ is a sum of terms generated
by stabilizers of the old and new codes.

Let $\mathcal B_\mathrm{sw}$ be the union of the bad set families for oversized
connected components, failures of the row and column promises, extraction,
and readout for the two CSS components.
The union is a family of bad sets on the original physical locations of the grid and
satisfies
\begin{equation}
  \mathcal W(\mathcal B_\mathrm{sw};p)
  \le \operatorname{poly}(n)e^{-cn}.
  \label{eq:bcs-raw-bad-family}
\end{equation}
The family depends only on the faulty support.  In particular, it is chosen
before the measurement record and the values of the supported fault maps are
known.

\begin{lemma}[Canonical code switching]
\label{lem:canonical-code-switching}
The direct grid maps $k_2$ full $\mathcal Q_1$ blocks to $k_1$ full
$\mathcal Q_2$ blocks by \eqref{eq:canonical-transpose}.  The map is exact on
inputs entangled with arbitrary reference systems outside a bad event of
weight $Ce^{-cn}$.  Its quantum space is $O(n_1n_2)$ and its physical time is
$O(\log^2(n+2))$.  A parallel call on $s$ disjoint grids uses
$O(sn_1n_2)$ qubits and has a failure bound with a factor polynomial in $s$.
No lower bound on $k_1$ or $k_2$ is needed.  The quantum grid and its
measurements also define a raw switching segment in which the classical
component calculation and its Pauli correction are deferred.  The
corresponding recordwise output identity is proved below.
\end{lemma}

\begin{proof}
Before the row transition, complete syndrome extraction on each old column and
defer its known physical correction through the Clifford grid.  The repaired
records satisfy \eqref{eq:bcs-parent-detector}.  The componentwise algorithm
produces $F$ in $O(\log^2 n)$ depth with bounded fan-in and bounded fan-out: Boolean
adjacency squaring finds the components, parallel prefix ranks choose
independent rows and columns, and determinant and adjugate circuits invert
the chosen minors~\cite{Berkowitz1984}.  All matrices attached to the two
code members are
compilation data.

On a good support, the branch for the transition to rows can be written
$B_\mathrm{post}\Pi_sE\Phi_1$, where $\Phi_1$ is the padded input embedding.
Equation~\eqref{eq:rectangular-cleaning} gives, up to phase,
\begin{equation}
  FB_\mathrm{post}\Pi_sE\Phi_1
  =B_\mathrm{post}\Pi_{s'}\Phi_1.
  \label{eq:bcs-branch-cleaning}
\end{equation}
Old stabilizers act trivially on the input and new stabilizers are scalar after their
measurement.  Terminal raw words are corrected using the decoded syndrome
labels before their contributions to the Pauli frame on the surviving rows are
formed.  The final syndrome section is then applied physically, followed by
reserved error correction.

Tracking the canonical logical $X$ and $Z$ representatives across the grid
shows that every corrected nonzero branch intertwines the full input and
output logical Pauli algebras according to
\eqref{eq:canonical-transpose}.  Irreducibility makes it a scalar multiple of
the same transpose isometry, with scalar independent of the logical input.
Trace preservation fixes the total branch weights, completing the instrument
identity with arbitrary reference systems.

The quantum grid and all check ancillas occupy $O(n_1n_2)$ qubits.  The
quantum core has constant depth, followed by a constant number of
$O(\log^2n)$ classical tasks.  Memory correction relative to a stored baseline
is applied to the output rows
during that computation.  Equation~\eqref{eq:bcs-large-component} and the
row and column tail bounds give the exponentially small bad event for the
complete call.  The
argument uses the actual dimensions only through the occupied grid and is
therefore uniform down to $k_i=1$.

The raw switching segment uses the same quantum grid and physical locations,
with the component calculation and consolidated Pauli correction deferred.
The component algorithm, terminal word decoders, logical feedback, and
syndrome section define a total deterministic
Pauli $C_\mathrm{sw}(\omega)$ for every complete record $\omega$; the prescribed
fallbacks are used outside the decoding promises.  If an original fault
support avoids $\mathcal B_\mathrm{sw}$, the preceding cleaning argument gives,
for both sides of every Pauli matrix unit in the expansion of the supported
fault maps,
\begin{equation}
  K_{\omega,\bullet}
  =C_\mathrm{sw}(\omega)e_{\omega,\bullet}
    \mathcal E_{\mathcal Q_2}^\mathrm{can}
    U_\mathrm{tr}^{\mathcal Q_1\to\mathcal Q_2},
  \qquad
  |e_{\omega,\bullet}|_R\le t_\mathrm{sw}n,
  \qquad \bullet\in\{L,R\},
  \label{eq:raw-switch-output}
\end{equation}
up to phase.  The same $C_\mathrm{sw}(\omega)$ occurs on the two sides because
it is determined by the actual physical record, not by a term in the Pauli
  expansion.  Here
  $U_\mathrm{tr}^{\mathcal Q_1\to\mathcal Q_2}$ is the canonical transpose and
  $t_\mathrm{sw}>0$
is independent of the code dimensions, the record, and the input reference
system.  An incoming Pauli changes the later grid and readout records only by
the binary linear shifts determined by the compiled Clifford circuit.
Consequently \eqref{eq:raw-switch-output} is obtained directly from the grid
cleaning and its witnesses defined on the original fault supports, without using deferred
correction.  A fresh complete extraction is appended to each surviving output
before consecutive raw segments are composed.
\end{proof}

\subsection{Linear circuits for active encodings}

The active subspaces in \eqref{eq:typed-active-encoding} are realized by CNOT
circuits for linear maps on the block indices.

\begin{lemma}[CNOT circuit for a binary linear map]
\label{lem:clean-linear-addition}
For a binary $\mu\times\nu$ matrix $T$, the map
\begin{equation}
  \ket{x}_X\ket{y}_Y
  \longmapsto
  \ket{x}_X\ket{y+Tx}_Y
  \label{eq:clean-linear-addition}
\end{equation}
has a CNOT circuit using $O(\mu+\nu)$ clean zero work registers,
$O(\mu\nu+\mu+\nu)$ gates, and depth
$O(\min(\mu,\nu)+\log(\mu+\nu))$.  Every layer is a matching and all workspace
returns to zero.  The identity holds for arbitrary coherent inputs and
reference entanglement.
\end{lemma}

\begin{proof}
Suppose first that $\mu\ge\nu$ and put
$h=\lceil\mu/\nu\rceil$.  Partition the target indices into $h$ groups of at
most $\nu$.  For each control bit, use a balanced CNOT tree to make $h$
coherent copies in the computational basis, counting the original register as one
copy.  This uses $\nu(h-1)<\mu$ clean registers and depth
$\lceil\log_2h\rceil$.  Within one target group, the bipartite graph of all
nonzero entries of $T$ has degree at most $\nu$ at every copied control and
every target.  Padding gives a $\nu$-regular bipartite multigraph.  Hall's theorem
gives a perfect matching; removing one matching at a time decomposes the
graph into $\nu$ matching layers.  Reverse the fanout trees after these
layers to clear every copy.

Suppose instead that $\nu\ge\mu$ and put
$h=\lceil\nu/\mu\rceil$.  Partition the controls into $h$ groups of at most
$\mu$.  For every output and every group, introduce one zero scratch bit for
the corresponding partial parity.  The number of such registers is
$\mu h<\mu+\nu$.  The additions from controls to registers holding partial parities
form a bipartite graph of degree at most $\mu$ and hence use at most $\mu$
matching layers.  For each output, combine its $h$ partial parities by a
balanced reversible CNOT tree, add the resulting total to the target, and
reverse the tree and the computation of the partial parities.  This gives depth
$2\mu+2\lceil\log_2h\rceil+1$.

Both constructions use $O(\mu+\nu)$ clean registers and
$O(\mu\nu+\mu+\nu)$ gates.  They implement
\eqref{eq:clean-linear-addition} on every computational basis vector and clear
all work registers.  Linearity gives the same identity for coherent inputs
and arbitrary references.  All partitions, trees, and edge colorings are
determined during compilation.
\end{proof}

\begin{lemma}[CNOT circuit for an injective linear map]
\label{lem:clean-linear-embedding}
Let $A:\mathbb F_2^r\to\mathbb F_2^k$ be injective and choose a left inverse
$L$ with $LA=I_r$.  Then
\begin{equation}
  \ket{x}_X\ket0_Y
  \longmapsto
  \ket0_X\ket{Ax}_Y
  \label{eq:clean-linear-embedding}
\end{equation}
is implemented by first adding $Ax$ into $Y$ and then adding $Ly$ back into
$X$.  The circuit uses $O(k+r)$ block registers and
$O(r+\log(k+r))$ matching layers; its inverse compresses the image of $A$
back into $r$ registers.  All scratch registers end in zero, and the statement holds with
arbitrary references.
\end{lemma}

\begin{proof}
After the first addition, $Y=Ax$.  The second addition changes the source to
\begin{equation}
  X\longmapsto x+LY=x+LAx=0.
  \label{eq:linear-embedding-cleaning}
\end{equation}
The two applications of \Cref{lem:clean-linear-addition} clear their scratch
registers.  Starting instead from $X=0$ and $Y=Ax$ and applying the same two
additions in reverse order recovers $x$ and clears $Y$.  Since these are
reversible linear circuits, the basis identity proves the coherent identity
and its extension to arbitrary reference systems.
\end{proof}

The circuits for linear maps are lifted to full primary code blocks, with an error correction
cycle after every matching layer.  One zero work bank is retained throughout
the matching layers of a call.  Each matching is a transversal CNOT between
blocks of the same code, and memory error correction protects coherent
control copies, scratch blocks, and idle data until the inverse circuit clears
the workspace.  For one group the sharper resource bounds are
\begin{equation}
  O(n^2)\ \text{qubits},
  \qquad
  O\bigl(k_{\mathrm{CCZ}}+\log(n+2)\bigr)\ \text{time}.
  \label{eq:single-injection-resources}
\end{equation}
The $O(k_{\mathrm{CCZ}})$ matching layers each use error correction of constant depth, while
the work bank is prepared once.  The
number of elementary block windows is polynomial in $n$, so their union has
probability $\operatorname{poly}(n)e^{-cn}$ below the common threshold.

\begin{lemma}[Mixed active stabilizer resources]
\label{lem:mixed-stabilizer-resources}
For any $\mathcal Q\in\mathfrak C$, any
$\sigma\in\{0,+,+i,-i\}$, and any $q\ge1$, the state
$\mathcal E_{\mathcal Q}^\mathrm{act}
  (\ket\sigma^{\otimes k_{\mathrm{CCZ}}})$ can be prepared on $q$ output
blocks with
\begin{equation}
  \text{width }O\bigl(qn+F(n)\bigr),
  \qquad
  \text{time }O\bigl(k_{\mathrm{CCZ}}+g(n)\bigr),
  \label{eq:mixed-resource-bounds}
\end{equation}
and aggregate probability of its bad supports
$\operatorname{poly}(q,n)e^{-cn}$.  Outside this event, all outputs jointly
satisfy the corrected boundary condition about the stated stabilizer target,
including with arbitrary references and correlations among output blocks.
\end{lemma}

\begin{proof}
Let $P$ be the primary code of dimension $k=\Theta(n)$.  Start
with $k_{\mathrm{CCZ}}$ full primary blocks in the uniform state
$\ket\sigma^{\otimes k}$.  For a primal target $P_j$, apply
\Cref{lem:clean-linear-embedding} to their block index using $A_j$ and then
apply the canonical switch from $P$ to $P_j$.  The
$k_{\mathrm{CCZ}}\times k$ input logical array has $k$ columns equal to
$\ket\sigma^{\otimes k_{\mathrm{CCZ}}}$.  The embedding
circuit followed by transposition therefore gives exactly $k$ output blocks in
$\mathcal E_{P_j}^\mathrm{act}
  (\ket\sigma^{\otimes k_{\mathrm{CCZ}}})$, with zero on every complementary logical
coordinate.

For $\sigma=-i$, first prepare the $+i$ source and apply the chosen physical
representative of canonical logical $Z$ on every primary logical coordinate,
using $\ket{-i}=Z\ket{+i}$.  Thus all four source states in the statement are
prepared by \Cref{cor:encoded-stabilizer-resources} and one possible Pauli
layer.

For $q$ requested output blocks, put
\begin{equation}
  t=\left\lceil\frac qk\right\rceil,
  \label{eq:mixed-resource-batches}
\end{equation}
One batched preparation produces the $k_{\mathrm{CCZ}}t$ primary source rows, and a second
produces the zero work rows for the embedding circuits.  The $t$ embedding
groups run in parallel, followed by one invocation of
\Cref{lem:canonical-code-switching} on all resulting groups.  Extra outputs
are discarded at predetermined positions.  Since $k_{\mathrm{CCZ}}\le k$ and
$k=\Theta(n)$,
\begin{equation}
  k_{\mathrm{CCZ}}t\le q+k_{\mathrm{CCZ}}\le q+n,
  \qquad
  tn^2=O(qn+n^2).
  \label{eq:mixed-resource-counts}
\end{equation}
The first quantity counts the primary source rows; the second counts the
switching grids.  Source blocks, output blocks, embedding scratch, and
accumulated outputs remain under ordinary memory throughout the calls.

Full zero, plus, and $Y$ source preparation contributes
$O((q+n)n+nM(n))$ qubits.  The direct switching grids contribute
$O(qn+n^2)$.  Combining them gives the width in
\eqref{eq:mixed-resource-bounds}.  The embedding matchings cost
$O(k_{\mathrm{CCZ}})$ time,
whereas source preparation and switching cost $O(g(n))$.  Their elementary
block windows and complete preparation batches give the displayed
failure bound.

For a dual target, the definition \eqref{eq:dual-active-encoding} gives
\begin{equation}
  \mathcal E_{D_j}^\mathrm{act}
    (\ket\sigma^{\otimes k_{\mathrm{CCZ}}})
  =H_\mathrm{phys}
    \mathcal E_{P_j}^\mathrm{act}\!\left(
      (H\ket\sigma)^{\otimes k_{\mathrm{CCZ}}}\right).
  \label{eq:dual-mixed-resource}
\end{equation}
The required source states follow from
\begin{equation}
  H\ket0=\ket+,
  \qquad
  H\ket+=\ket0,
  \qquad
  H\ket{+i}=e^{i\pi/4}\ket{-i},
  \qquad
  H\ket{-i}=e^{-i\pi/4}\ket{+i}.
  \label{eq:Hadamard-stabilizer-states}
\end{equation}
The phases are global product phases.  One physical-$H$ layer followed
by error correction therefore gives the dual resources with the required
plus complement and the same asymptotic bounds.  No lower bound on an
auxiliary code rate is used.
\end{proof}

\subsection{Controlled phases from transversal \texorpdfstring{$\CCZ$}{CCZ}}

Let $(j_1,j_2,j_3)$ be a permutation of $(1,2,3)$.  For $z\in H_{j_3}$, define the
bilinear slice
\begin{equation}
  M_z^{j_1,j_2}:H_{j_1}\longrightarrow H_{j_2}^*,
  \qquad
  (M_z^{j_1,j_2}x)(y)=\tau(u_1,u_2,u_3),
  \qquad
  u_{j_1}=x,\quad u_{j_2}=y,\quad u_{j_3}=z.
  \label{eq:bilinear-slice}
\end{equation}
Define the corresponding logical diagonal gate by
\begin{equation}
  \CZ_{M_z^{j_1,j_2}}\ket{x}\ket{y}
  =(-1)^{\tau(u_1,u_2,u_3)}\ket{x}\ket{y}.
  \label{eq:logical-bilinear-slice}
\end{equation}
For each $j\in[3]$, let $C_{X,j}$ denote the binary $X$-stabilizer subspace
in the computational coordinates of the $j$th code, and choose a binary linear
lift
\begin{equation}
  \operatorname{lift}_j:H_j\longrightarrow\ker H_{Z,j}
  \label{eq:third-code-lift}
\end{equation}
of its logical computational representatives.  Write the corresponding
physical transversal circuit as a product of physical $\CCZ$ gates indexed
by a set $\Xi$ determined by the circuit, where $p_{\xi,j}$ is the physical qubit in the $j$th
code used by the gate indexed by $\xi$.

\begin{lemma}[A bilinear slice of the transversal action]
\label{lem:full-space-slice}
For every permutation $(j_1,j_2,j_3)$ of $(1,2,3)$ and every $z\in H_{j_3}$,
the physical layer
\begin{equation}
  V_{j_1,j_2;\operatorname{lift}_{j_3}(z)}
  =\prod_{\xi\in\Xi}
    \CZ_{p_{\xi,j_1},p_{\xi,j_2}}^{(\operatorname{lift}_{j_3}(z))_{p_{\xi,j_3}}}
  \label{eq:physical-slice}
\end{equation}
implements the logical gate $\CZ_{M_z^{j_1,j_2}}$ on the complete $H_{j_1}$
and $H_{j_2}$ logical spaces.  The identity is valid with an arbitrary
reference and requires no quantum program block.
\end{lemma}

\begin{proof}
In the chosen computational basis,
\begin{equation}
  \mathcal E_{P_{j_3}}^\mathrm{can}\ket z
  =|C_{X,j_3}|^{-1/2}
    \sum_{d\in C_{X,j_3}}\ket{\operatorname{lift}_{j_3}(z)+d}.
  \label{eq:encoded-program-coset}
\end{equation}
The physical $\CCZ$ gate is symmetric in its three inputs.  Keeping the
arguments of $\tau$ in their original order, the complete transversal unitary is block
diagonal in the computational basis of the $j_3$th code:
\begin{equation}
  U_{\tau}=\sum_w V_{j_1,j_2;w}\otimes\ket{w}\!\langle w\rvert_{j_3},
  \label{eq:transversal-block-diagonal}
\end{equation}
where $V_{j_1,j_2;w}$ is the product in \eqref{eq:physical-slice} with
$\operatorname{lift}_{j_3}(z)$ replaced by $w$.  Apply the identity on the full logical spaces in
\Cref{thm:qltc-family}, with the three logical arguments kept in
their original slots.  Every representative $\operatorname{lift}_{j_3}(z)+d$ in
\eqref{eq:encoded-program-coset} appears with nonzero coefficient, so
equality of coefficients gives
\begin{equation}
  V_{j_1,j_2;w}
    (\mathcal E_{P_{j_1}}^\mathrm{can}\otimes
     \mathcal E_{P_{j_2}}^\mathrm{can})
  =(\mathcal E_{P_{j_1}}^\mathrm{can}\otimes
    \mathcal E_{P_{j_2}}^\mathrm{can})\CZ_{M_z^{j_1,j_2}}
  \label{eq:full-space-slice-identity}
\end{equation}
for each such $w$, in particular for
$w=\operatorname{lift}_{j_3}(z)$.  The coefficient comparison permutes the
physical input slots and does not require symmetry of $\tau$.
\end{proof}

Let
$\zeta=(\zeta_1,\ldots,\zeta_{k_{\mathrm{CCZ}}})
\in\mathbb F_2^{k_{\mathrm{CCZ}}}$.  For any
ordered pair $j_1\ne j_2$, let $j_3$ be the remaining index.  Taking
$z=i_{j_3}(\zeta)$ and placing $i_{j_1}(x)$, $i_{j_2}(y)$, and
$i_{j_3}(\zeta)$ in their original three slots gives
\begin{equation}
  \tau(u_1,u_2,u_3)
  =\sum_{v=1}^{k_{\mathrm{CCZ}}}x_vy_v\zeta_v.
  \label{eq:subrank-slice}
\end{equation}
Hence \Cref{lem:full-space-slice} gives $\CZ$ on the active coordinates with
$\zeta_v=1$.  Conjugating the physical block in position $j_2$ by Hadamard
gives the following map between active encodings:
\begin{equation}
  \mathsf C_\zeta(P_{j_1}\to D_{j_2}):
  \ket{x}\ket{y}\longmapsto
  \ket{x}\ket{y+(\zeta_1x_1,\ldots,
    \zeta_{k_{\mathrm{CCZ}}}x_{k_{\mathrm{CCZ}}})},
  \qquad j_1\ne j_2.
  \label{eq:program-free-CNOT}
\end{equation}
The corresponding physical mask is
$\operatorname{lift}_{j_3}(i_{j_3}(\zeta))$.  Thus $j_3=3$
gives $P_1\to D_2$ and $P_2\to D_1$, $j_3=2$ gives $P_1\to D_3$ and
$P_3\to D_1$, and $j_3=1$ gives $P_2\to D_3$ and $P_3\to D_2$.  The three
binary linear maps from $\zeta$ to
$\operatorname{lift}_{j_3}(i_{j_3}(\zeta))$ are precomputed.
Their values for a runtime selector $\zeta$ are evaluated in $O(\log n)$
depth with bounded fan-in.

\begin{lemma}[Logical state transfer between codes]
\label{lem:type-conversion}
For any $\mathcal Q,\mathcal Q'\in\mathfrak C$, there is an exact
instrument implementing
\begin{equation}
  \mathcal E_{\mathcal Q}^\mathrm{act}(\psi)
  \longmapsto
  \mathcal E_{\mathcal Q'}^\mathrm{act}(\psi)
  \label{eq:logical-state-transfer}
\end{equation}
for every $\psi\in(\mathbb C^2)^{\otimes k_{\mathrm{CCZ}}}$.  The instrument uses a constant number
of mixed zero or plus target banks, transversal CNOT layers, complete source
readouts, and Pauli corrections.  The identity holds with an arbitrary
reference system.
\end{lemma}

\begin{proof}
Let $j\ne\ell$.  A full active identity transfer from $P_j$ to $D_\ell$
prepares
$\mathcal E_{D_\ell}^\mathrm{act}(\ket0^{\otimes k_{\mathrm{CCZ}}})$ and applies
$\mathsf C_{\mathbf1}(P_j\to D_\ell)$.  For an arbitrary
input entangled with a reference system, written as
\begin{equation}
  \sum_{x\in\mathbb F_2^{k_{\mathrm{CCZ}}}}
    \mathcal E_{P_j}^\mathrm{act}\ket{x}\ket{\phi_x}_{R},
  \label{eq:type-transfer-input}
\end{equation}
the state immediately before source measurement is
\begin{equation}
  \sum_x
    \mathcal E_{P_j}^\mathrm{act}\ket{x}\,
    \mathcal E_{D_\ell}^\mathrm{act}\ket{x}\ket{\phi_x}_{R}.
  \label{eq:type-transfer-entangled}
\end{equation}
A complete logical $X$ measurement of the source with active outcome $b$
leaves the target with the known frame $Z(b)$.  Apply its physical
representative from \eqref{eq:specified-Pauli-labels}.  After this correction,
the target and reference are in the state
\[
  \sum_x\mathcal E_{D_\ell}^\mathrm{act}\ket{x}\ket{\phi_x}_{R}.
\]
The source block has been measured and is discarded.  The instrument uses
one target block and one complete source readout for each source block.

For the reverse direction, prepare
$\mathcal E_{P_j}^\mathrm{act}(\ket+^{\otimes k_{\mathrm{CCZ}}})$ as the
control and use the old $D_\ell$ data as the target of the same transversal CNOT.
A complete logical $Z$ measurement of the old dual source produces an active
outcome $b$ and hence a known $X(b)$ frame on the new primal block.  After the
physical correction, the new primal block is in the active $P_j$ encoding of
the original logical state.  Both
identities follow on the state \eqref{eq:type-transfer-input}, and therefore
hold with arbitrary reference entanglement.

The transfer graph has vertex set $\mathfrak C$.  An edge $P_j-D_\ell$,
for $j\ne\ell$, records the availability of the direct transfer above.  The
graph is the six-cycle
\begin{equation}
  P_1-D_2-P_3-D_1-P_2-D_3-P_1.
  \label{eq:type-cycle}
\end{equation}
Thus any ordered pair of active encodings is connected by at most three
direct transfers.
\end{proof}

\begin{lemma}
\label{lem:selected-active-gates}
For every $\mathcal Q\in\mathfrak C$, full logical $H$ and
$S$ can be implemented on the active $\mathcal Q$ encoding.  Together with the masked CNOTs from
\eqref{eq:program-free-CNOT}, they give a swap on any prescribed subset of the
active coordinates.  Conjugating the full logical gates by these swaps gives
$H$, $S$, and the transversal logical $\CCZ$ on the same
subset.  All identities
hold with arbitrary references.
\end{lemma}

\begin{proof}
Physical Hadamard obeys
\begin{equation}
  H_\mathrm{phys}\mathcal E_{P_j}^\mathrm{act}(\psi)
  =\mathcal E_{D_j}^\mathrm{act}(H_\mathrm{act}\psi),
  \qquad
  H_\mathrm{phys}\mathcal E_{D_j}^\mathrm{act}(\psi)
  =\mathcal E_{P_j}^\mathrm{act}(H_\mathrm{act}\psi).
  \label{eq:full-H}
\end{equation}
Composing it with the appropriate logical state transfer implements full
active Hadamard and produces the output in the prescribed active encoding.

For full active $S$, prepare the resource using the same active encoding,
$\mathcal E_{\mathcal Q}^\mathrm{act}
  (\ket{+i}^{\otimes k_{\mathrm{CCZ}}})$, apply transversal CNOT from data to
ancilla, and destructively read the ancilla in the canonical $Z$ basis.  On
one active coordinate, outcome zero gives $S$, whereas outcome one
gives
\begin{equation}
  iS^\dagger=iZS.
  \label{eq:S-injection-branches}
\end{equation}
If the complete canonical outcome is $m$, put
$m_\mathrm{act}=(B_{\mathcal Q}^{-1}m)_{[k_{\mathrm{CCZ}}]}$ and apply the active $Z$ Pauli
with canonical logical label
\begin{equation}
  B_{\mathcal Q}^{-T}(m_\mathrm{act},0).
  \label{eq:S-correction}
\end{equation}
The correction turns every branch into $S$.  On complementary
coordinates, a CNOT between blocks encoded by the same code leaves
$\ket{00}$ unchanged for a primal code and $\ket{++}$ for a dual code.  Their measurement outcomes are discarded and
no complementary correction is applied.

For two codes joined by an edge of the transfer graph, let $H_A$ and
$H_B$ denote the full active Hadamard operations on the two banks.
The chronological circuit
\begin{equation}
  \mathsf C_\zeta;\ H_AH_B;\ \mathsf C_\zeta;\
  H_AH_B;\ \mathsf C_\zeta
  =\SWAP_\zeta
  \label{eq:selected-swap}
\end{equation}
is a swap on the selected coordinates and the identity elsewhere.  The two
physical-$H$ layers reverse the logical direction of the middle CNOT
on the selected coordinates, while the circuit remains the identity on every
disabled coordinate.

To restrict $V\in\{H,S\}$ to a mask $\zeta$, prepare an
active zero scratch block using a code adjacent in the transfer graph, apply
$\SWAP_\zeta$, apply the full active $V$ to the scratch, and apply
$\SWAP_\zeta$ again.  The induced operation on the data is $V$ precisely on
the selected coordinates.  The scratch is zero on the
selected coordinates and $V\ket0$ elsewhere, so it is independent of the
data and can be discarded unconditionally.

For an addressable $\CCZ$, conjugation by the selection swap
$\SWAP_{1-\zeta}$ places zero in the third operand at every disabled
coordinate.  The full transversal action is therefore the identity there and
acts as $\CCZ$ precisely on the requested mask.  To implement this identity,
logical state transfers first produce three operand blocks encoded by
$P_1,P_2,P_3$.  All pending syndrome computations are completed and their
physical Pauli corrections are applied before the selection swap, since a
large Pauli frame cannot be commuted formally through the non-Clifford layer.
The selection circuit may produce further syndrome sections and logical
Pauli corrections; these too are computed and applied physically before the
transversal $\CCZ$.  A reserved error correction cycle follows that layer,
after which the inverse selection swap restores the active zero scratch block.
Further logical state transfers give the prescribed output encodings.  Thus
the logical $\CCZ$ is the transversal operation of the computational code;
the encoded $Y$ state is used only for logical $S$.
\end{proof}

\subsection{Addressable gates and state transfer}
\label{sec:specified-operations}

\begin{lemma}[Logical state transfer preserving coordinate labels]
\label{lem:aligned-transfer}
Let the source and target banks each contain $q$ blocks, and let a
predetermined bijection between them preserve the active coordinate label.
Let $\Pi$ be the induced permutation.  For an input state
$\ket\psi=\sum_x\ket{x}_\mathrm{source}\ket{\phi_x}_R$, define
$\Pi\ket\psi=\sum_x\ket{\Pi x}_\mathrm{target}\ket{\phi_x}_R$.  The corresponding
permutation instrument is
implemented in $k_{\mathrm{CCZ}}$ CNOT rounds whose physical masks are determined by the
compiled matching, followed
by one complete source
readout.  Immediately before readout its ideal state is
\begin{equation}
  \sum_x\ket{x}_\mathrm{source}\ket{\Pi x}_\mathrm{target}
  \ket{\phi_x}_{R}.
  \label{eq:aligned-entanglement}
\end{equation}
After the decoded Pauli correction, the target is exactly
$\Pi\ket\psi$.  The construction uses one target block per source block,
not one resource per coordinate, and is exact with arbitrary references.
\end{lemma}

\begin{proof}
For adjacent primal and dual codes in the transfer graph, prepare one dual zero target for each
primal source.  At active coordinate $j$, apply
\eqref{eq:program-free-CNOT} along the compiled matching between source and
target blocks.  Follow each matching by error correction on every
involved block, and apply ordinary memory to all idle source and target
blocks.  Starting from
$\sum_x\ket{x}_\mathrm{source}\ket{\phi_x}_R$, induction over the
$k_{\mathrm{CCZ}}$ active
coordinates gives \eqref{eq:aligned-entanglement}.  Thus the complete logical
$X$ readout of the sources produces only a known $Z$ frame on the targets.
The known frame is computed and applied while the targets are protected, and
an error correction cycle completes the transfer.  The reverse orientation uses primal plus
targets, complete logical $Z$ readout of the old dual sources, and the
corresponding active $X$ correction.  General endpoint encodings add only a constant
number of additional identity transfers along \eqref{eq:type-cycle} before
and after these $k_{\mathrm{CCZ}}$ rounds.

The $k_{\mathrm{CCZ}}$ matching rounds use one target block per source block and one
complete readout of each source block, not $k_{\mathrm{CCZ}}$ separate target banks.  The resource
preparation, $O(qk_{\mathrm{CCZ}})$ block windows of bounded depth, readout, protected feedback,
and final error correction have width $O(qn+F(n))$, time
  $O(k_{\mathrm{CCZ}}+g(n))$, and aggregate probability of the bad supports
$\operatorname{poly}(q,n)e^{-cn}$.  On the complementary support event, the
Pauli expansion conditioned on the actual measurement record in
\Cref{lem:recordwise-Pauli} proves the same
instrument identity for arbitrary supported CPTP faults and arbitrary
references.
\end{proof}

\begin{proposition}[Addressable logical operations]
\label{prop:specified-logical-gates}
For one layer of disjoint logical operations on $q$ active blocks, with
predetermined masks and matchings, logical Paulis, $H$, $S$,
$\CNOT$, $\CZ$, and $\CCZ$ can be applied on any specified subset of the
first $k_{\mathrm{CCZ}}$ logical coordinates.  Logical states can also be transferred between
blocks while preserving their coordinate labels, with the same boundary
condition.  On every support outside the stated bad family, the
physical instrument implements the ideal logical instrument with its actual
measurement record, preserves arbitrary reference entanglement, and produces
surviving output blocks satisfying the boundary condition
\eqref{eq:typed-boundary-map}.  Each call has
\begin{equation}
  \text{width }O\bigl(qn+F(n)\bigr),
  \qquad
  \text{time }O\bigl(k_{\mathrm{CCZ}}+g(n)\bigr),
  \label{eq:addressable-resource-bounds}
\end{equation}
and aggregate probability of the bad supports
$\operatorname{poly}(q,n)e^{-cn}$.
\end{proposition}

\begin{proof}
Active Paulis use \eqref{eq:specified-Pauli-labels}.  The matrices
$B_{\mathcal Q}$ define logical labels and never appear as uncharged physical
gates.  The masked CNOT and CZ are given by
\eqref{eq:program-free-CNOT} and
\Cref{lem:full-space-slice}.  Full active Hadamard follows from
\eqref{eq:full-H} and logical state transfer; full active $S$ uses the
canonical $Y$ resource and \eqref{eq:S-correction}.  The selected one-qubit
gates and selected $\CCZ$ use the scratch constructions following
\eqref{eq:selected-swap}.  Before the transversal $\CCZ$, all three operands
are encoded by $P_1,P_2,P_3$, respectively; afterward logical state transfers
produce their prescribed output encodings.  Physical
Hadamard, CNOT, CZ, $\CCZ$, and Pauli layers have bounded support spread.
The following error correction cycles reduce every output error to the
radius $t_*n$.

Apart from the explicit $O(k_{\mathrm{CCZ}})$ matching rounds, the construction uses a
constant number of logical state transfers, complete resource calls, source
readouts, and scratch banks.  The resources and their accumulated outputs are
counted by \Cref{lem:canonical-code-switching} and
\eqref{eq:mixed-resource-bounds}.  The remaining family of bad fault sets
contains $O(qk_{\mathrm{CCZ}})$ block windows, together with a polynomial number of readout,
waiting, feedback, and windows for the final correction.  Since
$k_{\mathrm{CCZ}}\le n$, their union
contributes $\operatorname{poly}(q,n)e^{-cn}$ and no new asymptotic space
term.  Memory error correction is applied to live blocks during every classical computation, and
every pending Pauli correction is computed and applied physically before a
later non-Clifford layer.

Condition on avoiding these bad support events and retain the complete
measurement record.  Each identity proved above holds term by term
in the left--right Pauli expansion of \Cref{lem:recordwise-Pauli}.  Sequential
and parallel composition therefore gives the asserted fault-tolerant
instrument and \eqref{eq:addressable-resource-bounds}.  The global register
permutation constructed in \Cref{sec:main-proof} uses the transfers above,
which preserve coordinate labels inside the dense schedule.
\end{proof}

%% file: sections/qltc-decoder.tex
\section{qLTC and single-shot decoding}
\label{sec:qltc-decoder}

The qLTC family is constructed in
Ref.~\cite{GoodQLTCTransversal2026}.  It has linear distance, constant
normalized soundness, positive rate for the primary block, and a transversal
physical $\CCZ$ described by a cohomological invariant form.  The local code
maps and expansion properties yield the decoder and memory results proved
below.  Polynomially many diagonal coordinates of the trilinear form are
used later for addressable gates and their schedule; the decoder itself uses
the local geometry in degrees two and four.

\subsection{The computational qLTC family}
\label{sec:code-family}

For every sufficiently large integer $d$, let $P_j(d)$, $j\in[3]$, be the
three primal CSS codes associated with the arguments of the transversal
$\CCZ$ trilinear form, and let $D_j(d)$ be their physical-$\mathsf H$ dual
codes.  These are the six codes in $\mathfrak C_d$ from
\eqref{eq:computational-code-set}.  Write $n_d$ for the combined binary length after applying the
constant length repetition code.  Let $n_j$ and $k_j$ be the physical and
logical dimensions of the $j$th block, and let $H_j$ be its logical space.
For each $\mathcal Q\in\mathfrak C_d$, the listed check maps are the two
physical syndrome maps given by the
$X$- and $Z$-check matrices after binary restriction and repetition.
Their row lists retain the dependent and zero rows in the construction.  Thus, for a listed
map $H:\mathbb F_2^{n_H}\to\mathbb F_2^{m_H}$, the quantity $m_H$ is its
actual number of listed rows.

The decoder treats two syndrome problems.  For a measured degree two
cochain syndrome, compatible tensor product identifications on each upward
star allow local corrections to be chosen in parallel.  The expansion
estimates proved below show that their combined action contracts the error.
For the dual error, the canonical kernel diagrams reduce the chain problem
to the corresponding degree four cochain problem.

\begin{theorem}[Good qLTCs with transversal logical $\CCZ$ gates]
\label{thm:qltc-family}
The qLTC family has the following properties.
\begin{enumerate}
\item Each block length satisfies $n_j=\Theta(n_d)$.  The listed check matrices
  have $O(n_d)$ rows, including their dependencies
  and zero rows.  The physical checks have bounded weight and incidence.
  Every block has
  distance at least $\delta n_d$ for a constant $\delta>0$, and the primary code
  has $k_1\ge\kappa n_d$ for a constant $\kappa>0$.  The auxiliary codes are
  used at their actual logical dimensions.
\item Before binary restriction and repetition, the underlying cubical
  complex has dimension $6$ and the physical cochain degree is two.  On every
  upward star, the coefficient spaces admit compatible simultaneous tensor
  product identifications under which the restriction maps have the form
  stated in the construction.  The canonical kernel diagrams have the
  corresponding dual identifications and the local reduction between chains
  and cochains used below.
\item The local codes in the degree two and three coefficient diagrams and
  the degree four and five canonical kernel diagrams obey uniform product
  expansion estimates.  The degree two and four instances used for memory
  are stated in \Cref{sec:degree-specific-memory}.  Their parallel face graphs
  obey the spectral estimate used there, and each component contains at
  least a constant fraction of the faces in its degree.
\item For either physical syndrome map $H$, let $T$ be the next check map,
  so that $TH=0$.  There is a constant $a_T>0$ such that every word
  $y\in\ker T$ of weight at most $a_Tn_d$ lies in $\operatorname{im}H$.
  Such a word has a preimage under $H$ of weight $O(|y|)$, computable in
  $O(\log n_d)$ depth.  These statements also hold for the physical
  $\mathsf H$ duals and after the constant length repetition.
\item A cohomological invariant form
  \[
    \tau_d:H_1\times H_2\times H_3\longrightarrow\mathbb F_2
  \]
  gives the action of a strictly transversal physical $\CCZ$ circuit on the
  three complete logical spaces after the same repetition.  There are
  constants $\sigma\in(0,1]$ and $c_\sigma>0$ for which
  \begin{equation}
    k_{\mathrm{CCZ}}(\tau_d):=\operatorname{subrank}(\tau_d),
    \qquad
    k_{\mathrm{CCZ}}(\tau_d)\ge c_\sigma n_d^\sigma
    \label{eq:qltc-family-subrank}
  \end{equation}
  on every sufficiently large member of the family.
\item The family has normalized soundness bounded below by a constant.  For
  every listed check map $H:\mathbb F_2^{n_H}\to\mathbb F_2^{m_H}$ after
  binary restriction and repetition,
  \begin{equation}
    \frac{|Hx|}{m_H}
    \ge \rho_0
       \frac{\operatorname{dist}(x,\ker H)}{n_d},
    \qquad
    m_H=\Theta(n_d),
    \label{eq:global-soundness}
  \end{equation}
  for a constant $\rho_0>0$.
\item The usable combined lengths are
  \begin{equation}
    n_d=n_*(Q^{2d}-1)
    \label{eq:exact-dense-length}
  \end{equation}
  for every sufficiently large $d$, where $n_*$ and the even prime power $Q$
  are constants determined by the construction.
\end{enumerate}
\end{theorem}

Ref.~\cite{GoodQLTCTransversal2026} establishes the block parameters, the compatible local
code identifications, the transversal gate, the polynomial subrank bound,
and the soundness bound.  The product expansion and spectral estimates
needed by the decoder are verified below.  Section~\ref{sec:bounded-code-lengths}
constructs a sequence of congruence covers for which consecutive lengths
have bounded ratio; the covers need not be nested.  The local code maps and
the reduction between chains and cochains give the decoder estimates.

In the cohomological convention used for the computation, one CSS error type is described
by
\begin{equation}
  C^1\xrightarrow{\delta^1}C^2\xrightarrow{\delta^2}C^3,
  \qquad
  (H,G)=(\delta^2,\delta^1),
  \label{eq:cohomological-complex}
\end{equation}
and the other by
\begin{equation}
  C_3\xrightarrow{\partial_3}C_2\xrightarrow{\partial_2}C_1,
  \qquad
  (H,G)=(\partial_2,\partial_3).
  \label{eq:homological-complex}
\end{equation}
Thus $H$ is the measured syndrome map, $\im G$ is the stabilizer space for
the same Pauli type, and $HG=0$.  Binary restriction in trace dual bases, the repetition
conversion, and physical Hadamard change weights and local neighborhoods
only by constant factors.

\subsection{The logical \texorpdfstring{$\CCZ$}{CCZ} action}

Let $\sigma\in(0,1]$ and $c_\sigma>0$ be the constants supplied by
\Cref{thm:qltc-family}.  On every sufficiently large member, choose injections
\begin{equation}
  i_j:\mathbb F_2^{k_{\mathrm{CCZ}}(\tau_d)}\longrightarrow H_j,
  \qquad j\in[3],
  \label{eq:subrank-injections}
\end{equation}
with
\begin{equation}
  k_{\mathrm{CCZ}}(\tau_d)\ge c_\sigma n_d^\sigma,
  \qquad
  \tau_d(i_1x,i_2y,i_3z)
  =\sum_{v=1}^{k_{\mathrm{CCZ}}(\tau_d)}x_vy_vz_v
  \label{eq:subrank-diagonal}
\end{equation}
for all $x,y,z\in\mathbb F_2^{k_{\mathrm{CCZ}}(\tau_d)}$.

The injections and extensions of their images to logical bases are
compilation data.  Their physical implementation in
\Cref{sec:addressable-logic} uses CNOT circuits on block indices.  Each
auxiliary logical space contains the resulting $k_{\mathrm{CCZ}}(\tau_d)$-dimensional active
subspace, while the primary rate in \Cref{thm:qltc-family} gives the dense
storage capacity.

On the images of the three injections, the transversal operation
acts as
\begin{equation}
  \ket{x}\ket{y}\ket{z}\longmapsto
  (-1)^{\sum_{v=1}^{k_{\mathrm{CCZ}}(\tau_d)}x_vy_vz_v}
  \ket{x}\ket{y}\ket{z}.
  \label{eq:parallel-CCZ}
\end{equation}
It therefore applies $k_{\mathrm{CCZ}}(\tau_d)$ logical $\CCZ$ gates in parallel.  The matrices in
\eqref{eq:subrank-injections} are used as static compilation data below.

For the decoder analysis, fix one member of the family and write $n=n_d$.
The denominator in \Cref{eq:global-soundness} is the combined repeated
length $n$; using an individual block length changes only constants.  The
decoder and memory estimates below follow from the local geometry in degrees
two and four, with constants uniform over the family.

\subsection{Single-shot decoding}
\label{sec:single-shot-decoder}

We derive the decoder only in the two degrees used by the physical memory.
Degree two contains the measured cochain syndrome.  Degree four contains the
cochain problem obtained from the dual chain error.  The local tensor product
identifications and their duals are uniform over the family.  The
same is true of the local expansion and parallel face spectral bounds.
Coordinate changes inside a stalk preserve block support.

The small set inequality gives a restriction flip for a small reduced
cochain, and parallel colored flips contract the cochain syndrome problem.
For the chain syndrome problem, the upward view complex reduces the error to
degree four cochain decoding and reconstructs a global chain.  These two
decoders give the error correction and memory statements with a stored syndrome used
in the preceding sections.

\subsubsection{Expansion for small cochains}
\label{sec:degree-specific-memory}

Let $Z_j$ be the number of $j$-faces.  For $j=2,4$, the face and stalk
dimensions give constants $c_Z,C_Z>0$, uniform over the six codes in
$\mathfrak C_d$, such that
\begin{equation}
 c_Zn\le Z_j\le C_Zn.
 \label{eq:face-count-physical-length}
\end{equation}
The weight $|x|$ of a cochain is its Hamming support in face blocks.
A cochain is block reduced if it has
minimum block support in its coset modulo coboundaries.  It is locally
minimal on an upward link if no coboundary supported in that link reduces
the restricted support.  Every cochain reduced in block support is locally minimal on
each upward link, since a local improvement could be extended by zero to a
global improvement.

Let $d_-$ and $d_+$ be the minimum and maximum local degrees,
and put $\chi=d_+/d_-$.  For $0\le u<s\le6$, let
$\kappa_{s,u}>0$ be the uniform product expansion factor in degree $u$ on an
$s$-dimensional upward link, for both the coefficient diagrams and the
canonical kernel diagrams of the codes.  Thus every locally reduced
$u$-cochain $y$ obeys
\begin{equation}
 |\delta y|\ge \kappa_{s,u}d_-|y|.
 \label{eq:local-product-robustness}
\end{equation}
For a component $\mathcal C$ of a parallel face graph, let
$M_{\mathcal C}$ be its normalized adjacency operator.  Let $\lambda$ be
the common spectral bound on the subspace orthogonal to the constant and
bipartition functions.
Thus, for every component $\mathcal C$ and every nonnegative function $f$ on it,
\begin{equation}
 \langle f,M_{\mathcal C}f\rangle
 \le \lambda\lVert f\rVert_2^2
   +\frac1{|\mathcal C|}\left(\sum_{v\in\mathcal C}f(v)\right)^2.
 \label{eq:parallel-face-spectral}
\end{equation}
All these constants are determined by the code construction and are
independent of the cover.  For an $\ell$-face $g$, let
$X_{\ge g}(j)$ be the $j$-faces containing $g$, let
$\mathsf{Nb}_g(j)$ be the adjacent $j$-faces meeting this upward link, and
let $\mathsf{Op}_g(j)$ be the $j$-faces reached across one parallel face
edge.  These are the local sets used in the decoder of
Ref.~\cite{NguyenPattison2025}.

\begin{lemma}[Mixing in the parallel face graphs]
\label{lem:support-mixing}
Let $\mathcal A$ be a set of $j$-faces, where $j=2$ for a coefficient
diagram or $j=4$ for the canonical kernel diagram of the code.  Let
$0\le\ell\le j$.  If $f_\ell(g)$ is the number of faces of
$\mathcal A$ containing the $\ell$-face $g$, define
\begin{equation}
 S_\ell(\mathcal A)=
 \sum_{u\in\mathcal A}\ 
 \sum_{\substack{g\preceq u\\\dim g=\ell}}
 |\mathcal A\cap\mathsf{Op}_g(j)|.
 \label{eq:opposite-incidence-sum}
\end{equation}
Put
\begin{equation}
\begin{aligned}
 F_{j\ell}&=\binom j\ell 2^{j-\ell},\\
 C_{j\ell}&=(6-\ell)F_{j\ell}
              \binom{6-\ell}{j-\ell},\\
 B_{j\ell}&=(6-\ell)F_{j\ell}^2
              2^{5-j}\binom6j \chi^\ell.
\end{aligned}
\label{eq:support-mixing-constants}
\end{equation}
Then
\begin{equation}
 S_\ell(\mathcal A)\le d_+^{j+1-\ell}
 \left(
 C_{j\ell}\lambda |\mathcal A|
 +B_{j\ell}\frac{|\mathcal A|^2}{Z_j}
 \right).
 \label{eq:opposite-incidence-bound}
\end{equation}
\end{lemma}

\begin{proof}
Count a path from a $j$-face in $\mathcal A$ down to an $\ell$-face, across
one parallel face edge, and back to a $j$-face.  The number of downward
choices is $F_{j\ell}$, and the anisotropic upward multiplicities are bounded
by the displayed powers of $d_+$.  On each component of the parallel face graph, apply
the spectral estimate to the nonnegative function $f_\ell$.  Its
nonconstant part contributes the first term of
\eqref{eq:opposite-incidence-bound}.  The constant part contributes the
second term because every component contains a constant positive fraction of
all $\ell$-faces.  Summing over types and inactive directions gives
\eqref{eq:support-mixing-constants}.  In particular,
\begin{equation}
  \frac{B_{j\ell}}{C_{j\ell}}
  =\binom6\ell 2^{5-\ell}\chi^\ell
  \le120\chi^j.
  \label{eq:support-mixing-ratio}
\end{equation}
\end{proof}

\begin{lemma}
\label{lem:local-restriction}
Let $x$ be block reduced and let $0<\epsilon\le1$.  If no restriction
$w=x|_{X_{\ge g}(j)}$ obeys
\begin{equation}
 |\delta x|-|\delta x+\delta w|
 >(1-\epsilon)|\delta w|,
 \label{eq:restriction-flip}
\end{equation}
then, for every $\ell$-face $g$,
\begin{equation}
 \kappa_{6-\ell,j-\ell}d_-
 |x|_{X_{\ge g}(j)}
 \le \frac{2}{\epsilon}
 \left(
 |x|_{\mathsf{Nb}_g(j)}
 +|x|_{\mathsf{Op}_g(j)}
 \right).
 \label{eq:local-restriction-bound}
\end{equation}
\end{lemma}

\begin{proof}
Use the simultaneous tensor product identification on the upward link of
$g$.  The restriction
of $x$ is locally minimal there.  Local product expansion gives a
coboundary supported in the same link unless the restricted support is
bounded by the neighboring and opposite supports.  The fact that two
relevant faces intersect in at most one face ensures that every exterior
$j$-face contributing to
the changed syndrome is counted in only one relevant coface.  Applying the
triangle inequality to the failure of \eqref{eq:restriction-flip} gives
\eqref{eq:local-restriction-bound}.  Coordinate changes within a stalk do
not alter block support, so the same inequality holds in the actual
coefficient diagram.
\end{proof}

Summing \eqref{eq:local-restriction-bound} over all faces and substituting
\eqref{eq:opposite-incidence-bound} yields the next lemma.  No parameter in
these inequalities grows with the cover.

\begin{lemma}[Small set inequality]
\label{lem:small-set-inequality}
For $j=2$ in a coefficient diagram and $j=4$ in the canonical kernel
diagram of the code, there are constants
$L_j,Q_j,R_j>0$, independent of the cover, such that every
cochain $x$ reduced in block support satisfies
\begin{equation}
  |x|
  \le \lambda L_j|x|
     +Q_j\frac{|x|^2}{Z_j}
     +R_j|\delta x|.
  \label{eq:small-set-inequality}
\end{equation}
The construction parameters satisfy
$\lambda L_j<1$.  Consequently there are constants $\vartheta_j,C_j>0$
such that
\begin{equation}
  |x|\le \vartheta_jZ_j
  \quad\Longrightarrow\quad
  |x|\le C_j|\delta x|.
  \label{eq:small-set-expansion}
\end{equation}
\end{lemma}

\begin{proof}
For a cochain $x$ reduced in block support, apply local product expansion in every upward link.
The same calculation as in \Cref{lem:local-restriction}, now retaining the
syndrome term, gives a recurrence from dimension $j$ down to dimension
$\ell$.  With
\begin{equation}
 b_{j\ell}=
 \frac{
 \chi^{j-\ell+1}\binom j\ell\binom{6-\ell}{j-\ell}
 (6-\ell)2^{j-\ell}(j-\ell)!
 }{
 \displaystyle\prod_{u=\ell}^j\kappa_{6-u,j-u}
 },
 \label{eq:small-set-recursion-constant}
\end{equation}
and, for $0<\epsilon\le1$, define
\begin{align}
 \Lambda_j(\epsilon)
   &=\sum_{\ell=0}^j
       (2/\epsilon)^{j-\ell+1}b_{j\ell},
   \label{eq:restriction-linear-constant}\\
 \Theta_j(\epsilon)
   &=120\chi^j\Lambda_j(\epsilon),
   \label{eq:restriction-quadratic-constant}\\
 R_j
   &=\sum_{\ell=0}^j
     \frac{
       (j-\ell)!\chi^{j-\ell}\binom{j+1}{\ell}
       2^{j+1-\ell}\binom{6-\ell}{j-\ell}
     }{
       d_-\displaystyle\prod_{u=\ell}^j\kappa_{6-u,j-u}
     }.
   \label{eq:restriction-syndrome-constant}
\end{align}
One may take the following constants:
\begin{equation}
 L_j=\sum_{\ell=0}^j b_{j\ell},
 \qquad
 Q_j=120\chi^j L_j,
 \label{eq:small-set-linear-quadratic}
\end{equation}
The factors in \eqref{eq:small-set-recursion-constant} count the order in
which coordinates are removed, the powers of the incidence ratio $\chi$,
and the robustness loss at each local degree.  Equation
\eqref{eq:support-mixing-ratio} bounds the quadratic contribution.  A
syndrome face is counted only a constant number of times, giving
\eqref{eq:restriction-syndrome-constant}.  Summing the recurrence gives
\eqref{eq:small-set-inequality}.

Put
\begin{equation}
  \mathcal L=\max_{2\le u\le5}L_u.
  \label{eq:common-local-constant}
\end{equation}
Set $\epsilon_0=2^{-17}$.  Since $2/\epsilon_0=2^{18}$, the definitions above
give, for $2\le j\le5$,
\begin{equation}
  \Lambda_j(\epsilon_0)
  \le (2/\epsilon_0)^{j+1}L_j
  \le 2^{108}\mathcal L.
  \label{eq:restriction-linear-uniform-bound}
\end{equation}
For this construction, the face spectrum and local robustness estimates
of Ref.~\cite{GoodQLTCTransversal2026} give the explicit uniform bound
\begin{equation}
  2^{108}\lambda\mathcal L
  <2^{239627}512^{-239521}
  =2^{-1916062}<\frac14.
  \label{eq:explicit-degree-specific-margin}
\end{equation}
The same estimates hold for every cover in the sequence because the local diagrams
and the spectral bound for the face graphs are independent of the cover.  Hence the
memory margins in the two required degrees are
\begin{equation}
  \lambda\Lambda_2(\epsilon_0)\le1-\zeta_2,
  \qquad
  \lambda\Lambda_4(\epsilon_0)\le1-\zeta_4
  \label{eq:degree-specific-memory-margins}
\end{equation}
with the choice $\zeta_2=\zeta_4=3/4$.  In particular, in either
required degree,
\begin{equation}
  \lambda\Lambda_j(\epsilon_0)<1.
  \label{eq:retained-spectral-margin}
\end{equation}
In particular, $\lambda L_j<1$.  Taking
\begin{equation}
 \vartheta_j=\frac{1-\lambda L_j}{2Q_j},
 \qquad
 C_j=\frac{2R_j}{1-\lambda L_j}
 \label{eq:small-set-constants}
\end{equation}
makes the quadratic term at most half of
$(1-\lambda L_j)|x|$.  Absorbing the linear and quadratic terms on
the left proves \eqref{eq:small-set-expansion}.  All constants are uniform
because the local diagrams and degree ratios remain unchanged as the cover
grows.
\end{proof}

\begin{corollary}[Existence of a local flip]
\label{cor:restriction-flip}
For $j=2$ or $j=4$, put
\begin{equation}
 g_j=\frac{1-\lambda\Lambda_j(\epsilon_0)}
 {\Theta_j(\epsilon_0)}>0.
 \label{eq:restriction-flip-radius}
\end{equation}
Every cochain $x$ reduced in block support with $0<|x|<g_jZ_j$ has a
restriction $w$ to an upward star for which
\begin{equation}
  |\delta x|-|\delta x+\delta w|
  >(1-\epsilon_0)|\delta w|.
  \label{eq:restriction-flip-conclusion}
\end{equation}
\end{corollary}

\begin{proof}
If no such restriction existed, apply
\eqref{eq:local-restriction-bound} at every face and sum.  The spectral
mixing estimate gives a bound of the form
\[
  |x|\le\lambda\Lambda_j(\epsilon_0)|x|
      +\Theta_j(\epsilon_0)\frac{|x|^2}{Z_j}
\]
and \eqref{eq:restriction-flip-radius} makes the right-hand side strictly
smaller than $|x|$, a contradiction.

The restriction $x|_{X_{\ge g}(j)}$ is used only in this existence proof.
The decoder does not know $x$.  Since every upward link has constant size, it
enumerates all local cochains $w$, computes their syndrome changes from the
observed syndrome, and uses a prescribed tie order among the candidates satisfying
the gain inequality.
\end{proof}

The parallel decoder uses the rule
\begin{equation}
  |\sigma|-|\sigma+\delta w|
  \ge\frac12|\delta w|,
  \label{eq:half-gain}
\end{equation}
and chooses one with maximum $|\delta w|$ among the local flips satisfying this
gain condition
with $|\delta w|>0$.  If there is no such flip, it chooses $w=0$.  Among
words giving the same syndrome change, it first minimizes $|w|$ and then
uses the prescribed tie order.  In particular, zero syndrome produces zero
correction.  The stricter $\epsilon_0$ restriction from
\Cref{cor:restriction-flip} is used in the parallel contraction proof to
certify sufficiently many eligible flips.  Both searches are over a
upward link of constant size and therefore have constant local complexity.
\subsubsection{Parallel cochain decoder}

Color the intersection graph of upward stars.  Its maximum degree is bounded
by a constant, so a constant number of colors suffices.  Within one color, the stars have
disjoint syndrome supports and their flips may be performed simultaneously.

\begin{lemma}
\label{lem:block-reduced-subcochain}
Let $z$ be block reduced.  If $z'\subseteq z$, meaning that
$\supp z'\subseteq\supp z$ and $z'$ agrees with $z$ on its support, then
$z'$ is block reduced.
\end{lemma}

\begin{proof}
For a coboundary $y$, let $A$ be the faces in
$\supp y\setminus\supp z$ and let $B$ be the faces on which $y$ agrees with
$z$.  Block reduction gives
\[
  0\le |z+y|-|z|=|A|-|B|.
\]
Define $A'$ and $B'$ in the same way with $z'$ in place of $z$.  Since
$z'\subseteq z$, one has $|A'|\ge|A|$ and $|B'|\le|B|$.  Therefore
\[
  |z'+y|-|z'|=|A'|-|B'|\ge0,
\]
which proves the claim.
\end{proof}

\begin{lemma}
\label{lem:parallel-progress}
Let $\sigma=\delta x$, where $x$ is block reduced and lies inside the
radius in \eqref{eq:restriction-flip-radius} corresponding to
$\epsilon_0=2^{-17}$.  One
complete round through all color
classes decreases the syndrome weight by at least
\begin{equation}
  \delta_{\parallel}|\sigma|,
  \qquad
  \delta_{\parallel}
  =\frac{\epsilon_0}{32}
     \bigl(1-4^8\epsilon_0\bigr)=2^{-23}.
  \label{eq:parallel-contraction-constant}
\end{equation}
\end{lemma}

\begin{proof}
Repeatedly apply \Cref{cor:restriction-flip} to the current block reduced
cochain, each time taking its restriction to the current
upward link.  By \Cref{lem:block-reduced-subcochain}, the remaining
subcochain stays block reduced.  This gives pairwise disjoint cochains
$w^{(1)},\ldots,w^{(T)}$ whose sum is $x$.  Put
$\nu^{(i)}=\delta w^{(i)}$.  The restriction gain inequality gives
\begin{equation}
  |\sigma|
  \ge(1-\epsilon_0)\sum_{i=1}^T|\nu^{(i)}|.
  \label{eq:hypothetical-sequential-gain}
\end{equation}

For each $i$, let $\nu'^{(i)}$ be the part of $\nu^{(i)}$ that belongs to no
other $\nu^{(i')}$.  A syndrome face contains fewer than $4^6$ relevant
subfaces, so it occurs in fewer than $4^6$ of the changes.  With
$\epsilon'=4^7\epsilon_0$, call $i$ good when
$|\nu'^{(i)}|\ge(1-\epsilon')|\nu^{(i)}|$.  Counting the multiply covered
faces in \eqref{eq:hypothetical-sequential-gain} gives
\begin{equation}
  \left|\bigsqcup_{i\ \mathrm{good}}\nu'^{(i)}\right|
  \ge\frac{\epsilon_0}{4}|\sigma|.
  \label{eq:parallel-good-coverage}
\end{equation}

Consider the upward link used by one good $w^{(i)}$.  During its color
subround, let $\nu^\parallel$ be the syndrome change chosen by the parallel
decoder, and let $U$ be the union of the supports of all changes chosen in
the complete round.  If fewer than
\[
  \frac14(1-4^8\epsilon_0)|\nu^{(i)}|
\]
coordinates of $\nu'^{(i)}$ met $U$, then adding $\nu^{(i)}$ to the
candidate $\nu^\parallel$ would still satisfy the gain condition
\eqref{eq:half-gain} and would
have strictly larger support.  This contradicts the maximal choice in that
upward link.  The sets $\nu'^{(i)}$ are disjoint, so
\eqref{eq:parallel-good-coverage} yields
\[
  |U|\ge
  \frac{\epsilon_0}{16}
  (1-4^8\epsilon_0)|\sigma|.
\]
Every accepted update decreases the syndrome by at least half the size of
its change, and changes in the same color class have disjoint syndrome supports.
Serializing the constant number of colors therefore proves
\eqref{eq:parallel-contraction-constant}.
\end{proof}

\begin{lemma}[Parallel cochain decoder]
\label{lem:parallel-cochain}
For every $\eta>0$, a constant number of colored rounds gives a circuit of
constant depth with
\begin{equation}
  |e+\mathcal D(He+m)|_R
  \le\eta|e|_R+\gamma_{\mathrm{dec}}|m|
  \label{eq:decoder-contraction-sec6}
\end{equation}
whenever $|e|_R+\alpha|m|\le\beta n$.  With noiseless syndrome,
$O(\log n)$ rounds give exact correction and a correction of absolute
weight $O(|e|_R)$.
\end{lemma}

\begin{proof}
By \Cref{lem:parallel-progress}, one complete set of colors reduces the
noiseless syndrome by a factor $1-\delta_{\parallel}$.  The condition on the
error size is preserved as follows.  If $\Delta_j$ bounds the number of
syndrome coordinates incident to one error coordinate, the initial syndrome
has weight at most $\Delta_j|e|_R$, and every accepted update decreases that
weight.  Choose $\beta>0$ so that $C_j\Delta_j\beta Z_j$ lies below one
quarter of the radii in \eqref{eq:small-set-expansion} and
\eqref{eq:restriction-flip-radius}.  After increasing the first usable
length, one local update cannot cross the remaining gap because its support
has constant size.  The small set inequality then keeps every iterate, after
reduction in block support, inside both
radii.  Iterating $s$ complete rounds gives
\begin{equation}
  |e+f_s|_R
  \le C_j(6-j)d_+
       (1-\delta_{\parallel})^s|e|_R.
  \label{eq:parallel-geometric-contraction}
\end{equation}
Thus any prescribed constant contraction $\eta$ is obtained with a constant
number of rounds, and $O(\log n)$ rounds reduce the noiseless residual below
one, where it is a coboundary.

For noisy syndrome, compare the decoder circuit with the prescribed number
of rounds on $He$ and on $He+m$.  This circuit has bounded fan-in and
bounded fan-out.  The forward light cones of
$|m|$ changed input coordinates contain at most $C_\mathrm{light}|m|$ output
coordinates.  Taking $\gamma_{\mathrm{dec}}=C_\mathrm{light}$ and decreasing
the input radius by constant factors, comparison with
\eqref{eq:parallel-geometric-contraction} gives
\eqref{eq:decoder-contraction-sec6} whenever
$|e|_R+\alpha|m|\le\beta n$.

Finally, if $B_\mathrm{flip}$ bounds the support of one local flip, the rule
requiring at least half of the syndrome change
and serialization of changes in the same color class give
\[
  \sum_{\text{all flips }w}|w|
  \le 2B_\mathrm{flip}|\sigma_0|=O(|e|_R)
\]
in the noiseless case.  This proves the bound on the absolute weight of the correction.
\end{proof}

\subsubsection{Reduction of chain errors}

Canonical kernel diagrams reduce the second CSS error type, through local
maps, to cochain decoding on the kernels of the top local boundary maps.

\begin{definition}[Upward view complex]
\label{def:upward-view-complex}
Let $X$ denote the underlying cubical complex.  For the
chain degree $q$, define
\begin{equation}
 \mathbf C_q^i
 =\bigoplus_{g\in X(i)} C_q(X_{\ge g}),
 \qquad 0\le i\le q,
 \label{eq:upward-view-space}
\end{equation}
where $C_q(X_{\ge g})$ is the direct sum of the actual chain stalks on
$q$-faces containing $g$.  The horizontal map
$\mathbf\Delta_q^i:\mathbf C_q^i\to\mathbf C_q^{i+1}$ is
\begin{equation}
 (\mathbf\Delta_q^iy)(g')
 =\sum_{g\prec g'}y(g)|_{X_{\ge g'}(q)}.
 \label{eq:upward-view-map}
\end{equation}
The block weight counts the faces $g$ on which the local view is nonzero.
\end{definition}

\begin{lemma}
\label{lem:upward-view-exactness}
For $1\le i\le q$, the complex in
\eqref{eq:upward-view-space} is exact at $\mathbf C_q^i$, with chosen local
right inverses that map zero to zero.  If
$z\in\mathbf C_q^0$ satisfies $\mathbf\Delta_q^0z=0$, then its local views
are the restrictions of one global chain in $C_q(X)$.  The local right
inverses and the reconstruction map each input coordinate into a bounded
neighborhood, have polynomial size, and have constant parallel depth.
\end{lemma}

\begin{proof}
Reorder the direct sum in \eqref{eq:upward-view-space} by top $q$-faces:
\begin{equation}
 \mathbf C_q^i
 =\bigoplus_{u\in X(q)} C^i(X_{\le u},V_u),
 \label{eq:downward-cube-decomposition}
\end{equation}
where $V_u$ is the actual stalk on $u$.  On each summand,
$\mathbf\Delta_q$ is the ordinary cochain differential of the downward
cube $X_{\le u}$.  That cube is a product of complexes with two terms and has
no cohomology in positive degree.  Its size is constant because the geometric
dimension and all stalk dimensions are independent of the cover.  Choose a
right inverse on
each summand and take their direct sum.  This proves exactness, the support
bound, and the statement about circuit depth.

For the consistency statement, assign to a $q$-face $u$ the value given by
the local view at any vertex below $u$.  Equation
$\mathbf\Delta_q^0z=0$ makes the values at the two endpoints of every edge
equal.  The vertices below a cube are connected, so the assignment is
independent of the chosen vertex and defines a global chain.  Each global
coordinate is copied from only a constant number of local views.
\end{proof}

Let $\mathbf\partial_L$ apply the actual transpose boundary separately to
each local view.  Compatibility of the local tensor product identifications
gives
\begin{equation}
 \mathbf\partial_L^2=0,
 \qquad
 \mathbf\Delta\mathbf\partial_L
 =\mathbf\partial_L\mathbf\Delta.
 \label{eq:chain-bicomplex}
\end{equation}
Below the top local degree, $\mathbf\partial_L$ is exact and has chosen local
right inverses.  More precisely, for every $q<6$ and every cochain degree
$i$, choose a local map
\begin{equation}
 \mathsf J_q^i:
 \ker\!\left(
   \mathbf\partial_L:\mathbf C_q^i\longrightarrow\mathbf C_{q-1}^i
 \right)
 \longrightarrow \mathbf C_{q+1}^i,
 \qquad
 \mathbf\partial_L\mathsf J_q^i z=z.
 \label{eq:local-boundary-section}
\end{equation}
It maps zero to zero and has bounded support spread.  At the top degree the kernel is the tensor product of the
coordinate dual local codes.  Under the compatible dual identifications,
this kernel and its incidence maps agree with the cochain diagram formed by
the actual kernels of the top local boundary maps.

For a local vertical equation, extend
\eqref{eq:local-boundary-section} to a total map by
\begin{equation}
 \overline{\mathsf J}_{q,g}^i(y)=
 \begin{cases}
  \mathsf J_{q,g}^i(y),&
  y\in\operatorname{im}\partial_{L,g}
     =\ker\partial_{L,g},\\
  0,&\text{otherwise}.
 \end{cases}
 \label{eq:total-local-boundary-section}
\end{equation}
All other local inverse maps below are extended in the same way: a consistent
equation is solved by the chosen local inverse and an inconsistent equation
is mapped to zero.  The resulting maps preserve zero and have bounded support
spread.  The decoder evaluates every level on every record and follows the
same schedule without an early exit.

\begin{lemma}[Chain reduction]
\label{lem:chain-reduction}
For every $\eta>0$, there are constants
$\alpha_\mathrm{ch},\beta_\mathrm{ch},\gamma_\mathrm{ch}>0$ and a decoder
$\mathcal D_\mathrm{ch}$ of constant depth for the degree two chain syndrome problem such
that
\begin{equation}
 \bigl|e+\mathcal D_\mathrm{ch}(\partial_2e+m)\bigr|_R
 \le \eta|e|_R+\gamma_\mathrm{ch}|m|
 \label{eq:chain-decoder-contraction}
\end{equation}
whenever
$|e|_R+\alpha_\mathrm{ch}|m|\le\beta_\mathrm{ch}n$.
For $m=0$, an $O(\log n)$-depth version returns a correction $f$ satisfying
\begin{equation}
 e+f\in\operatorname{im}\partial_3,
 \qquad
 |f|=O(|e|_R).
 \label{eq:chain-decoder-exact}
\end{equation}
Every map in the reduction sends one input coordinate into a bounded
neighborhood and maps zero input to zero.
\end{lemma}

\begin{proof}
Let $\bar e$ be a representative of minimum support in
$e+\operatorname{im}\partial_3$, so that $|\bar e|=|e|_R$.  Write
$\widehat x^{(0)}\in\mathbf C_2^0$ for the upward local views of $\bar e$.
Let $\widehat\sigma^{(0)},\sigma^{(0)}\in\mathbf C_1^0$ be the vertexwise
local views of $\partial_2\bar e$ and $\partial_2\bar e+m$, respectively.
At every vertex, choose a local explanation of minimum support for the ideal
syndrome $\partial_2\bar e$; these are the ideal guesses
$\widehat g^{(0)}$.  Applying the same deterministic rule to the measured
syndrome $\partial_2\bar e+m$ gives $g^{(0)}$.  The prescribed output for an
inconsistent local equation is zero by \eqref{eq:total-local-boundary-section}.  Thus
\begin{equation}
 \mathbf\partial_L\widehat g^{(0)}
 =\mathbf\partial_L\widehat x^{(0)}
 =\widehat\sigma^{(0)}.
 \label{eq:level-zero-local-syndrome}
\end{equation}

\paragraph{Ideal and noisy syndrome explanation sequences.}
For $0\le i\le4$, define
\begin{equation}
 \widehat\sigma^{(i+1)}
   =\mathbf\Delta_{2+i}^i\widehat g^{(i)},
 \qquad
 \sigma^{(i+1)}
   =\mathbf\Delta_{2+i}^ig^{(i)}.
 \label{eq:ideal-noisy-syndromes}
\end{equation}
For $0\le i<4$, the ideal local equation is consistent, and we set
\begin{equation}
 \widehat g^{(i+1)}
   =\mathsf J_{2+i}^{i+1}\widehat\sigma^{(i+1)},
 \qquad
 g^{(i+1)}
   =\overline{\mathsf J}_{2+i}^{i+1}\sigma^{(i+1)}.
 \label{eq:ideal-noisy-upward-lifts}
\end{equation}
The ideal guesses have a second sequence that records the underlying error.
Starting from $\widehat x^{(0)}$, local exactness gives
\begin{equation}
 \mathbf\partial_L\widehat z^{(i)}
   =\widehat g^{(i)}+\widehat x^{(i)},
 \qquad
 \widehat x^{(i+1)}
   =\mathbf\Delta_{3+i}^i\widehat z^{(i)},
 \qquad 0\le i<4.
 \label{eq:ideal-error-sequence}
\end{equation}
Induction using \eqref{eq:chain-bicomplex} gives
\begin{equation}
 \mathbf\Delta\widehat x^{(i)}=0,
 \qquad
 \mathbf\partial_L\widehat x^{(i)}
   =\widehat\sigma^{(i)}.
 \label{eq:ideal-sequence-identities}
\end{equation}
The initial minimum explanations have total support $O(|e|_R)$.  Since the
number of levels is constant and every map has bounded support spread, there
are constants $c_i,c_i'$ such that
\begin{align}
 |\widehat g^{(i)}|+|\widehat x^{(i)}|
   &\le c_i'|e|_R,\label{eq:ideal-sequence-weight}\\
 |g^{(i)}-\widehat g^{(i)}|
 +|\sigma^{(i+1)}-\widehat\sigma^{(i+1)}|
   &\le c_i|m|.\label{eq:noisy-sequence-weight}
\end{align}
The second bound follows by comparing the two executions one local map at a
time.  A changed input coordinate affects only a bounded neighborhood, and an
inconsistent equation can occur only in the same bounded neighborhood.

\paragraph{The degree four cochain problem.}
Let $\mathsf U_i$ and $\mathsf E_i$ be the local unencoding and reencoding
maps furnished by the canonical kernel diagram.  They preserve block support
and commute with the horizontal incidence
maps.  At the top of \eqref{eq:ideal-error-sequence}, define
\begin{equation}
 \widetilde c
 =\mathsf U_4\bigl(\widehat g^{(4)}+\widehat x^{(4)}\bigr),
 \qquad
 \widetilde s
 =\mathsf U_5\bigl(\widehat\sigma^{(5)}\bigr)
 =\widetilde\delta\widetilde c.
 \label{eq:top-cochain-and-syndrome}
\end{equation}
Here $\widetilde c$ is a degree four cochain in the
canonical kernel diagram of the code, while $\widetilde s$ is its syndrome.
Apply the unencoding rule whose output is zero for an inconsistent local equation to
$\sigma^{(5)}$ and denote the result by $s$.  Equations
\eqref{eq:ideal-sequence-weight} and
\eqref{eq:noisy-sequence-weight} give constants
$c_\mathrm{top},c_\mu>0$ such that
\begin{equation}
 |\widetilde c|_R\le c_\mathrm{top}|e|_R,
 \qquad
 |s-\widetilde s|\le c_\mu|m|.
 \label{eq:top-cochain-weight}
\end{equation}
For the finite list of diagrams, the number $Z_4$ of degree four faces
satisfies $Z_4\ge c_Zn$ for a constant $c_Z>0$.  Decreasing
$\beta_\mathrm{ch}$ and increasing $\alpha_\mathrm{ch}$ by constant factors makes
the promise in \eqref{eq:chain-decoder-contraction} imply the required bound
on the error size for the degree four cochain decoder.

Apply that decoder with contraction $\theta>0$ and correction $d$.  A cochain
$\xi_\mathrm{top}$ and a minimum representative $\omega$ can be chosen so that
\begin{equation}
 \widetilde c+d=\widetilde\delta\xi_\mathrm{top}+\omega.
 \label{eq:top-cochain-residual}
\end{equation}
By \Cref{lem:parallel-cochain},
\begin{equation}
 |\omega|
 \le\theta c_\mathrm{top}|e|_R+\gamma_4c_\mu|m|
 \label{eq:top-residual-bound}
\end{equation}
for a constant $\gamma_4$.  The corrected noisy top guess is
$g'^{(4)}=g^{(4)}+\mathsf E_4(d)$.  Compatibility of the local maps gives
\begin{equation}
 g'^{(4)}
 =\widehat x^{(4)}
  +\mathbf\Delta_6^3\mathsf E_3(\xi_\mathrm{top})
  +\mathsf E_4(\omega)
  +(g^{(4)}-\widehat g^{(4)}).
 \label{eq:corrected-top-guess}
\end{equation}
Thus the corresponding ideal corrected guess
\begin{equation}
 \widehat g'^{(4)}
 =\widehat x^{(4)}
  +\mathbf\Delta_6^3v^{(3)},
 \qquad
 v^{(3)}=\mathsf E_3(\xi_\mathrm{top}),
 \qquad
 \mathbf\partial_Lv^{(3)}=0
 \label{eq:ideal-corrected-top-guess}
\end{equation}
differs from $g'^{(4)}$ on at most
$C(|\omega|+|m|)$ local coordinates.

\paragraph{Reverse syndrome explanation sequence.}
The reverse construction starts from $\widehat g'^{(4)}$ with the induction invariant
\begin{equation}
 \widehat g'^{(i)}
 =\widehat x^{(i)}
  +\mathbf\Delta_{2+i}^{i-1}v^{(i-1)},
 \qquad
 \mathbf\partial_Lv^{(i-1)}=0.
 \label{eq:reverse-sequence-invariant}
\end{equation}
The invariant, \eqref{eq:ideal-sequence-identities}, and
$\mathbf\Delta^2=0$ give
\begin{equation}
 \mathbf\Delta_{2+i}^{i}\widehat g'^{(i)}=0,
 \qquad
 \mathbf\partial_L\widehat g'^{(i)}
   =\widehat\sigma^{(i)}.
 \label{eq:ideal-corrected-consistency}
\end{equation}
Thus, given this identity at level $i\ge1$, horizontal exactness permits the
prescribed
preimage $\widehat u^{(i-1)}$ satisfying
\begin{equation}
 \mathbf\Delta_{2+i}^{i-1}\widehat u^{(i-1)}
   =\widehat g'^{(i)}
 \label{eq:ideal-horizontal-preimage}
\end{equation}
and put
\begin{equation}
 \widehat g'^{(i-1)}
 =\widehat g^{(i-1)}
  +\mathbf\partial_L\widehat u^{(i-1)}.
 \label{eq:ideal-downward-update}
\end{equation}
Indeed, with
$w^{(i-1)}=\widehat u^{(i-1)}+\widehat z^{(i-1)}+v^{(i-1)}$,
\eqref{eq:ideal-error-sequence} and
\eqref{eq:reverse-sequence-invariant} give
$\mathbf\Delta w^{(i-1)}=0$.  For $i>1$, horizontal exactness gives
$w^{(i-1)}=\mathbf\Delta b^{(i-2)}$.  Taking
$v^{(i-2)}=\mathbf\partial_Lb^{(i-2)}$ proves
\eqref{eq:reverse-sequence-invariant} at the next level.

At $i=1$, the same calculation gives
$w^{(0)}\in\ker\mathbf\Delta_3^0$.  By
\Cref{lem:upward-view-exactness}, these are the local views of a global
three chain $z$.  The reconstructed ideal correction $\widehat f$ satisfies
\begin{equation}
 \widehat f+\bar e=\partial_3z.
 \label{eq:ideal-chain-boundary}
\end{equation}
This proves homological equivalence for the ideal descent.

The noisy descent applies the same local rules to $g'^{(4)}$, using
zero on every inconsistent horizontal equation and on every inconsistent
final local view.  Since there are only four downward levels, bounded
support spread gives a constant $c_\mathrm{down}$ such that the final
correction $f$ obeys
\begin{equation}
 |f-\widehat f|
 \le c_\mathrm{down}\bigl(|\omega|+|m|\bigr).
 \label{eq:noisy-descent-bound}
\end{equation}
Combining \eqref{eq:ideal-chain-boundary},
\eqref{eq:top-residual-bound}, and
\eqref{eq:noisy-descent-bound} gives
\begin{equation}
 |e+f|_R
 \le c_\mathrm{down}c_\mathrm{top}\theta |e|_R
   +c_\mathrm{down}(\gamma_4c_\mu+1)|m|.
 \label{eq:chain-reduction-final-bound}
\end{equation}
Choose
$\theta=\eta/(c_\mathrm{down}c_\mathrm{top})$ and absorb the constants into
$\alpha_\mathrm{ch},\beta_\mathrm{ch},\gamma_\mathrm{ch}$.  This proves
\eqref{eq:chain-decoder-contraction}.

For the exact decoder, $m=0$ and the degree four cochain subroutine runs for
$O(\log n)$ colored rounds.  Its residual then has weight below one, so
$\omega=0$ by integrality.  Moreover, $m=0$ makes the noisy and ideal upward
sequences identical, and the deterministic reverse sequences therefore give
$f=\widehat f$.  Equation \eqref{eq:ideal-chain-boundary} now gives
\eqref{eq:chain-decoder-exact}.  The initial local explanations have total
support $O(|e|_R)$, the exact top cochain correction has absolute weight
$O(|\widetilde c|_R)$, and the remaining constant number of maps have bounded
support spread.  Hence the final correction has absolute weight
$O(|e|_R)$.
\end{proof}

\subsubsection{Binary restriction and repetition}

The preceding argument is naturally stated in block support over the
coefficient field.  One face contains at most $B_\mathrm{bin}$ binary
coordinates.  Consequently
\begin{equation}
  |e|_{R,\mathrm{block}}
  \le |e|_{R,\mathrm{bin}}
  \le B_\mathrm{bin}|e|_{R,\mathrm{block}},
  \label{eq:block-binary-weight}
\end{equation}
and the same comparison holds for syndrome support.  Changing between a
field basis and its trace dual basis is a local linear map of constant size.  The
cochain and chain decoders therefore retain their depths and their
contraction form after the constants
$\alpha,\beta,\gamma_{\mathrm{dec}},\eta$ are changed
by constant factors.

We also record the decoder transfer through the repetition used in the
transversal physical $\CCZ$ circuit.  Let $\ell_\mathrm{rep}$ be its constant
length.  For a $Z$ error
$z=(z_{i,t})$, define the parity on the $i$th repetition block by
\begin{equation}
  (\mathsf Pz)_i=\sum_{t=1}^{\ell_\mathrm{rep}}z_{i,t}.
  \label{eq:repetition-parity}
\end{equation}
The lifted old $X$-check syndrome is exactly $H_X\mathsf Pz$.  Repetition
$Z$ stabilizers identify $z$ with $\mathsf Pz$ supported on the first
copies, and
\begin{equation}
  |z|_{R,\mathrm{rep}}=|\mathsf Pz|_{R,\mathrm{base}}.
  \label{eq:repetition-Z-weight}
\end{equation}
Thus the base decoder may be run on the old syndrome and its correction
placed on the first copies.

For an $X$ error, write $x_i^\mathrm{base}=x_{i,1}$ and let
\begin{equation}
  u_{i,b}=x_{i,1}+x_{i,b}+\mu_{i,b},
  \qquad 2\le b\le \ell_\mathrm{rep},
  \label{eq:repetition-X-syndrome}
\end{equation}
be the measured repetition checks, where $\mu_{i,b}$ is the error in that
repetition check outcome.  Apply these
bits to the nonfirst copies.  The remaining word is the repetition of
$x^\mathrm{base}$,
plus a vector supported on those nonfirst copies whose weight is at most
the number of faulty repetition check outcomes.  Independently decode the
old $Z$-check syndrome $H_Zx^\mathrm{base}+m_\mathrm{old}$, where $m_\mathrm{old}$ is the error
in the record for the original checks.  If the base correction is $f$, define
$f_\mathrm{rep}$ to be the correction consisting of the bits $u_{i,b}$ on the
nonfirst copies together with $f$ repeated on all copies.  Put
\begin{equation}
 |m_\mathrm{rep}|=|m_\mathrm{old}|+
   \sum_{i}\sum_{b=2}^{\ell_\mathrm{rep}}\mu_{i,b}.
 \label{eq:repetition-record-weight}
\end{equation}
Projection to the first copies sends every repeated
$X$ stabilizer to its base stabilizer and gives
\begin{equation}
  |x+f_\mathrm{rep}|_R
  \le \ell_\mathrm{rep}\eta |x|_R
     +\ell_\mathrm{rep}\gamma_{\mathrm{dec}}|m_\mathrm{old}|
     +\sum_i\sum_{b=2}^{\ell_\mathrm{rep}}\mu_{i,b}
  \le \ell_\mathrm{rep}\eta |x|_R
     +(\ell_\mathrm{rep}\gamma_{\mathrm{dec}}+1)|m_\mathrm{rep}|.
  \label{eq:repetition-X-bound}
\end{equation}

It follows that one common set of parameters after repetition is
\begin{equation}
  \alpha'=\max\{\alpha,1\},\qquad
  \beta'=\beta/\ell_\mathrm{rep},\qquad
  \gamma'_{\mathrm{dec}}=\ell_\mathrm{rep}\gamma_{\mathrm{dec}}+1,\qquad
  \eta'=\ell_\mathrm{rep}\eta.
  \label{eq:repetition-decoder-parameters}
\end{equation}
Choose the base contraction with $\eta\le1/(8\ell_\mathrm{rep})$.  All
additional syndrome processing and Pauli layers have constant depth and
bounded spread.  Thus the same asymptotic decoder bounds hold after
repetition.

For an exact syndrome, choose a representative of minimum weight modulo
stabilizers for the
repeated error.  The correction on the nonfirst copies then has weight at
most
\begin{equation}
  (\ell_\mathrm{rep}-1)|x|_R.
  \label{eq:repetition-local-correction-weight}
\end{equation}
If the exact base decoder returns a correction of weight at most
$A|x|_R$, its repetition has weight at most
$\ell_\mathrm{rep}A|x|_R$.  Hence the complete correction after repetition has
absolute weight $O(|x|_R)$, with a constant depending only on the
repetition length.

\begin{proof}[Proof of \Cref{lem:decoder-interface}]
For one CSS component of a primal block, apply
\Cref{lem:parallel-cochain} in cochain degree two.  For the other component,
apply \Cref{lem:chain-reduction}, whose top decoder is the degree four
cochain decoder for the canonical kernel diagram of the code.  Physical $H$
exchanges the two components.  The binary restriction and repetition
calculation above transfers both statements to the physical binary codes.

For the six codes in $\mathfrak C_d$, taking the minimum of the radii and the
maximum of the coefficients gives constants
$\alpha,\beta,\gamma_{\mathrm{dec}}>0$.  For any $\eta>0$, choose the
cochain contractions sufficiently small before applying the constant
conversion factors.  This gives
\[
  |e+\mathcal D(He+m)|_R
  \le \eta|e|_R+\gamma_{\mathrm{dec}}|m|
\]
under the stated promise.  With $m=0$, repeating the contraction for
$O(\log n)$ rounds makes the reduced residual vanish by integrality.  The
sum of the correction supports is a geometric series, so the resulting
correction has absolute weight $O(|e|_R)$ and differs from $e$ by an element
of $\im G$.
\end{proof}

The adjacent degrees give the corresponding operation on noisy syndrome
records.  Let $H$ be either physical syndrome map and let $T$ be its next
check map, so that $TH=0$.

\begin{lemma}
\label{lem:syndrome-record-repair}
There are constants $a_R,C_R>0$ and a deterministic circuit
$\mathcal R_\mathrm{syn}$ of
$O(\log n)$ depth such that, whenever $s\in\operatorname{im}H$ and
$|v|\le a_Rn$,
\begin{equation}
  \mathcal R_\mathrm{syn}(s+v)=s+\mathcal R_\mathrm{syn}(v),
  \qquad \mathcal R_\mathrm{syn}(v)\in\operatorname{im}H,
  \qquad
  |\mathcal R_\mathrm{syn}(v)|+|\mathcal R_\mathrm{syn}(v)+v|
  \le C_R|v|.
  \label{eq:syndrome-record-repair}
\end{equation}
Moreover, there is a representative $h(v)$ with
\begin{equation}
  Hh(v)=\mathcal R_\mathrm{syn}(v),
  \qquad |h(v)|\le C_R|v|.
  \label{eq:syndrome-record-short-filling}
\end{equation}
The circuit is defined on every input word and sends zero to zero.
\end{lemma}

\begin{proof}
Apply the exact decoder in the adjacent degree to the syndrome $Tv$.  For a
cochain syndrome this is the degree three coefficient diagram; for a chain
syndrome it is the degree five canonical kernel diagram.  The decoder is
defined to return zero on zero input.  The
absolute correction bound gives a constant $C_\mathrm{adj}$ such that the output word
$f(Tv)$ satisfies
\begin{equation}
  |f(Tv)|\le C_\mathrm{adj}|v|,
  \qquad T(v+f(Tv))=0.
  \label{eq:adjacent-degree-correction}
\end{equation}
Let $a_\mathrm{dec}>0$ be a common input radius for the exact decoders in the
adjacent degrees, and let $\Delta_T$ bound the column weight of the next check maps.
Choose $a_R>0$ so that
\begin{equation}
  \Delta_Ta_R<a_\mathrm{dec},
  \qquad
  (1+C_\mathrm{adj})a_R<a_T,
  \label{eq:record-repair-radii}
\end{equation}
where $a_T$ is the exactness radius in \Cref{thm:qltc-family}.  The first
inequality places $Tv$ in the decoder promise, since
$|Tv|\le\Delta_T|v|$.  The second gives
$|v+f(Tv)|<a_Tn$, and the adjacent degree statement in that theorem gives
$v+f(Tv)\in\operatorname{im}H$.  Define
\begin{equation}
  \mathcal R_\mathrm{syn}(y)=y+f(Ty),
  \label{eq:syndrome-record-repair-map}
\end{equation}
extending the decoder deterministically outside its promise and setting its
output to zero there.
Since $Ts=0$, the decoder receives the same word on $s+v$ and on $v$, which
proves the first identity.  The remaining statements follow from
$v+f(Tv)\in\operatorname{im}H$ and the absolute correction bound.  Applying
the preimage circuit in \Cref{thm:qltc-family} to
$\mathcal R_\mathrm{syn}(v)$ gives $h(v)$ in
\eqref{eq:syndrome-record-short-filling}, with $h(0)=0$ and
$O(\log n)$ depth.  Binary
restriction, the physical Hadamard dual, and the constant length repetition
change only the constants.
\end{proof}

Choose the independent check rows and canonical logical representatives
in \eqref{eq:canonical-checks}--\eqref{eq:canonical-logicals}.  For one CSS
component, let $J$ be the linear section obtained by filling the pivot
coordinates, and let $L(\nu)$ be the chosen physical representative of a
logical Pauli label $\nu$.  If $\lambda(P)$ is the logical label of a
physical Pauli $P$, then, modulo a stabilizer and phase,
\begin{equation}
  P=J(HP)+L(\lambda(P)).
  \label{eq:canonical-pauli-decomposition}
\end{equation}
The same notation denotes the direct sum of the two CSS components.  The
maps $J$, $L$, and $\lambda$ are binary linear maps determined by the code
member.  Copy and XOR trees evaluate them in $O(\log n)$ depth with bounded
fan-in and bounded fan-out, followed by one physical Pauli layer.
Choose total linear extensions of $J$, $L$, and $\lambda$ to the complete
binary record spaces.  Their feedback circuits are then defined on every
record, while \eqref{eq:canonical-pauli-decomposition} is invoked when the
repaired syndrome lies in $\operatorname{im}H$.  Componentwise, the repair
map and syndrome section are the direct sums of their $X$- and $Z$-component
maps, with the same decomposition for the logical label and representative
maps.  Physical
Hadamard exchanges the two components; binary restriction and the constant
length repetition conjugate the maps by their encoding circuits.

Canonical code switching, Clifford segments on complete blocks, and direct
zero and plus preparation have a common form before correction.  On a good support
and for the actual complete record, the gadget proof determines a total Pauli
correction $C$ and an output $Ce\mathcal EU\ket\psi$.  Let
\begin{equation}
  t_\mathrm{core}
  =\max\{t_\mathrm{sw},t_\mathrm{Cliff},t_{0/+}\}
  \label{eq:raw-core-radius}
\end{equation}
be a common residual radius for these segments.  Complete readout obeys the
analogous recordwise identity for its corrected outcome and has no surviving
quantum output.  Since the list of segments is finite, the boundary and
allowed fault fractions can be chosen so that
$|e|_R\le t_\mathrm{core}n$ in every case and
$t_\mathrm{core}<\beta_\mathrm{full}$.  These raw identities precede the deferred
computation below.  In particular, canonical code switching follows from
\eqref{eq:bcs-raw-bad-family} and
\eqref{eq:raw-switch-output}.

\begin{lemma}[Deferred Clifford corrections]
\label{lem:deferred-clifford-corrections}
Consider $N_\mathrm{seg}$ consecutive segments from the preceding list.  Append a fresh
complete syndrome extraction to every surviving
output of each segment.  The quantum cores can be executed without waiting
for the decoder computation of the preceding segment.  If $V$ is the total
number of local records and logical labels, all syndrome corrections and
logical Pauli frames can be computed after the last quantum core and then
applied physically in
\begin{equation}
  O\bigl(N_\mathrm{seg}+\log^c(n+2)+\log(V+2)\bigr)
  \label{eq:deferred-clifford-time}
\end{equation}
physical time, for a constant $c$.  Every surviving block remains within the
boundary radius after the final correction, outside the union of the bad set
families assigned to those segments and intervening memory windows.  The statement is
recordwise and holds with arbitrary reference systems.
\end{lemma}

\begin{proof}
Consider one CSS component of one segment.  Let $C$ be the nominal physical
Pauli correction determined by the records of its quantum core.  On a Pauli
branch, write the output before correction as
\begin{equation}
  P C e\,\mathcal E U\ket\psi,
  \qquad
  a=HC+He+v,
  \label{eq:raw-clifford-output}
\end{equation}
where $e$ is the residual entering the fresh extraction, while $P$ and $v$
are the data and record errors produced by that extraction.  Put
\begin{equation}
  \widehat a=\mathcal R_\mathrm{syn}(a),
  \qquad
  g=\mathcal D_\mathrm{exact}(\widehat a+HC),
  \qquad
  \nu=\lambda(C+g),
  \label{eq:deferred-normalization-data}
\end{equation}
where $\mathcal R_\mathrm{syn}$ is the map in
\Cref{lem:syndrome-record-repair} and $\mathcal D_\mathrm{exact}$ is the exact
decoder in \Cref{lem:decoder-interface}.
Translation covariance gives
\begin{equation}
  \widehat a+HC=He+\mathcal R_\mathrm{syn}(v)=H(e+h(v)).
  \label{eq:deferred-normalization-syndrome}
\end{equation}
Let $\beta_\mathrm{full}n$ be a common radius for the exact decoders.  There are
constants $c_P,c_v$ such that an appended extraction containing at most
$a_\mathrm{ext}n$ faults gives
\(
  |P|_R\le c_Pa_\mathrm{ext}n
\)
and
\(
  |v|\le c_va_\mathrm{ext}n
\).
Choose the common boundary and $a_\mathrm{ext}>0$ so that
\begin{equation}
  c_va_\mathrm{ext}<a_R,
  \qquad
  t_\mathrm{core}+C_Rc_va_\mathrm{ext}<\beta_\mathrm{full},
  \qquad
  (c_P+C_Rc_v)a_\mathrm{ext}<\frac{t_*}{2}.
  \label{eq:deferred-common-radii}
\end{equation}
The finite list of segment types permits one choice of these constants.  The
middle inequality places $e+h(v)$ inside the exact decoder radius.  The exact decoder then
gives $e+g=h(v)$ modulo stabilizers.  Moreover,
$H(C+g)=\widehat a$, and \eqref{eq:canonical-pauli-decomposition} gives
\begin{equation}
  P C e\,\mathcal E U\ket\psi
  =[P+h(v)]J(\widehat a)L(\nu)\mathcal E U\ket\psi
  \label{eq:deferred-normalization-identity}
\end{equation}
up to phase.  The residual $P+h(v)$ is supported in a bounded neighborhood
of the fresh extraction faults and lies inside the boundary radius by
\eqref{eq:deferred-common-radii}.  The same calculation on matrix units proves
the identity for arbitrary supported fault maps and external references.

Before the next segment, the Pauli $J(\widehat a)L(\nu)$ is removed virtually.  Its
propagation through the compiled Clifford circuit gives linear shifts of
the later measurement records and a known output Pauli $A$.  The local
decoder is run on these shifted records and outputs $C_\mathrm{seg}$; the
nominal correction in \eqref{eq:raw-clifford-output} is $C=C_\mathrm{seg}+A$.
Thus the decoder computation for a segment depends only on the raw syndrome
baselines at its input and on its own
records, not on the completed decoder output of the previous segment.  Raw
zero and plus ancillas have zero logical frame by their direct preparation
identity.  During every wait, each block uses its fresh raw syndrome as the
unchanged baseline in \Cref{lem:sector-storage}.  Relative to the sector
label $\widehat a$, that baseline has error
\begin{equation}
  a+\widehat a=v+\mathcal R_\mathrm{syn}(v),
  \qquad |a+\widehat a|\le C_R|v|.
  \label{eq:deferred-baseline-error}
\end{equation}
After choosing the memory fault fraction $a_\mathrm{mem}>0$ so that
\begin{equation}
  \theta t_*+C_\mathrm{sec}C_Rc_va_\mathrm{ext}
    +C_\mathrm{mem}a_\mathrm{mem}<t_*,
  \label{eq:deferred-memory-radius}
\end{equation}
the stored sector recurrence remains inside the same boundary throughout
the wait.  Here $\theta,C_\mathrm{sec},C_\mathrm{mem}$ are the common constants in
\Cref{lem:sector-storage}.  Virtual Pauli removal is an equality of branch
maps for the original physical instrument and therefore preserves the
assigned fault support.

For a classically enabled sequence of coordinate operations, group the constant number of
Clifford segments acting at coordinate $u$ into one template.  Let
$M_u(e_u)$ be its symplectic action, equal to the identity when the enable
$e_u$ vanishes, and let $\nu_u$ be its additive Pauli label.  If
$\nu_{w,u}$ denotes the component of $\nu_w$ on coordinate $u$, then the
final logical frame is
\begin{equation}
  f_{\mathrm{out},u}
  =M_u(e_u)\left(f_{\mathrm{in},u}+\sum_{w<u}\nu_{w,u}\right)
   +\sum_{w\ge u}\nu_{w,u}.
  \label{eq:enabled-frame-recurrence}
\end{equation}
This formula includes terms with nonzero components on many logical coordinates, but it
requires only balanced XOR and copy trees, not a serial product of dense
symplectic matrices.  A constant number of intervening compiled routes does
not change the depth bound.  After the local decoder computations finish,
evaluate the corresponding $J(\widehat a)$ for each surviving block together
with $L(f_\mathrm{out})$, apply them in one Pauli layer, and run the
terminal error correction cycle.  A complete logical readout is handled by
subtracting the appropriate component of $J(\widehat a)$ from the raw word and then
applying the logical frame shift to the decoded outcome.  The branch
identities intertwine the corrected logical projectors with the ideal
readout projectors.  Trace preservation and
\Cref{lem:instrument-normalization} therefore give the complete measurement
instrument, including its outcome probabilities.  The $N_\mathrm{seg}$ quantum
cores, the parallel local decoder computations, and these balanced linear circuits give
\eqref{eq:deferred-clifford-time}.
Overlapping bad set families are combined by the sum operation in
\eqref{eq:bad-family-sum}; the product operation is used only for disjoint
sets of physical locations.
\end{proof}

\subsubsection{Error correction and memory}

\begin{proof}[Proof of \Cref{lem:ec-gadget}]
Measure each $X$ check with a fresh plus ancilla, CNOTs between the ancilla
and the support of the check, and an $X$ measurement; use the conjugate
circuit for the $Z$ checks.  An edge coloring of the Tanner graphs gives
constant quantum depth.  The extraction ancillas can be reused, and one
complete extraction, classical decoding, and Pauli correction window has
$O(n)$ qubits and $O(n)$ locations, including all idles during the classical
layers.

Condition on a Pauli fault path with $f$ faulty locations.  Propagation
through the bounded light cones represents these faults by a final data Pauli
and two errors in the syndrome record whose total
weight is at most $hf$, for a constant $h$.  If
$t_\mathrm{in}$ is the reduced weight of the incoming Pauli error, applying
\Cref{lem:decoder-interface} to the two CSS components and including faults
in the final Pauli correction gives
\begin{equation}
  t_\mathrm{out}
  \le 2\eta t_\mathrm{in}
    +(2\gamma_{\mathrm{dec}}+1)hf.
  \label{eq:physical-memory-recurrence}
\end{equation}

Let $s_\mathrm{spr}$ be the largest support spread among the physical gate
segments used between two error correction windows.  These segments preserve
the codespace and act on at most a constant number of blocks.  For a segment
acting on several blocks, take the maximum reduced support spread into any
one output block.  Choose constants
$t_*>0$ and $c_\mathrm{fault}>0$ such that
\begin{align}
 s_\mathrm{spr}t_*+\alpha h c_\mathrm{fault}&<\beta,
 \label{eq:physical-memory-promise}\\
 2\eta s_\mathrm{spr}t_*
 +(2\gamma_{\mathrm{dec}}+1)h c_\mathrm{fault}&\le t_*,
 \label{eq:physical-memory-closure}\\
 2t_*&<\delta.
 \label{eq:physical-memory-distance}
\end{align}
The contraction $\eta$ is chosen first to compensate the support spread, and
then $t_*$ and $c_\mathrm{fault}$ are decreased so that these inequalities
hold.  Consequently, whenever the input reduced weight is at most $t_*n$ and
the window contains at most $c_\mathrm{fault}n$ faults, the output reduced weight
is also at most $t_*n$.

If a window has at most $C_\mathrm{loc}n$ locations and the faulty support is
$p$-locally stochastic, then
\begin{equation}
 \Pr[\text{window is bad}]
 \le
 \binom{C_\mathrm{loc}n}{\lceil c_\mathrm{fault}n\rceil}
 p^{\lceil c_\mathrm{fault}n\rceil}
 \le
 \left(\frac{C p}{c_\mathrm{fault}}\right)^{c_\mathrm{fault}n}
 \le e^{-cn}
 \label{eq:physical-memory-failure}
\end{equation}
below a positive threshold.

For arbitrary CPTP faults, condition on the faulty support and every
measurement record, and expand the supported operations on the left and right in the
Pauli basis.  The same correction is applied to all terms with the same
record, and \eqref{eq:physical-memory-recurrence} bounds the support on both
sides.  Since \eqref{eq:physical-memory-distance} makes all such Pauli errors
jointly correctable, ideal recovery gives the identity logical channel on
the encoded state and an arbitrary reference system.  The same argument
covers the finite list of physical gate segments used in the simulation,
including the transversal physical $\CCZ$ circuit, because conjugation
enlarges support only within a bounded light cone.  The bad family is defined
by the original faulty support, so no independence assumption or stochastic
statement about a stochastic output error is used.
\end{proof}

\begin{proof}[Proof of \Cref{lem:sector-storage}]
Let $s_0$ be the syndrome of the reference state defining this sector, and
write the first extraction record as $\widehat s_0=s_0+v_0$.  The word $v_0$ includes
the syndrome of the initial small Pauli and the faults in this first
syndrome record; data faults created during that extraction are included in
$t_0$.  Keep $\widehat s_0$ unchanged and subtract it from every later measured
syndrome.  Relative to that reference state, the decoder input consists of the
syndrome of the current error, the same offset $v_0$, and the fresh faults in
the current window.  Equation \eqref{eq:physical-memory-recurrence}, with a
constant contraction $0<\theta<1$, gives
\[
  t_{j+1}\le\theta t_j+C_\mathrm{sec}|v_0|+C_\mathrm{mem}f_j.
\]
Therefore
\[
  t_j\le\theta^jt_0+
  \frac{C_\mathrm{sec}|v_0|+C_\mathrm{mem}\max_i f_i}{1-\theta}.
\]
If $|v_0|/n$ and $f_i/n$ are sufficiently small constants, the state
continues to satisfy \eqref{eq:typed-boundary-map} for any prescribed
polynomial number of cycles.  The same baseline is used throughout; its syndrome is
neither identified nor filled.  Every code stabilizer acts as a scalar on
the same syndrome sector, so replacing a Pauli by a representative of
minimum weight modulo stabilizers does not change the branch, including when the block is
entangled with a reference system.  The recordwise Pauli expansion in the
proof of \Cref{lem:ec-gadget} therefore gives the statement for arbitrary
CPTP faults and an arbitrary reference system.  If the preparation estimate
shows that the delayed correction places the branch inside the ordinary
memory input radius, the subsequent error correction cycle produces an
output with reduced error at most $t_*n$.
\end{proof}

\subsection{Usable code lengths}
\label{sec:bounded-code-lengths}

The original nested tower gives unbounded cover degree and increasing
injectivity radius, but its available block lengths may be too sparse.  We
therefore select congruence covers of the same arithmetic complex in every
sufficiently large degree, so that the first usable length above a target is
within a constant factor of that target.

\begin{lemma}[Usable code lengths]
\label{lem:usable-code-lengths}
Let $Q$ be the even prime power used in the arithmetic construction.  For
every sufficiently large integer $d$, there is an odd congruence cover in
the same construction whose combined binary length is
\begin{equation}
 n_d=n_*(Q^{2d}-1),
 \label{eq:usable-length-construction}
\end{equation}
where $n_*$ is a constant.  The injectivity radii of these covers tend to infinity
with $d$, and all code and cohomological properties in
\Cref{thm:qltc-family} hold for all sufficiently large members.
\end{lemma}

\begin{proof}
Let $F=\mathbb F_Q(t)$, let $D$ be the quaternion division algebra of
the construction, and put $H_\mathrm{ar}=\operatorname{SL}_1(D)$.  If
$S=\{\mathfrak p_1,\ldots,\mathfrak p_6\}$ is the set of six split places
defining the cubical directions, the arithmetic subgroup
that preserves types has the exact description
\begin{equation}
 \Gamma=H_\mathrm{ar}(F)\cap K^S,
 \label{eq:arithmetic-group-description}
\end{equation}
where $K^S$ denotes the prescribed compact open conditions away from $S$.
The six split places have degrees
\begin{equation}
  \nu_i=\nu_0+2(i-1),
  \qquad i\in[6],
  \qquad \nu_0\ \text{odd},
  \label{eq:walking-degrees}
\end{equation}
so the first two degrees are coprime.  This replaces an argument using a
place of degree one,
which is not available for the actual construction.

The group $H_\mathrm{ar}$ is simply connected and absolutely almost simple of
type $A_1$, and it is split, hence noncompact, at every place in $S$.  To
apply strong approximation in the exact subgroup
\eqref{eq:arithmetic-group-description}, first note that scalar weak
approximation holds with the original coordinate $t$ at a split place of
arbitrary degree: clear the finitely many prescribed denominators, impose
the jets at finite places by polynomial Chinese remaindering, and use the high
degree coefficients to impose the desired jet at infinity.  In
characteristic two, an element of reduced norm one in the division quaternion
algebra with reduced trace zero is the identity; elements of reduced norm one
and nonzero trace are
therefore dense in every constrained local factor.  The trace polynomial
$X^2-uX+1$ is separable for $u\ne0$, and local conjugacy together with weak
approximation in $D^\times$ removes the auxiliary adelic conjugation.

The compact representative argument also uses the actual walking degrees.
If an off-$S$ reduced norm has absolute value $Q^z$, choose integers
\begin{equation}
  b_1=z\frac{\nu_0+1}{2},
  \qquad
  b_2=-z\frac{\nu_0-1}{2},
  \qquad
  \nu_0b_1+(\nu_0+2)b_2=z.
  \label{eq:two-place-norm-compensation}
\end{equation}
Diagonal elements at the first two split places cancel this norm.
Thus the projection of the adelic group with product norm one onto the
restricted product over the places outside $S$ is surjective.  Strong approximation over $F$
\cite[Theorem~A]{Prasad1977} now applies to the group and the compact open
conditions occurring in \eqref{eq:arithmetic-group-description}.  For every
fresh split finite place $w\notin S$, the reduction map
\begin{equation}
 \rho_w:\Gamma\longrightarrow\operatorname{SL}_2(k_w)
 \label{eq:arithmetic-reduction-map}
\end{equation}
is onto.  Indeed, one prescribes an arbitrary residue class at $w$ and keeps
the compact open conditions defining $K^S$ at the remaining places.

For each sufficiently large $d$, choose an irreducible place $w_d$ of degree
$d$, outside the finite set of excluded places.  Then
$k_{w_d}\cong\mathbb F_{Q^d}$.  Let
\begin{equation}
 U_d=
 \left\{
 \begin{pmatrix}
  1&u\\
  0&1
 \end{pmatrix}:u\in k_{w_d}
 \right\}.
 \label{eq:unipotent-subgroup}
\end{equation}
Since $Q$ is even,
\begin{equation}
 |\operatorname{SL}_2(k_{w_d})|
 =Q^d(Q^{2d}-1),
 \qquad
 |U_d|=Q^d.
 \label{eq:SL2-order}
\end{equation}
Thus $U_d$ is a Sylow $2$-subgroup.  Surjectivity of
\eqref{eq:arithmetic-reduction-map} shows that
\begin{equation}
 \Gamma(d)=\rho_{w_d}^{-1}(U_d)
 \qquad\text{satisfies}\qquad
 [\Gamma:\Gamma(d)]=Q^{2d}-1.
 \label{eq:dense-cover-index}
\end{equation}
In particular, every cover has odd degree.  The groups $\Gamma(d)$ need not
be nested.

The injectivity radii tend to infinity.  For any prescribed radius, only
finitely many nonidentity elements displace one of the finitely many
representatives of the base orbits by at most that radius.  Their reduced traces are nonzero.
After $d$ exceeds the degrees of the irreducible factors in the numerators
and denominators of these traces, reduction at $w_d$ keeps every such trace
nonzero.  Every element of $U_d$, and of each conjugate of $U_d$, has trace
zero in characteristic two.  The usual conjugacy argument therefore
excludes all such short elements from $\Gamma(d)$.

The spectral statement for congruence covers, the theorem giving positive
rate for the primary code, and the pullback--transfer argument for covers of odd degree in
Ref.~\cite{GoodQLTCTransversal2026} apply to these sufficiently deep
covers.  Hence the local code identifications used by the decoder persist,
as do the distance and soundness bounds.  The primary rate remains positive,
and the cohomological invariant form still gives the same transversal logical
$\CCZ$ after the repetition.  All relevant face and stalk dimensions
scale with the cover degree; the binary restriction and repetition contribute
only constant factors.  If $n_*$ is the combined binary length over the base complex,
\eqref{eq:dense-cover-index} gives \eqref{eq:usable-length-construction}.
\end{proof}

\begin{corollary}
\label{cor:bounded-length-ratios}
For every sufficiently large $d$,
\begin{equation}
  \frac{n_{d+1}}{n_d}
  =Q^2+\frac{Q^2-1}{Q^{2d}-1}
  \le Q^2+1.
  \label{eq:dense-ratio}
\end{equation}
Consequently, for every sufficiently large $x$, the first usable length
$n_d\ge x$ satisfies
\begin{equation}
  x\le n_d\le(Q^2+1)x.
  \label{eq:length-selection}
\end{equation}
\end{corollary}

\begin{proof}
The first identity follows from
\[
  \frac{Q^{2d+2}-1}{Q^{2d}-1}
  =Q^2+\frac{Q^2-1}{Q^{2d}-1}.
\]
If $d$ is the first index with $n_d\ge x$, then $n_{d-1}<x$, and
\eqref{eq:dense-ratio} gives $n_d\le(Q^2+1)x$.
\end{proof}

The same sequence of code lengths is equipped with the maps in
\eqref{eq:subrank-injections}--\eqref{eq:subrank-diagonal}.  Their uniformity
on this sequence is what
allows the compiler to choose the first usable length above its target
scale.  Equations
\eqref{eq:exact-dense-length}--\eqref{eq:length-selection} prove
\eqref{eq:dense-lengths-main}.

%% file: sections/injection-route.tex
\section{An alternative scheme based on $\CCZ$-state injection}
\label{sec:injection-route}

The preceding construction obtains its non-Clifford gates directly from the
transversal action of the computational qLTC.  The alternative construction
uses the same memory, readout, and canonical code switching operations, but
implements non-Clifford gates by injecting encoded $\ket{\CCZ}$ states.  Its
finite distillation circuit is the punctured
quantum Reed--Solomon construction of Ref.~\cite{NguyenPattison2025}.  We fix
one member of that construction and increase only the lengths of the qLTC
blocks protecting it.  Code switching between consecutive lengths is defined
on arbitrary source states, and a factory at protection length $m$ takes
time $O(m)$.

The growing distillation code in that construction has a yield that
contributes a subpolylogarithmic loss after
scheduling.  Here the target accuracy instead comes from recursively applying
the same finite circuit to qLTC blocks of geometrically increasing length.
Only a constant number of rounds of child factories is needed at each level.
The resulting recurrence has total time $O(m)$ and failure probability
$\exp(-\Omega(m))$.  Taking $m=\Theta(\log(WD/\varepsilon))$ gives a strictly
logarithmic time bound.  The Clifford masks and dense routing are constructed
on the complete logical space of the primary code, independently of the
active coordinates used in the main construction.

Throughout this section, $G(m)$ denotes the primary computational code at a
usable block length $m$, with
\begin{equation}
  \kappa m\le k(m)\le m,
  \qquad d_X(m),d_Z(m)\ge\delta m.
  \label{eq:inj-code-parameters}
\end{equation}
Let $m_\mathrm{next}(x)$ be the first usable length at least $x$.  The dense
family gives a constant $R_\mathrm{gap}$ such that
\begin{equation}
  x\le m_\mathrm{next}(x)\le R_\mathrm{gap}x
  \label{eq:inj-usable-gap}
\end{equation}
for all sufficiently large $x$.  We use the full zero and plus preparation,
complete logical readout, canonical code switching, and memory constructions
already proved.  The Clifford masks and dense routing required below are
constructed in \Cref{subsec:injection-clifford-routing} on the complete
logical space of $G(m)$ and its physical $H$ dual.

\subsection{Finite Reed--Solomon distillation}
\label{subsec:finite-qrs}

Let $q=2^\ell$ with $\ell\ge3$, choose a subset
$B_\mathrm{RS}\subseteq\mathbb F_q$ of size $s=q/4$, and put
$\Omega_\mathrm{RS}=\mathbb F_q\setminus B_\mathrm{RS}$, so $n_\mathrm{RS}=3q/4$.
With $r_\mathrm{RS}=\lfloor q/3\rfloor$, define evaluation codes on
$\Omega_\mathrm{RS}$ by
\begin{align}
 C_1&=\{(p(a))_{a\in\Omega_\mathrm{RS}}:\deg p<r_\mathrm{RS}\},
 \label{eq:qrs-C1}\\
 C_X&=\{(p(a))_{a\in\Omega_\mathrm{RS}}:\deg p<r_\mathrm{RS},
                         \ p|_{B_\mathrm{RS}}=0\},
 \label{eq:qrs-CX}\\
 C_2&=C_X^\perp.
 \label{eq:qrs-C2}
\end{align}
Write $\operatorname{Tr}:\mathbb F_q\to\mathbb F_2$ for the absolute trace,
and define
\begin{equation}
  \CCZ^{(q)}\ket{x,y,z}
  =(-1)^{\operatorname{Tr}(xyz)}\ket{x,y,z},
  \qquad
  \ket{+_q}=q^{-1/2}\sum_{x\in\mathbb F_q}\ket x.
  \label{eq:qary-ccz-plus}
\end{equation}

\begin{lemma}[Punctured quantum Reed--Solomon code]
\label{lem:inj-qrs-code}
The CSS code defined by $C_2^\perp\subseteq C_1$ has parameters
\begin{equation}
  \left[\left[\frac{3q}{4},\frac q4,
  \left\lfloor\frac q{12}\right\rfloor+1\right]\right]_q.
  \label{eq:qrs-parameters}
\end{equation}
On three blocks, coordinatewise physical $\CCZ^{(q)}$ implements $q/4$
coordinatewise logical $\CCZ^{(q)}$ gates.
\end{lemma}

\begin{proof}
Evaluation on $\Omega_\mathrm{RS}$ is injective in the stated degree range, and
vanishing on the $s$ points of $B_\mathrm{RS}$ imposes $s$ independent
conditions.  Hence
\begin{equation}
  \dim C_1=r_\mathrm{RS},
  \qquad \dim C_X=r_\mathrm{RS}-s,
  \qquad \dim C_2=q-r_\mathrm{RS}.
  \label{eq:qrs-dimensions}
\end{equation}
To identify $C_2$, let $p$ have degree below $q-r_\mathrm{RS}$ and let the
polynomial representing $h\in C_X$ vanish on $B_\mathrm{RS}$.  Since
$\deg(ph)<q-1$,
\begin{equation}
  \sum_{a\in\Omega_\mathrm{RS}}p(a)h(a)
  =\sum_{a\in\mathbb F_q}p(a)h(a)=0.
  \label{eq:qrs-dual-inclusion}
\end{equation}
This gives the required inclusion between the evaluation spaces in $C_X^\perp$; equality
follows from \eqref{eq:qrs-dimensions}.  The two classical distances are
\begin{equation}
  d(C_1)=n_\mathrm{RS}-r_\mathrm{RS}+1,
  \qquad
  d(C_2)=r_\mathrm{RS}-s+1.
  \label{eq:qrs-classical-distances}
\end{equation}
A minimum word of $C_1$ can be chosen with all its roots in $\Omega_\mathrm{RS}$ and
is then nonzero on $B_\mathrm{RS}$, so it does not belong to $C_X$.  On the
dual side, the Vandermonde description gives
$d(C_1^\perp)=r_\mathrm{RS}+1$, whereas $C_2$ contains a word of weight
$r_\mathrm{RS}-s+1$.  The latter word cannot be a $Z$ stabilizer.  The quantum
distance is consequently
$r_\mathrm{RS}-s+1=\lfloor q/12\rfloor+1$.

Write
\begin{equation}
  R_{B_\mathrm{RS}}(X)=\prod_{b\in B_\mathrm{RS}}(X-b).
  \label{eq:qrs-vanishing-polynomial}
\end{equation}
Restriction to $B_\mathrm{RS}$ maps the polynomials of degree below
$r_\mathrm{RS}$ onto $\mathbb F_q^s$, with kernel
$R_{B_\mathrm{RS}}\mathbb F_q[X]_{<r_\mathrm{RS}-s}$.  For
$u\in\mathbb F_q^s$, let $p_u$ be the polynomial of degree below $s$ that
interpolates $u$ on $B_\mathrm{RS}$.  A canonical logical basis is therefore
\begin{equation}
  \ket{\bar u}
  =|C_X|^{-1/2}\sum_{h\in C_X}
    \ket{\operatorname{eval}_{\Omega_\mathrm{RS}}(p_u)+h}.
  \label{eq:qrs-logical-basis}
\end{equation}
This basis also gives a finite encoder without invoking a decoding oracle.
Take first the $s$ evaluation columns of the interpolation polynomials for
the standard basis of $\mathbb F_q^s$, followed by the $r_\mathrm{RS}-s$
columns
$\operatorname{eval}_{\Omega_\mathrm{RS}}(R_{B_\mathrm{RS}}X^j)$, and extend them to an
invertible $n_\mathrm{RS}\times n_\mathrm{RS}$ matrix over $\mathbb F_q$.  In the
chosen binary basis, Gaussian elimination implements the corresponding
linear permutation with CNOT and SWAP gates.  Acting on
$\ket u\ket{+_q}^{\otimes(r_\mathrm{RS}-s)}
 \ket{0_q}^{\otimes(n_\mathrm{RS}-r_\mathrm{RS})}$ gives
\eqref{eq:qrs-logical-basis}; the inverse circuit extracts the logical
register coherently, also in the presence of a reference system.

Logical values are the evaluations on $B_\mathrm{RS}$.  For three representing
polynomials, $\deg(p_1p_2p_3)\le3r_\mathrm{RS}-3<q-1$, so its sum over the
complete field vanishes.  Since the characteristic is two, the sums over
$B_\mathrm{RS}$ and $\Omega_\mathrm{RS}$ are equal:
\begin{equation}
  \sum_{a\in\Omega_\mathrm{RS}}p_1(a)p_2(a)p_3(a)
  =\sum_{b\in B_\mathrm{RS}}p_1(b)p_2(b)p_3(b).
  \label{eq:qrs-cubic-identity}
\end{equation}
Applying the field trace to the phase of coordinatewise $\CCZ^{(q)}$ turns the
right-hand side into
\begin{equation}
  (-1)^{\operatorname{Tr}\sum_{b\in B_\mathrm{RS}}
                 p_1(b)p_2(b)p_3(b)}
  =\prod_{b\in B_\mathrm{RS}}
       (-1)^{\operatorname{Tr}(p_1(b)p_2(b)p_3(b))}.
  \label{eq:qrs-logical-phase}
\end{equation}
This is the claimed product of $s$ coordinatewise logical $\CCZ^{(q)}$ gates.
\end{proof}

Write
\(
  d_\mathrm{RS}=\lfloor q/12\rfloor+1
\)
for the quantum distance.  The code corrects arbitrary errors on at most
\begin{equation}
  t_\mathrm{RS}=\left\lfloor\frac{d_\mathrm{RS}-1}{2}\right\rfloor
             =\left\lfloor\frac q{24}\right\rfloor
  \label{eq:qrs-correction-radius}
\end{equation}
outer coordinates.  Its encoded plus state has a deterministic,
trace-preserving preparation.  All $n_\mathrm{RS}\ell$ binary constituents are
initialized in $\ket{+}$, after which independent generators of both CSS
check spaces are measured in a prescribed order.  The $X$-check outcomes are
trivial, and a chosen right inverse of the binary check matrix determines an
$X$ correction from the recorded $Z$ syndrome $\sigma$.  On every record of
nonzero probability, the corrected state is the uniform superposition over
$C_1$, namely $\ket{+_q}^{\otimes s}$ encoded in the basis
\eqref{eq:qrs-logical-basis}.  Defining the same syndrome section on
impossible records makes the instrument total.

The trace pairing $(x,y)\mapsto\operatorname{Tr}(xy)$ is nondegenerate and
nonalternating over $\mathbb F_2$, and therefore has an orthonormal binary
basis.  Choose one such basis
$(\beta_i)_{i=1}^{\ell}$ of $\mathbb F_q$.  Writing
$x=\sum_ix_i\beta_i$ and similarly for $y,z$, one has
\begin{equation}
  \operatorname{Tr}(xyz)
  =\sum_{i,j,k=1}^{\ell}
    \operatorname{Tr}(\beta_i\beta_j\beta_k)x_iy_jz_k.
  \label{eq:field-ccz-decomposition}
\end{equation}
We assign all $\ell^3$ binary slots to each outer coordinate.  A slot whose
coefficient in \eqref{eq:field-ccz-decomposition} vanishes is measured and
discarded at its prescribed position.  Thus one invocation has
\begin{equation}
  M_\mathrm{RS}=n_\mathrm{RS}\ell^3=\frac{3q}{4}\ell^3,
  \qquad K_\mathrm{RS}=s=\frac q4
  \label{eq:qrs-MK}
\end{equation}
binary input and output triples.

A binary resource triple is consumed by CNOTs from data to resource,
followed by $Z$ readout of the resource.  For outcome
$z=(z_1,z_2,z_3)$, the correction is
\begin{equation}
  C_z=
  \CZ_{23}^{z_1}\CZ_{13}^{z_2}\CZ_{12}^{z_3}
  Z_3^{z_1z_2}Z_2^{z_1z_3}Z_1^{z_2z_3}.
  \label{eq:ccz-consumption-correction}
\end{equation}
The corrected branch is a scalar depending on the record times $\CCZ$ on arbitrary
data and reference states.  In the field circuit, a bad binary input triple
therefore inserts an arbitrary operator on only its assigned outer
coordinate.  Several bad inputs, even when their states and faults are
correlated, produce an operator in the span supported on the corresponding
set of outer coordinates.

After all resource injections, the two CSS check spaces are measured on each
of the three qRS blocks.  The resulting recovery corrects the entire operator
span supported on at most $t_\mathrm{RS}$ outer coordinates.  To define it on
every record, let $Q_\sigma$ be the first qudit Pauli of minimum support
$|\operatorname{supp}(a)\cup\operatorname{supp}(b)|$ with syndrome $\sigma$,
in the prescribed lexicographic order.  Every syndrome has such a
representative.  If a Pauli $E$ has weight at most $t_\mathrm{RS}$ and produces
$\sigma$, then
$\operatorname{wt}(Q_\sigma)\le t_\mathrm{RS}$ and
\begin{equation}
  \operatorname{wt}(Q_\sigma^{-1}E)
  \le 2t_\mathrm{RS}<d_\mathrm{RS}.
  \label{eq:qrs-recovery-weight}
\end{equation}
The operator $Q_\sigma^{-1}E$ has trivial syndrome and is consequently a
stabilizer on the code space, up to phase.  If $\Pi_\mathrm{code}$ is the code
projector and $\Pi_\sigma$ the syndrome projector, then
\begin{equation}
  Q_\sigma^{-1}\Pi_\sigma E\Pi_\mathrm{code}
  =\lambda_E\Pi_\mathrm{code},
  \qquad
  \Pi_\tau E\Pi_\mathrm{code}=0\quad(\tau\ne\sigma).
  \label{eq:qrs-correctable-span}
\end{equation}
Applying the same identity to matrix units and retaining the syndrome record
proves the statement for the full operator span, including correlations with a reference.
The tensor product of the three block recoveries therefore corrects any set
of at most $t_\mathrm{RS}$ bad binary resource triples, even when the affected
outer coordinates differ among the three blocks.  On arbitrary inputs the
corrected output lies in the code space, so inverse encoding and the
prescribed discards are well defined.

After qudit distillation, apply a predetermined binary linear change of
basis to each logical qudit.  In the first two registers the first basis
vector is $1$, while in the third register it is an element
$\theta\in\mathbb F_q$ with $\operatorname{Tr}(\theta)=1$.  Write
\begin{equation}
  x=a+X,\qquad y=b+Y,\qquad z=c\theta+Z,
  \label{eq:qrs-binary-splitting}
\end{equation}
where $a,b,c\in\mathbb F_2$ and $X,Y,Z$ belong to the complementary binary
subspaces.  Expansion of the trace gives
\begin{equation}
  \operatorname{Tr}(xyz)
  =abc+Aab+Bac+Cbc+Da+Eb+Fc+H,
  \label{eq:qrs-output-polynomial}
\end{equation}
where
\begin{equation}
\begin{aligned}
 A&=\operatorname{Tr}(Z),
 &B&=\operatorname{Tr}(\theta Y),
 &C&=\operatorname{Tr}(\theta X),\\
 D&=\operatorname{Tr}(YZ),
 &E&=\operatorname{Tr}(XZ),
 &F&=\operatorname{Tr}(\theta XY),
 &H&=\operatorname{Tr}(XYZ).
\end{aligned}
\label{eq:qrs-output-coefficients}
\end{equation}
Measuring the other $3(\ell-1)$ bits determines these coefficients.  At most
three $\CZ$ and three $Z$ corrections remove the terms of lower degree and leave
a standard binary $\ket{\CCZ}$ state.

\begin{proposition}[Finite distillation instrument]
\label{prop:finite-qrs-instrument}
There is a deterministic finite instrument using real Clifford operations
\begin{equation}
  \mathcal I_q:(\mathbb C^8)^{\otimes M_\mathrm{RS}}
  \longrightarrow(\mathbb C^8)^{\otimes K_\mathrm{RS}}
  \label{eq:finite-qrs-instrument}
\end{equation}
with the following property.  Let $J$ be any set of input positions with
$|J|\le t_\mathrm{RS}$.  Inputs in $J$ may be in an arbitrary joint state and
may be entangled with a reference system, while the other inputs are exact
$\ket{\CCZ}$ states.  For every nonzero complete record $h$,
\begin{equation}
  \rho_{\mathrm{out},R,h}
  =\ket{\CCZ}\!\bra{\CCZ}^{\otimes K_\mathrm{RS}}
     \otimes\tau_{R,h},
  \qquad \tau_{R,h}\ge0.
  \label{eq:qrs-recordwise-output}
\end{equation}
The sum over records is trace preserving; no outcome is postselected, and
all $K_\mathrm{RS}$ outputs are jointly exact.  The statement remains valid for
any finite number of parallel invocations by placing the other invocations
and their records in the reference system.
\end{proposition}

\begin{proof}
The field decomposition assigns every binary input to one outer qRS
coordinate.  Expanding arbitrary bad inputs in an operator basis and using
\eqref{eq:ccz-consumption-correction} therefore confines their relative
operator support to the coordinates in $J$.  The total syndrome rule
corrects that span because $|J|\le t_\mathrm{RS}$.  The logical qRS outputs are
then exact and factor from the reference.  Equation
\eqref{eq:qrs-output-polynomial} converts each logical qudit triple to one
standard binary output by deterministic measurements and Clifford
corrections.  Every branch is a scalar multiple of the same pure output
state; \Cref{lem:instrument-normalization} then gives the statement for the
complete instrument.
\end{proof}

After $q$ is chosen, $\mathcal I_q$ is a constant size measured Clifford
circuit that can be compiled before the simulated computation is given.  Its
binary basis changes are implemented by CNOT and SWAP gates, and its Pauli
corrections determined by the syndrome are evaluated by Boolean circuits with
bounded fan-in and constant size.  The location count includes all ancillas,
measurements, prescribed discards, classical control intervals, and idles of
this finite circuit.  Every invocation is charged this complete cost.

Choose once and for all
\begin{equation}
  \ell_0=\min\{\ell\ge3:2^\ell>1152R_\mathrm{gap}\ell^3\},
  \qquad q_0=2^{\ell_0},
  \label{eq:qrs-fixed-member}
\end{equation}
and abbreviate
\begin{equation}
  M_0=\frac{3q_0}{4}\ell_0^3,
  \quad K_0=\frac{q_0}{4},
  \quad t_0=\left\lfloor\frac{q_0}{24}\right\rfloor,
  \quad B_0=t_0+1,
  \quad \beta_0=\frac {M_0}{K_0}=3\ell_0^3.
  \label{eq:qrs-fixed-parameters}
\end{equation}
The floors leave the strict margin
\begin{equation}
  \frac{B_0K_0}{M_0}>
  \frac{q_0}{72\ell_0^3}>16R_\mathrm{gap}.
  \label{eq:qrs-fixed-margin}
\end{equation}
The field, bases, encoders, decoder table, and finite distillation circuit do
not vary with the qLTC protection scale.

\subsection{Code switching across scales}
\label{subsec:injection-switching}

The recursive factory must remain well defined even when a child factory fails
and produces an arbitrary physical state.  Let $\mathcal E$ be the encoder of an
$[[m,k,d]]$ code and define
\begin{equation}
  \mathcal K_R=\operatorname{span}
   \{P\mathcal E\ket\phi:\operatorname{wt}(P)\le R,
                     \ \phi\in\mathcal H_\mathrm{log}\}.
  \label{eq:common-correctable-subsystem}
\end{equation}
When $2R<d$, Knill--Laflamme gives one recovery isometry, independent of the
logical state and reference system,
\begin{equation}
  V_RP\mathcal E\ket\phi=\ket\phi\otimes\ket{\eta_P}.
  \label{eq:common-recovery-isometry}
\end{equation}
The environment vector may depend on $P$, but $V_R$ does not.  Indeed, for
two Pauli representatives $P,Q$ in the defining span, $P^\dagger Q$ has
weight below $d$.  The product $P^\dagger Q$ either has nonzero syndrome, in
which case $\mathcal E^\dagger P^\dagger Q\mathcal E=0$, or acts as a scalar on the code space.
The resulting scalars form a positive semidefinite Gram matrix.  Choosing vectors
$\ket{\eta_P}$ with that Gram matrix makes
\eqref{eq:common-recovery-isometry} preserve inner products and hence
defines $V_R$ on the whole span.  The tensor product of these isometries over
a bank of blocks acts on the corresponding tensor product span, including coherent
sums of residual Paulis on both sides of a density operator.  A canonical
logical Pauli preserves the range: it commutes or anticommutes with every
residual Pauli and changes only the logical and syndrome factors after
$V_R$.

For adjacent usable lengths $m_-<m_+$, write $k_-=k(m_-)$ and
$k_+=k(m_+)$, and let $\mathcal E_-$ and $\mathcal E_+$ be their canonical
encoders.  Set
\begin{equation}
  R_+=\lfloor t_*m_+\rfloor.
  \label{eq:injection-target-radius}
\end{equation}
The choice of $t_*$ in \Cref{lem:ec-gadget} gives
$2R_+<d_X(m_+),d_Z(m_+)$.  The canonical transpose sends
\begin{equation}
  (\text{old block }b,\text{ old coordinate }a)
  \longmapsto
  (\text{new block }a,\text{ new coordinate }b).
  \label{eq:inj-canonical-transpose}
\end{equation}
Its encoded Bell resource is
\begin{equation}
  \ket{R_{-+}}
  =(\mathcal E_-^{\otimes k_+}\otimes \mathcal E_+^{\otimes k_-})
    \bigotimes_{a,b}\ket\Phi_{L_{b,a},R_{a,b}},
  \label{eq:inj-switch-choi}
\end{equation}
of physical size
\begin{equation}
  s_{-+}=k_+m_-+k_-m_+.
  \label{eq:inj-switch-size}
\end{equation}
This is a stabilizer state.  Bell pairs followed by the canonical encoders
give an ordinary real Clifford preparation using $O(s_{-+})$ qubits and
depth polylogarithmic in $m_-$; the same statement holds with $m_+$ because
the two lengths differ by at most a constant factor.  Thus this is an
unprotected Clifford preparation of the kind required in
\Cref{prop:batched-stabilizer-factory}.
The source block $D_b$ controls a transversal CNOT into $L_b$; $D_b$ is
read in the $X$ basis and $L_b$ in the $Z$ basis.  For every possible
readout word, including a word outside the readout guarantee, the total decoder returns a
prescribed logical Pauli.  If $x_{b,a}$ and $z_{b,a}$ are the
decoded Bell outcomes, the canonical logical Pauli with $X$ label
$(z_{b,a})_b$ and $Z$ label $(x_{b,a})_b$ is applied to $R_a$.  The
maps from logical Pauli labels to physical representatives are known binary
matrices, so copy and XOR trees
evaluate the physical correction in logarithmic classical depth, followed
by one Pauli layer.  No physical correction that fills a syndrome and is derived from
an invalid word from the left bank is applied to the target.  Memory correction
relative to the zero syndrome is applied to the right bank throughout the readout and feedback
interval.  After the physical correction, a final calibrated error correction
cycle produces an output satisfying the boundary radius.

\begin{lemma}[Code switching with arbitrary inputs]
\label{lem:arbitrary-input-code-switch}
If the preparation of \eqref{eq:inj-switch-choi} and all operations on the
right bank are good, the switching instrument maps every source physical
state, including a state outside the source codespace and entangled with a
reference, into the tensor product of the target ranges $\mathcal K_{R_+}$
in \eqref{eq:common-correctable-subsystem}.  If the source also satisfies the
boundary condition \eqref{eq:typed-boundary-map} and the coupling and readout
windows on the left bank are good, the logical action is the exact transpose
\eqref{eq:inj-canonical-transpose}.  The conclusion is recordwise and the
instrument is trace preserving.
\end{lemma}

\begin{proof}
On the good Choi event, the resource and every reference system are
supported in the span of $P_LP_R\ket{R_{-+}}$, with a residual Pauli of the
allowed weight on every constituent block.  Apply the tensor product of the right bank
isometries \eqref{eq:common-recovery-isometry} to this entire span.  Its
image is a tensor product of right logical systems and right syndrome
systems.  An arbitrary state on $D$, followed by any operation or
measurement on $D$ and $L$, may steer and correlate those two right
factors with the measurement record, but the reduced state of the right bank
remains supported on the same subspace.  The same subspace argument applies to
off-diagonal matrix units and therefore does not require a classical mixture of
physical errors or any promise on the source.

For each record, the prescribed feedback is a canonical logical Pauli.
The invariance of $\mathcal K_{R_+}$ under canonical logical Paulis shows that
the feedback cannot enlarge the right residual support.  Good memory on the
right bank preserves $\mathcal K_{R_+}$ and keeps the residual support within
the allowed radius for all matrix units and arbitrary logical content.
After terminal correction, every block satisfies the corrected boundary condition.
When the source also satisfies its correctable promise, good
readout of the left bank identifies the actual Bell labels, and the ordinary
teleportation identity gives \eqref{eq:inj-canonical-transpose} on all
logical coordinates at once.  Expanding supported faults on the two banks
into Pauli matrix units proves the same statements on the actual record;
positivity and trace preservation come from the physical circuit.
\Cref{lem:instrument-normalization} then gives the asserted instrument
identity.
\end{proof}

The three banks of a $\CCZ$ resource use the same permutation of tensor indices in
every scale change.  A failed child may consequently contaminate several new
rows and logical coordinates, but all of them remain inside the one qRS
input group assigned to that child.  Confinement to one group, rather
than an independence assertion about its outputs, is the property used by
the finite distiller.

The Choi resources in this secondary construction are prepared by a
factory for batched stabilizer state preparation.  Set
\(
  \lambda_B\ge1
\)
to be a constant and define
\begin{equation}
  \ell_B(x)=\lceil\log_2(\lambda_Bx)\rceil,
  \quad M_B(x)=31^{\ell_B(x)},
  \quad K_B(x)=21^{\ell_B(x)},
  \quad \gamma_B=\log_2(31/21)<1.
  \label{eq:stabilizer-batch-functions}
\end{equation}

\begin{proposition}[Batched preparation of stabilizer states]
\label{prop:batched-stabilizer-factory}
Let $\ket\psi$ be a known pure stabilizer state on a polynomial number of
computational code blocks or their physical $H$ duals, of total physical size
$s=\operatorname{poly}(x)$.  Suppose its ordinary Clifford preparation has
polynomial size and polylogarithmic depth in the protection length $x$.  For
this class of targets, there are constants $p_B,C,c,c',c_B>0$ such that every integer $v\ge1$ and
every $0<p<p_B$ admit a protected preparation with
\begin{align}
 S_B(v,x,s)&\le C[v+M_B(x)]s,
 \label{eq:stabilizer-batch-space}\\
 T_B(v,x,s)&\le Cx^{\gamma_B}\log^{c_B}(x+2),
 \label{eq:stabilizer-batch-time}\\
 \mathcal W(\mathcal B_B;p)&\le
 C[v+M_B(x)](x+2)^c e^{-c'x}.
 \label{eq:stabilizer-batch-failure}
\end{align}
If the fault support contains no member of $\mathcal B_B$, all requested
outputs are jointly exact at the corrected boundary of
\Cref{def:normalized-boundary}.
\end{proposition}

\begin{proof}
The measured Clifford compiler of \Cref{lem:fixed-clifford-compiler}, applied
to the ordinary preparation of $\ket\psi$, gives raw slots whose bad support
family has weight at most any prescribed constant $z_0>0$ below a corresponding
constant physical threshold.  Indeed, the preparation has polynomially many
locations, and \eqref{eq:fixed-compiler-recursion} makes their union at most
$z_0$ with $O(\log\log(x+2))$ concatenation levels.  The resulting width
multiplier and time are polylogarithmic in $x$.  Take
$z_0\le[4\binom{31}{2}]^{-1}$.  A single purification level uses the
$[[31,21,3]]$ CSS code whose check matrix has the nonzero vectors of
$\mathbb F_2^5$ as columns.  Its row space is self-orthogonal, and its
quotient has an orthonormal basis $g_1,\ldots,g_{21}$, which determines a CNOT
encoder.  The inverse encoder is applied at each corresponding physical
position of the $31$ candidate copies, followed by measurement of the five
$X$-check and five $Z$-check rows.  The quantum core has constant depth
because the outer code has constant length.

The invariant for one purification level is that one arbitrary input can affect the
selected correction, but after acceptance the difference from its actual
error is a target stabilizer together with a uniformly correctable residual
on every constituent block.  The following candidate test establishes this
invariant without assuming a classical mixture of input errors.

Let $L_\psi\subseteq\mathbb F_2^{2s}$ be the full Lagrangian stabilizer space
of the target, including correlations between distinct constituent blocks.
For a measured Pauli word $v$, decode the two physical syndromes on every
constituent block to a correction $c(v)$.  The test of the complete word accepts exactly
when $c(v)$ lies within the chosen radius on each block and
\begin{equation}
  v+c(v)\in L_\psi.
  \label{eq:stabilizer-full-word-test}
\end{equation}
Both the decoder and membership in $L_\psi$ are total Boolean circuits of
polynomial size and polylogarithmic depth.  The test accepts every word
within a sufficiently small correctable neighborhood of $L_\psi$ and, on
arbitrary words, acceptance itself certifies membership in the span obtained
from $L_\psi$ by the larger permitted residual Paulis.

Write the paired $X$- and $Z$-check outcomes as
$v_i\in\mathbb F_2^{2s}$, $i\in[5]$.  If row $j$ is the only arbitrary
input, Pauli covariance of the inverse encoder gives, on every nonzero
record,
\begin{equation}
  v_i=H_{ij}e+\ell_i+d_i,
  \qquad \ell_i\in L_\psi,
  \label{eq:stabilizer-syndrome-words}
\end{equation}
where $e$ is an arbitrary Pauli label on row $j$ and $d_i$ is a uniformly
small residual from the other inputs and the quantum core.  The decoder tests
the $31$ possible bad rows and one null candidate.  For each candidate $j'$,
choose an index $i(j')$ with $H_{i(j'),j'}=1$ and set
$e_{j'}=v_{i(j')}$.  The test in
\eqref{eq:stabilizer-full-word-test} is applied to all five words
\begin{equation}
  v_i+H_{ij'}e_{j'}.
  \label{eq:stabilizer-candidate-test}
\end{equation}
The first passing candidate is used, with a prescribed fallback if none
passes.  The true row $j$ passes.  If a different candidate passes, the two
distinct nonzero columns of $H$ are linearly independent; a combination of
the five accepted relations isolates $e$ modulo $L_\psi$.  Hence the
selected correction still differs from the actual arbitrary row by a target
stabilizer and a uniformly correctable residual on every constituent block.

The candidate tests take polylogarithmic classical depth, so the surviving
rows must be protected before the candidate is known.  Immediately after the
inverse encoder circuit, measure a complete syndrome on every surviving
constituent and retain the outcome $a_0$.  On each Pauli matrix unit, the
state at that time is a bounded residual acting on a Pauli translate of
$\ket\psi$; the translate has a definite syndrome $s_0$, and
$a_0=s_0+v$ for a bounded extraction error $v$.  During the candidate tests,
each new syndrome is compared with this same $a_0$, so
\Cref{lem:sector-storage} bounds the residual independently of the length of
the classical computation.  Once the candidate is selected, its prescribed
physical Pauli is applied.  After the terminal error correction cycle, the
output satisfies the corrected boundary condition.

Propagate the selected Pauli through the inverse encoder and apply its
surviving component to the $21$ output rows.  The logical representatives
$g_a$ have odd weight.  Their deterministic encoder signs define one Pauli
correction during compilation, after which every ideal branch has target
$\ket\psi^{\otimes21}$.  For an arbitrary state on the bad row, expand its
operators in the Pauli orbit of the stabilizer state and retain the same
actual measurement record on both sides.  The preceding identities place all
surviving terms in the span generated by the permitted residual Paulis acting
on $\ket\psi^{\otimes21}$.  Positivity and normalization follow from the
physical instrument, so all $21$ logical targets are jointly exact whenever
the fault support contains no member of the bad support family assigned to
this purification step, including with arbitrary references.

Index the $31^{\ell_B}$ raw slots by words in $[31]^{\ell_B}$ and declare a
parent bad when at least two of its children are bad.  Its enumerator on the original fault supports
satisfies
\begin{equation}
  f_0(z)=z,
  \qquad
  f_j(z)\le \binom{31}{2}f_{j-1}(z)^2.
  \label{eq:stabilizer-batch-recursion}
\end{equation}
The choice of $z_0$ gives
$f_{\ell_B}(z_0)=e^{-\Omega(x)}$.  Local extraction, feedback, storage, and
failures of the terminal correction contribute only a polynomial number of
$e^{-\Omega(x)}$ window families, and are included in the same complete
batch event.

To expose the yield cost, let $c_\mathrm{raw}(x)$ and $t_\mathrm{raw}(x)$ be
polylogarithmic bounds for the width multiplier and time of one raw
protected preparation.  Once $M_B(x)\ge2c_\mathrm{raw}(x)$, at least
$M_B/(2c_\mathrm{raw})$ raw preparations can run in parallel while the completed
rows remain stored.  Consequently all $M_B$ prescribed raw inputs are prepared in at most
$2c_\mathrm{raw}+1$ parallel rounds, using $O(M_Bs)$ space and polylogarithmic time.
The $\ell_B=O(\log x)$ purification levels are then executed completely;
measured check rows are released, while memory correction is applied to every
surviving row during the classical tests.  There is no retry, and the schedule does
not depend on acceptance.

To prepare $v\ge1$ copies, set
\begin{equation}
  R_B^\mathrm{req}=\max\{v,M_B\},\qquad
  b_{\parallel}=\left\lfloor\frac{R_B^\mathrm{req}}{M_B}\right\rfloor,
  \qquad
  N_\mathrm{rnd}=
  \left\lceil\frac{v}{b_{\parallel}K_B}\right\rceil.
  \label{eq:stabilizer-request-schedule}
\end{equation}
Each round runs $b_{\parallel}$ complete batches.  If $v<M_B$, then
$N_\mathrm{rnd}\le1+M_B/K_B$; otherwise
$b_{\parallel}\ge v/(2M_B)$ and
$N_\mathrm{rnd}\le1+2M_B/K_B$.  Since
$M_B/K_B=O(x^{\gamma_B})$, this schedule proves
\eqref{eq:stabilizer-batch-space} and \eqref{eq:stabilizer-batch-time}.  The number
of completed batches is at most
$v/K_B+R_B^\mathrm{req}/M_B$, which together with the union over local windows gives
\eqref{eq:stabilizer-batch-failure}.  Parallel batches occupy disjoint original
spacetime domains; all surplus outputs are produced and discarded at
predetermined locations.
\end{proof}

\subsection{Recursive \texorpdfstring{$\ket{\CCZ}$}{CCZ} preparation}
\label{subsec:protected-ccz}

Set
\begin{equation}
  w=\lfloor2\beta_0\rfloor+2,
  \qquad a=2w,
  \qquad A_+=R_\mathrm{gap}a.
  \label{eq:injection-scale-constants}
\end{equation}
The definition gives $\beta_0<w<a$.  Moreover, $w<4\beta_0$, while
\eqref{eq:qrs-fixed-margin} gives
$B_0>16R_\mathrm{gap}\beta_0>4R_\mathrm{gap}w=2A_+$.  Thus
\begin{equation}
  \beta_0<w<a,
  \qquad B_0>2A_+.
  \label{eq:injection-margin}
\end{equation}
Let $m_0$ be a usable base length, to be chosen after the constants below,
and define recursively
\begin{equation}
  m_j=m_\mathrm{next}(am_{j-1}),
  \qquad k_j=k(m_j).
  \label{eq:injection-scales}
\end{equation}
Then
\begin{equation}
  a\le\frac{m_j}{m_{j-1}}\le A_+.
  \label{eq:injection-scale-ratios}
\end{equation}

For $j\ge1$, define
\begin{equation}
  \Xi_j=(m_j+m_{j-1})^2M_B(m_{j-1})
  \label{eq:inj-Xi}
\end{equation}
and the integers
\begin{align}
  V_0&=k_0,
  &N_0^\mathrm{out}&=1,
  \label{eq:inj-base-volume}\\
  U_j&=wk_j\left\lceil
       \max\left\{1,\frac{\Xi_j}{K_0wk_jV_{j-1}}\right\}
       \right\rceil,
  \label{eq:inj-Uj}\\
  V_j&=K_0U_jV_{j-1},
  &N_j^\mathrm{out}&=\frac{V_j}{k_j}.
  \label{eq:inj-VC}
\end{align}
Thus $U_j$ is divisible by $wk_j$ and all batches below are integral.

\begin{theorem}[Recursive preparation of encoded $\ket{\CCZ}$ states]
\label{thm:protected-ccz-states}
For a sufficiently large usable choice of $m_0$, there are constants
$p_\mathrm{fac},\kappa_F,C_{\mathrm{fac},S},C_{\mathrm{fac},T}>0$ such that the
following holds for every $0<p<p_\mathrm{fac}$ and every $j\ge0$.  A complete
factory starting from empty workspace outputs $N_j^\mathrm{out}$ triples of
$G(m_j)$ code blocks, each containing $k_j$ logical $\ket{\CCZ}$ states.  Its
peak space $S_j$, time $T_j$, and a family $\mathcal F_j$ of bad supports
satisfy
\begin{equation}
  S_j\le C_{\mathrm{fac},S}V_j,
  \qquad T_j\le C_{\mathrm{fac},T}m_j,
  \qquad \mathcal W(\mathcal F_j;p)\le e^{-\kappa_Fm_j}.
  \label{eq:protected-ccz-service}
\end{equation}
If the fault support contains no member of $\mathcal F_j$, all outputs are
jointly exact at the corrected
boundary of \Cref{def:normalized-boundary}, including with arbitrary
reference systems and arbitrary supported CPTP fault values.
\end{theorem}

\begin{proof}
The level zero factory prepares one raw physical $\ket{\CCZ}$ for each
canonical logical coordinate of $G(m_0)$ by a noisy physical $\CCZ$ and the
canonical encoding circuit.  Because $m_0$ is chosen before the simulated
computation, the complete location count of this factory is a constant
$L_0$.  All bare non-Clifford gates of this route occur at this level.

A factory at level $j$ has $M_0$ input groups, and $U_j$ complete
level-$(j-1)$ child factories are assigned to each group.  A predetermined schedule has
exactly $w$ parallel rounds, with $M_0U_j/w$ children active in each round and
$U_j/w$ children from every input group.  Every child outputs all
$N_{j-1}^\mathrm{out}$ encoded $\CCZ$ rows.  The child output rows occupy their assigned
buffers at the old code length, with memory error correction applied until all $w$ rounds have finished;
neither the number of rounds nor the retained rows depend on measurement
outcomes.

Within each input group, the old rows are partitioned into consecutive
packets of $k_j$.  If an old row has packet label $b$ and position
$c\in[k_j]$, while $d\in[k_{j-1}]$ is its old logical coordinate, the
canonical transpose sends
\begin{equation}
  (b,c;d)\longmapsto (b,d;c).
  \label{eq:inj-three-bank-transpose}
\end{equation}
The same permutation of tensor factors is used in all three $\CCZ$ banks.  Thus
the three qubits of each old logical resource retain the same row and
coordinate labels even for a
failed child entangled with other factory systems.  Each input group
produces
\begin{equation}
  Q_j=\frac{U_jV_{j-1}}{k_j}
  \label{eq:inj-new-rows}
\end{equation}
new rows.  For each row label and canonical coordinate, the instrument
$\mathcal I_{q_0}$ acts on the $M_0$ resource triples in the corresponding
positions.  All circuit components of the finite instrument, including its
preparation, decoder, feedback, and prescribed discards, are executed.  The
factory output consists of
\begin{equation}
  N_j^\mathrm{out}=K_0Q_j,
  \qquad V_j=k_jN_j^\mathrm{out}
  \label{eq:inj-tile-output}
\end{equation}
encoded binary $\CCZ$ triples.

The change of scale uses
\Cref{lem:arbitrary-input-code-switch}.  A failed child may create an
arbitrary joint state on many new rows and coordinates, but all affected
tensor factors remain in the assigned input group.  A group without a
failed child contains the exact logical $\ket{\CCZ}$ state in the common
logical subsystem; purity decouples that state from all bad groups, syndrome
systems, and references.  If at most $t_0$ groups contain failed children,
each invocation of $\mathcal I_{q_0}$ consequently has at most $t_0$ arbitrary
inputs.  Proposition~\ref{prop:finite-qrs-instrument}, applied to one
invocation with all other rows, coordinates, and records included in its
reference system, gives exact joint $\ket{\CCZ}$ outputs.  Repeating this
argument over the invocations preserves joint exactness,
without an additional failure factor depending on $Q_j$, $k_j$, or
$N_j^\mathrm{out}$.
The proof permits arbitrary correlations among the failed children.

The number of Choi copies used by the three banks is
\begin{equation}
  v_{\mathrm{sw},j}=\frac{3M_0U_jN_{j-1}^\mathrm{out}}{k_j},
  \label{eq:inj-switch-copies}
\end{equation}
and one target has size
\begin{equation}
  s_j=k_jm_{j-1}+k_{j-1}m_j.
  \label{eq:inj-switch-target-size}
\end{equation}
Using $V_{j-1}=k_{j-1}N_{j-1}^\mathrm{out}$ and
$V_j=K_0U_jV_{j-1}$ gives
\begin{equation}
  v_{\mathrm{sw},j}s_j
  =3M_0U_jV_{j-1}
   \left(\frac{m_{j-1}}{k_{j-1}}+\frac{m_j}{k_j}\right)
  =O(V_j).
  \label{eq:inj-switch-volume}
\end{equation}
The padding in \eqref{eq:inj-Uj} also gives
\begin{equation}
  M_B(m_{j-1})s_j
  \le(m_j+m_{j-1})^2M_B(m_{j-1})
  =\Xi_j\le V_j.
  \label{eq:inj-choi-batch-volume}
\end{equation}
All $v_{\mathrm{sw},j}$ Choi states for one parent are prepared in one batch.
The Choi states may be prepared jointly, but no completed batch
is shared by distinct parent factories.  Failure of a complete Choi batch is
therefore assigned to that parent; no independence among its outputs is used.

Let $S_j,T_j$ be the peak space and the time of one factory starting from
empty workspace.
At level zero, positive rate gives $V_0=k_0=\Theta(m_0)$.  The three
canonical encoders and the physical $\CCZ$ layer use
\begin{equation}
  S_0\le C_{\mathrm{base},S}V_0,
  \qquad
  T_0\le C_{\mathrm{base},T}m_0
  \label{eq:inj-base-resource-bounds}
\end{equation}
after increasing constants to include syndrome extraction, correction, and
all idling locations of the level zero circuit.  The constants
$C_{\mathrm{base},S}$ and $C_{\mathrm{base},T}$ are uniform over all usable base
lengths above the minimum required by the code family.  The later choice of $m_0$ therefore does
not alter the recurrence constants.
Assuming inductively that a child uses at most
$C_{\mathrm{fac},S}V_{j-1}$ space, one parallel round occupies exactly
\begin{equation}
  \frac{M_0U_j}{w}C_{\mathrm{fac},S}V_{j-1}
  =\frac{\beta_0}{w}C_{\mathrm{fac},S}V_j.
  \label{eq:inj-active-child-space}
\end{equation}
The accumulated buffers at the old code length use $O(\beta_0 V_j)$ space.  Equations
\eqref{eq:inj-switch-volume} and \eqref{eq:inj-choi-batch-volume} include
the complete batch of Choi state preparations, while the switch banks, work registers for the
finite circuit, output buffers, and the $O(m_j^2)$ zero blocks used for
routing fit in another
$C_{\mathrm{fac},\mathrm{loc}}V_j$ qubits.  Since $\beta_0/w<1$,
\begin{equation}
  S_j\le\frac{\beta_0}{w}C_{\mathrm{fac},S}V_j
       +C_{\mathrm{fac},\mathrm{loc}}V_j,
  \qquad
  S_j=O(V_j).
  \label{eq:inj-space-recurrence}
\end{equation}
The choice
\begin{equation}
  C_{\mathrm{fac},S}\ge
  \max\left\{C_{\mathrm{base},S},
  \frac{C_{\mathrm{fac},\mathrm{loc}}}{1-\beta_0/w}\right\}
  \label{eq:inj-space-induction-constant}
\end{equation}
closes the induction.

After all child outputs are available, switching and the finite qRS circuit take
time $O(m_j)$, including their feedback and final correction.  Here
$m_{j-1}^{\gamma_B}\log^{c_B}(m_{j-1}+2)=o(m_{j-1})$ because
$\gamma_B<1$, and the growth estimate \eqref{eq:inj-growth-bounds} below
makes the depth of the classical control circuit with bounded fan-in
$O(\log V_j+\log^c(m_j+2))=o(m_j)$.  Therefore
\begin{equation}
  T_j\le wT_{j-1}+C_\mathrm{step}m_j.
  \label{eq:inj-time-recurrence}
\end{equation}
Because $m_j\ge am_{j-1}$ and $a=2w$,
\begin{equation}
  T_j\le
  \left(C_{\mathrm{base},T}+\frac{C_\mathrm{step}}{1-w/a}\right)m_j
  =(C_{\mathrm{base},T}+2C_\mathrm{step})m_j.
  \label{eq:inj-linear-latency}
\end{equation}
Indeed, \eqref{eq:inj-base-resource-bounds} bounds the first term in the
unrolled recurrence, while the remaining terms form a geometric series with
ratio at most $w/a=1/2$.
Set $C_{\mathrm{fac},T}=C_{\mathrm{base},T}+2C_\mathrm{step}$.
An old block waits only for the $O(m_j)=O(m_{j-1})$ work of its immediate
parent.  After switching to the code of length $m_j$, memory error correction
at that length protects the block during the next waiting interval.  Thus no
factor equal to the number of scales appears in
\eqref{eq:inj-linear-latency}.

For the failure bound, the domain of each child contains all locations from
the beginning of its preparation through the production of its output.  The
parent domain contains the subsequent storage, Choi state preparation and
use, qRS circuit, feedback, and terminal error correction.  The child and
parent domains are disjoint.  Let $\bar p>0$ be a common threshold for the finite collection of
qLTC, switching, and stabilizer batch operations performed by the parent.
Let $\mathcal R_j$ be the corresponding family of bad fault sets.  The
complete batch of Choi state preparations is bounded by
\eqref{eq:stabilizer-batch-failure}, with
$v=v_{\mathrm{sw},j}$ and $x=m_{j-1}$.  All other parent operations occupy at
most a polynomial number of error correction, memory, readout, and feedback
windows at lengths $m_{j-1}$ and $m_j$.  The bounds for those windows and a
union over their original supports therefore give constants
$C_R,d_R,e_R,c_R>0$ such that, for every $p<\bar p$,
\begin{equation}
  \mathcal W(\mathcal R_j;p)
  \le C_RV_j^{d_R}(m_j+2)^{e_R}e^{-c_Rm_{j-1}}.
  \label{eq:inj-parent-catastrophe}
\end{equation}
An input group is bad when at least one of its $U_j$ children is bad.  Let
$\mathcal F_0$ contain the singleton fault sets in the level zero circuit.
For $j\ge1$, let $\mathcal F_j$ be the union of $\mathcal R_j$ with the
products obtained by choosing one member of $\mathcal F_{j-1}$ from a child
in each of $B_0=t_0+1$ distinct input groups.  Thus $\mathcal F_j$ is a
family of bad sets on the original physical locations.  Put
$f_j=\mathcal W(\mathcal F_j;p)$.  Then
\begin{equation}
  f_0\le L_0p,
  \qquad
  f_j\le \binom{M_0}{B_0}[U_jf_{j-1}]^{B_0}
     +C_RV_j^{d_R}(m_j+2)^{e_R}e^{-c_Rm_{j-1}}.
  \label{eq:inj-failure-recurrence}
\end{equation}
The power $B_0$ follows from disjoint original spacetime domains of the chosen
children in distinct input groups.

The inequalities $a\le m_j/m_{j-1}\le A_+$ and the definitions give
constants $C_U,C_V$ with
\begin{equation}
  U_j\le C_U(m_j+2)^P,
  \qquad P=2+\log_2 31,
  \qquad
  \log V_j\le C_V\log^2(m_j+2).
  \label{eq:inj-growth-bounds}
\end{equation}
Indeed, $M_B(m_{j-1})=O(m_{j-1}^{\log_2 31})$ and
$m_{j-1}\le m_j\le A_+m_{j-1}$,
so $\Xi_j=O((m_j+2)^P)$.  Expanding the ceiling in
\eqref{eq:inj-Uj} then gives the first inequality.  Iterating
$\log V_j=\log V_{j-1}+\log K_0+\log U_j$ over
$j=O(\log m_j)$ levels gives the second.

Set $\kappa_F=c_R/(8A_+)$.  Choose one cutoff $X_0$ such that, for every
$x\ge X_0$, the local prefactor in
\eqref{eq:inj-parent-catastrophe} is at most $e^{c_Rx/2}$ when
$m_j\le A_+x$, and
\begin{equation}
  \binom{M_0}{B_0}C_U^{B_0}(A_+x+2)^{PB_0}
  \le \frac12 e^{\kappa_F(B_0-A_+)x}.
  \label{eq:inj-induction-cutoff}
\end{equation}
Such a common cutoff exists because all displayed prefactors have logarithm
$O(\log^2(x+2))$.  Increase $X_0$ beyond the minimum usable lengths required
for the code family, the raw base preparation, the factory for stabilizer
states, and the linear bounds on the classical control depth.  Choose usable
$m_0\ge X_0$ and
\begin{equation}
  0<p_\mathrm{fac}\le
  \min\left\{\bar p,\frac{e^{-\kappa_Fm_0}}{L_0}\right\}.
  \label{eq:inj-factory-threshold}
\end{equation}
The base case is immediate.  If
$f_{j-1}\le e^{-\kappa_Fm_{j-1}}$, set $x=m_{j-1}$.  The first term of
\eqref{eq:inj-failure-recurrence} is, by
\eqref{eq:inj-induction-cutoff},
\begin{equation}
  \binom{M_0}{B_0}[U_jf_{j-1}]^{B_0}
  \le\frac12 e^{-\kappa_FA_+x}
  \le\frac12 e^{-\kappa_Fm_j}.
  \label{eq:inj-child-induction}
\end{equation}
The term from operations performed by the parent is at most
$e^{-c_Rx/2}\le\frac12e^{-\kappa_Fm_j}$ after enlarging the same cutoff.
Here $m_j\le A_+x$ and $\kappa_F=c_R/(8A_+)$ leave a constant exponential
margin.  Thus induction gives
\begin{equation}
  f_j\le e^{-\kappa_Fm_j}
  \label{eq:inj-exponential-tile}
\end{equation}
for every level.

Equations~\eqref{eq:inj-space-recurrence} and
\eqref{eq:inj-linear-latency} give the resource bounds, while
\eqref{eq:inj-exponential-tile} gives the weight enumerator bound for the bad
sets of a complete factory.
If the fault support contains no member of $\mathcal F_j$, fewer than $B_0$
input groups are bad.  The switching
maps are defined on arbitrary source states.  Their outputs, including those
arising from failed child factories, lie in the tensor product of the ranges
$\mathcal K_{R_+}$ in \eqref{eq:common-correctable-subsystem}, on which the
finite qRS instrument identity applies.  Proposition
\Cref{prop:finite-qrs-instrument} then
gives the joint exact logical outputs, and terminal error correction
establishes the corrected boundary.
\end{proof}

For later use, there are constants $C_{\mathrm{req},S},C_{\mathrm{req},T}>0$,
independent of $q_\mathrm{row}$ and $j$, with the following bounds.  A request
for $q_\mathrm{row}$ complete rows is served by the required number of complete
factories on disjoint locations, with predetermined surplus outputs discarded.  If
$\mathcal F_\mathrm{req}$ is the union of their translated bad set families, then
\begin{align}
  S_\mathrm{req}(q_\mathrm{row},j)&\le C_{\mathrm{req},S}(q_\mathrm{row}m_j+V_j),
  \label{eq:protected-request-space}\\
  T_\mathrm{req}(q_\mathrm{row},j)&\le C_{\mathrm{req},T}m_j,
  \label{eq:protected-request-time}\\
  \mathcal W(\mathcal F_\mathrm{req};p)&\le
  \left\lceil\frac {q_\mathrm{row}}{N_j^\mathrm{out}}\right\rceil
    e^{-\kappa_Fm_j}.
  \label{eq:protected-request-failure}
\end{align}
The last bound is the sum rule for weight enumerators.  The case
$q_\mathrm{row}=0$ is the empty instrument.

\subsection{Clifford gates and routing}
\label{subsec:injection-clifford-routing}

State injection still requires Clifford gates on prescribed logical
coordinates.  The prescribed Clifford operations are constructed on the complete logical
space of the primary code and its physical $H$ dual.  For
$q_\mathrm{dat}$ full $G(m)$ blocks, pad the bank to
\begin{equation}
  q'_\mathrm{dat}=k\left\lceil\frac {q_\mathrm{dat}}k\right\rceil<q_\mathrm{dat}+k.
  \label{eq:injection-clifford-padding}
\end{equation}
Groups of $k$ blocks are transposed using the raw switching segment in
\Cref{lem:canonical-code-switching}; the decoder computations of consecutive
segments are deferred as in \Cref{lem:deferred-clifford-corrections}.  After
this transpose, a linear transformation involving a constant number of
logical coordinates is implemented by a constant length circuit of whole block CNOT and
SWAP gates, followed by the inverse transpose.  One such operation uses
$O(q_\mathrm{dat}m+m^2)$ qubits.

\begin{proposition}
\label{prop:injection-clifford-routing}
For every sufficiently large usable $m$, on $q_\mathrm{dat}$ full primary blocks, a
prescribed dense register permutation or one prescribed layer of disjoint
logical Clifford gates and Pauli measurements, with compiled classical enables, can
be implemented with
\begin{align}
  S_\mathrm{Cliff}(q_\mathrm{dat},m)&\le C\bigl[q_\mathrm{dat}m+m^2+mM(m)\bigr],
  \label{eq:injection-clifford-space}\\
  T_\mathrm{Cliff}(q_\mathrm{dat},m)&\le C\bigl[m+\log(q_\mathrm{dat}+1)\bigr].
  \label{eq:injection-clifford-time}
\end{align}
If the fault support contains no member of a bad family of weight
$\operatorname{poly}(q_\mathrm{dat},m)e^{-cm}$, the physical instrument is
the stated logical operation, including with arbitrary reference systems.
\end{proposition}

\begin{proof}
A CNOT on one prescribed coordinate is obtained as follows.  Choose distinct
coordinates $v,v'$.  On their span, let
\begin{equation}
  A=\begin{pmatrix}0&1\\1&0\end{pmatrix},
  \qquad
  B=\begin{pmatrix}1&1\\1&0\end{pmatrix},
  \qquad A+B=E_v,
  \label{eq:injection-single-coordinate-matrices}
\end{equation}
where $E_v$ has a single nonzero diagonal entry at $v$.  Extend $A$ and
$B$ by the identity on the other coordinates.  Both
matrices are invertible.  For a matching from source blocks $b$ to target
blocks $\pi(b)$, conjugate an enabled whole block CNOT by $A$ on each source
and then repeat with $B$.  If the enable for block $b$ is $e_b$, the source is
restored and the target changes by
\begin{equation}
  y_{\pi(b)}\longmapsto y_{\pi(b)}+e_b(A+B)x_b
  =y_{\pi(b)}+e_bE_vx_b.
  \label{eq:injection-single-coordinate-cnot}
\end{equation}
Thus the construction applies CNOT only to coordinate $v$.  Three such
CNOTs give a SWAP on that coordinate.  Repeating the construction for
$v=1,\ldots,k$ implements arbitrary CNOT and SWAP masks in $O(k)=O(m)$
constant depth Clifford segments.

The remaining Clifford gates use the same sequence of coordinate rounds.
Physical Hadamard maps a complete $G(m)$ block to its physical $H$
dual and applies logical Hadamard to every canonical coordinate.  The relevant
dual rows are grouped in batches of $k$, with independent dual zero rows in
the last incomplete batch.  A switch to $G(m)$ followed by the self switch of
$G(m)$ cancels the two transposes.  The output is therefore encoded by
$G(m)$ with the same coordinate labels, while the padding is again an
independent zero state and is discarded.

A complete canonical $Y$ row realizes logical $S$ on a full row.
CNOT from the data to $\ket{+i}$ followed by logical $Z$ readout gives
$S$ for outcome zero and $iS^\dagger=iZS$ for
outcome one.  For either $H$ or $S$, conjugating the full row
operation by the coordinate SWAP circuit restricts it to the enabled data
coordinates.  The zero work bank is subsequently measured or reset at all
prescribed locations, including coordinates on which the full row gate
changed its initial state, while disabled data remain unmeasured.

A masked $\CZ$ is obtained by conjugating the target of the masked CNOT by
Hadamard; the two Hadamards cancel when the enable bit vanishes.  The chosen
Clifford gate set is compiled from these operations.  Pauli measurement uses
the corresponding Clifford basis change followed by complete logical $X$ or
$Z$ readout.  A measurement bank contains no unmeasured logical state that
must survive, and unused coordinates are independent dummy states.  Reset
traces out the old row and replaces it by a freshly prepared zero row.

Dense routing also uses only these Clifford operations.  A prescribed
permutation $(b,v)\mapsto(\pi_v(b),v)$ that preserves the coordinate label is
implemented with zero target blocks and the CNOT copying circuit.  If
\begin{equation}
  \sum_x\alpha_x\ket{x}_S\ket{0}_T
  \longmapsto
  \sum_x\alpha_x\ket{x}_S\ket{Px}_T,
  \label{eq:injection-aligned-copy}
\end{equation}
then complete logical $X$ readout of the source gives an outcome $z$ and a
phase $(-1)^{z\cdot x}$.  Applying the canonical $Z$ label $Pz$ to the
target cancels this phase because $(Pz)\cdot(Px)=z\cdot x$.  The resulting
map is the required permutation instrument, including when the data are
entangled with a reference.  Conjugation by the groupwise transpose gives a
permutation within each block.

For a general permutation of the $q'_\mathrm{dat}k$ slots, consider the
bipartite multigraph whose vertices are source and destination blocks and
whose edges are the logical slots.  This graph is $k$ regular, so successive
perfect matchings give $k$ colors.  A permutation within the source blocks
assigns the colors, the routing that preserves coordinate labels implements the
permutation for each color, and a permutation within the destination blocks
places the coordinates in their final positions.  The complete routing circuit
uses $O(m)$ constant depth Clifford segments and
$O(q_\mathrm{dat}m+m^2)$ space.

The decoder computations attached to these segments are evaluated in parallel
after the relevant records become available.  Every raw switch or Clifford measurement segment is followed by a
fresh complete extraction.  Applying
\Cref{lem:deferred-clifford-corrections} to a sequence of $N_\mathrm{seg}$ segments with
record volume $V$ gives time
\begin{equation}
  O\bigl(N_\mathrm{seg}+\log^c(m+2)+\log(V+2)\bigr).
  \label{eq:injection-deferred-time}
\end{equation}
The lemma includes the repair of noisy syndrome records, the propagation of
deferred physical Pauli corrections through later measured Clifford segments, and
the enabled frame recurrence \eqref{eq:enabled-frame-recurrence}.  Memory
correction relative to the fresh raw syndrome is applied to every live block
until the final physical Pauli and terminal error correction.

The $k$ coordinate rounds and the routing construction use
$N_\mathrm{seg}=O(k)=O(m)$ constant
depth Clifford segments.  Equation~\eqref{eq:injection-deferred-time} gives the time
bound because the complete record volume is polynomial in $q_\mathrm{dat}m$.  The data,
transpose grids, target banks, and work banks use $O(q_\mathrm{dat}m+m^2)$ qubits.  Only
a layer of $S$ gates needs canonical $Y$ rows.  By
\Cref{lem:primary-y-resource-interface}, one execution of all $k$ coordinate rounds uses
$q'_\mathrm{dat}$ such rows and adds
$O(q_\mathrm{dat}m+m^2+mM(m))$ space and polylogarithmic time.
The bad support family is the union of the code switching, preparation,
readout, memory, and final correction families for these concrete segments.
For each raw switch, use the family supplied by
\Cref{lem:canonical-code-switching}, together with the families for the
appended syndrome extraction, memory during the deferred computation, and
terminal correction.  They are defined on the original physical locations,
uniformly over runtime enables, actual records, and supported CPTP fault
values.
Their number is polynomial in $q_\mathrm{dat}$ and $m$, while each has weight
$e^{-\Omega(m)}$.  The branch identities above prove correctness on the
complement of that family.
\end{proof}

\subsection{State injection and overhead}
\label{subsec:injected-computation}

Before a protected resource row is used, all pending syndrome calculations
are completed and the resulting Pauli corrections are applied physically to
the data.  A fresh extraction and error correction cycle then precedes the
coordinate SWAP between the enabled operands and zero work rows.  Any Pauli
corrections produced by this selection circuit are also applied physically
before injection.

For one logical coordinate, the injection circuit applies CNOT from each data
qubit to the corresponding qubit of the resource triple and reads the three
resource qubits in the logical $Z$ basis.  On outcome
$z=(z_1,z_2,z_3)$, the correction is
\eqref{eq:ccz-consumption-correction}.  The branch Kraus operator is a scalar
times the desired $\CCZ$ on the data, for arbitrary data--reference inputs.
All injections across a row run in parallel, and the row is followed by
terminal error correction.  The correction in
\eqref{eq:ccz-consumption-correction} contains three logical $\CZ$ masks and
three logical Pauli masks.  The operations in
\Cref{prop:injection-clifford-routing} implement these masks.  The three
quadratic outcome bits are evaluated by constant size Boolean gates and
copied by bounded fan-out trees before the corresponding $\CZ$ layer
starts.  Memory correction is applied to the data throughout this wait, and these non-Pauli
controls are completed before any subsequent non-Clifford gate.

Fix a complete branch record, including the factory record and all subsequent
measurement and correction data, and condition on a support outside the bad
events of the resource factory and the surrounding circuit windows.  The protected
resource theorem gives a pure encoded $\ket{\CCZ}$ state on the logical
factor, decoupled from its factory environment.  The branch
identity for \eqref{eq:ccz-consumption-correction} then gives the desired
logical gate on the data, while the surrounding circuit has bounded spread and leaves
only a correctable residual.  After final error correction, the data satisfy
the corrected boundary condition.  Applying this argument to Pauli matrix units on both sides
and then summing the actual physical instrument proves the statement for
arbitrary supported CPTP fault values; no measurement record or expansion
coefficient is assigned an additional probability.

At the top scale $m=m_j$, put $k=k(m)$.  The dense routing in
\Cref{prop:injection-clifford-routing} requires
\begin{equation}
  q_\mathrm{row}=O\!\left(\left\lceil\frac Wk\right\rceil+k\right)
  \label{eq:inj-requested-rows}
\end{equation}
complete resource rows for each of a constant number of roles.  Since
$k=\Theta(m)$,
\begin{equation}
  q_\mathrm{row}m=O(W+m^2).
  \label{eq:inj-row-volume}
\end{equation}
Encoded $\CCZ$ preparation occupies $O(V_j)$ qubits by
\eqref{eq:protected-request-space}.  Proposition
\ref{prop:injection-clifford-routing} bounds the data, routing, logical $Y$
resources, direct zero and plus preparations, and reusable correction
workspace by $O(W+m^2+mM(m))$.  Hence
\begin{equation}
  S_\mathrm{layer}^\mathrm{inj}
  \le C\bigl[W+V_j+m^2+mM(m)\bigr].
  \label{eq:inj-layer-space}
\end{equation}

All resource factories for one gate type start together.  The factory time is
$O(m)$.  The routing and Clifford operations are covered by
\Cref{prop:injection-clifford-routing}; their local decoder computations run in
parallel with later raw Clifford segments.  If $N_\mathrm{rec}$ is the number
of classical record bits and Pauli labels produced in the layer, then
\begin{equation}
  N_\mathrm{rec}\le
  [W+V_j+m^2+mM(m)+2]^{c_1}(m+2)^{c_2}.
  \label{eq:inj-record-count}
\end{equation}
The nonlinear local decoder computations have depth $\log^c(m+2)$.  The
remaining global calculations are linear: the compiled Clifford circuits
determine the record shifts, and \eqref{eq:enabled-frame-recurrence} gives the
enabled Pauli frames.  Constant size Boolean gates evaluate the injection
correction.  Balanced XOR and copy trees evaluate the linear maps
in $O(\log(N_\mathrm{rec}+2))$ depth.  Original enable bits are copied once per
ideal layer, and a measurement outcome used for a non-Pauli correction is
copied once before the corresponding gate type is executed.  Therefore
\begin{equation}
  T_\mathrm{layer}^\mathrm{inj}
  =O\!\left(m+\log^c(m+2)+\log(N_\mathrm{rec}+2)\right)
  =O(m+\log(W+1)).
  \label{eq:inj-layer-time}
\end{equation}

\begin{proposition}[One layer with injected $\CCZ$ gates]
\label{prop:injection-layer}
One adaptive ideal layer is implemented by the preceding schedule with the
space and time bounds \eqref{eq:inj-layer-space} and
\eqref{eq:inj-layer-time}.  If the fault support contains no member of the
resource factory's bad family or of the bad families for the remaining
circuit windows, the physical instrument is exactly the prescribed
ideal layer on inputs entangled with arbitrary reference systems, and all
surviving output blocks satisfy the corrected boundary condition.
\end{proposition}

\begin{proof}
At the beginning of the layer, the logical slots are routed into the banks
specified by their roles in the ideal circuit.  These banks contain gate
operands, measurement and reset positions, newly initialized wires, and idle
data; complete zero blocks fill the unused slots.  The permutation in
\Cref{prop:injection-clifford-routing} aligns the operands of every
multi-qubit gate.  A compiled Clifford decomposition uses the same original
enable bit on each of its factors.  Measurement and reset locations are
executed unconditionally.  A measurement bank contains no logical state that
must survive the layer, while a reset traces out the old row and prepares a
fresh zero row.  Memory correction is applied to all other blocks while the Clifford gates,
measurements, and routes are executed.

The non-Clifford operation uses the injection schedule for one row and the recordwise
identity established above.  The selection SWAP maps each enabled triple to
a zero work row, and an error correction cycle precedes its coupling to an
encoded $\ket{\CCZ}$ triple.
Every disabled work triple remains in
$\ket{000}$, so full row injection acts as the identity at those positions.
After the inverse selection SWAPs, the data are encoded in their assigned
blocks, and the final dense permutation gives the layout for the next layer.  The branch
identity above and the readout and memory lemmas preserve the complete record
and arbitrary reference systems.
After every required logical outcome has been decoded, the original
classical update $C_t$ is evaluated exactly once.  The resulting instrument
is therefore the prescribed adaptive ideal layer, rather than only its
unitary part.

Equation~\eqref{eq:inj-row-volume} counts the data and resource output banks,
while \eqref{eq:protected-request-space} accounts for the factory used by the
current gate type.
Successive gate types reuse the same workspace.  The time bound is
\eqref{eq:inj-layer-time}.  All descendant locations are included in the
factory's bad set family; memory windows after the factory output is produced are counted
separately.  The two
families can therefore be composed without counting the same physical
locations twice.
\end{proof}

For a constant $A_L$ chosen in the theorem below, let
\begin{equation}
  L=\log\frac{WD}{\varepsilon},
  \qquad
  j(L)=\min\{j:m_j\ge A_L L\},
  \qquad
  m=m_{j(L)},
  \label{eq:inj-top-scale}
\end{equation}
and define the additive space term for this route
\begin{equation}
  F_\mathrm{inj}(L)=V_{j(L)}+m^2+mM(m).
  \label{eq:injection-padding}
\end{equation}

\begin{theorem}
\label{thm:injection-route}
There are positive constants $p_\mathrm{inj},A_L,C,C_T,C_F,C_S$.  Every
adaptive Clifford+$\CCZ$ circuit of width $W$, depth $D$, and target error
$0<\varepsilon\le1/2$ has, for physical noise rate $p<p_\mathrm{inj}$, a
fault-tolerant simulation satisfying
\begin{align}
  S_\mathrm{FT}^\mathrm{inj}&\le C[W+F_\mathrm{inj}(L)],
  \label{eq:inj-final-space}\\
  D_\mathrm{FT}^\mathrm{inj}&\le C_T D\log\frac{WD}{\varepsilon},
  \label{eq:inj-final-time}\\
  \|P_\mathrm{FT}^\mathrm{inj}-P_\mathrm{ideal}\|_\mathrm{TV}&\le\varepsilon.
  \label{eq:inj-final-error}
\end{align}
If $W\ge C_F\,F_\mathrm{inj}(L)$, then
\begin{equation}
  S_\mathrm{FT}^\mathrm{inj}\le C_S W.
  \label{eq:inj-constant-space}
\end{equation}
\end{theorem}

\begin{proof}
Take $p_\mathrm{inj}$ below $p_\mathrm{fac}$ and the thresholds of all
preparation, routing, readout, feedback, and memory operations outside the
factory.
Choose $A_L$ large enough that $A_L\log2\ge m_0$ and that the exponential
absorption below holds.  Since $L\ge\log2$, the minimality of $j(L)$ and
\eqref{eq:injection-scale-ratios} give
\begin{equation}
  A_L L\le m\le A_+A_L L.
  \label{eq:inj-selected-scale}
\end{equation}
Evaluate the prescribed initial classical map and prepare the $W$ input logical
zeros in full primary code rows, padding unused coordinates by independent
zeros.  Apply \Cref{prop:injection-layer} to every ideal layer and reuse the
ancillary banks.  After the last layer, perform corrected destructive
readout and evaluate $C_D$ once.  These initial
and final operations have the same space and time order as one layer and are
included in the bad set families outside the factory.  The inequalities in
\eqref{eq:inj-selected-scale} give $m=\Theta(L)$, which proves
\eqref{eq:inj-final-time}.  From \eqref{eq:inj-growth-bounds},
\begin{equation}
  \log F_\mathrm{inj}(L)=O(\log^2(L+2)).
  \label{eq:inj-padding-growth}
\end{equation}

At the top level, let $\mathcal F_\mathrm{inj}$ be the union of the translated
bad support families for the complete factory calls and all scale-$m$
operations outside the factory.  The child families are already included in
the complete factory families.  Define $\mathsf{Bad}_\mathrm{inj}$ as the event
that the actual fault support is not $\mathcal F_\mathrm{inj}$-avoiding.  For
constants $C_\mathrm{bad},b_\mathrm{bad},c_\mathrm{bad}>0$, the physical spacetime
count gives
\begin{equation}
  \Pr[\mathsf{Bad}_\mathrm{inj}]
  \le C_\mathrm{bad}(D+1)[W+F_\mathrm{inj}(L)]
       (m+2)^{b_\mathrm{bad}}e^{-c_\mathrm{bad}m}.
  \label{eq:inj-global-bad}
\end{equation}
Equation~\eqref{eq:inj-padding-growth} absorbs the additive padding and
polynomial factors into a smaller exponential.  Hence there are
$C_E,c_E>0$ such that
\begin{equation}
  \Pr[\mathsf{Bad}_\mathrm{inj}]
  \le C_EWD e^{-c_Em}.
  \label{eq:inj-global-bad-absorbed}
\end{equation}
Choose $A_L$ so that $c_EA_L\ge2$ and
$C_Ee^{-c_EA_L\log2/2}\le1$.  Since
$WD=\varepsilon e^L$ and $m\ge A_L L$, the last expression is then at most
$\varepsilon$.  Outside
this event, recordwise composition of the layer instruments gives the ideal
adaptive output distribution.  Equation~\eqref{eq:inj-final-space} follows
from \eqref{eq:inj-layer-space}, and the width condition absorbs
$F_\mathrm{inj}(L)$ to give \eqref{eq:inj-constant-space}.
\end{proof}

For constants $c,c'>0$, the regime $D\le W^c$ and
$\varepsilon\ge W^{-c'}$ has $L=\Theta(\log W)$.  Equation
\eqref{eq:inj-padding-growth} then gives
\begin{equation}
  \log F_\mathrm{inj}(L)=O((\log\log W)^2)=o(\log W),
  \label{eq:inj-wide-regime}
\end{equation}
so the condition on the width holds automatically for all sufficiently large $W$ in
that regime.

%% file: sections/conclusion.tex
\section{Discussion and outlook}
\label{sec:conclusion}

In this work, we introduced two constructions that provably achieve quantum fault tolerance with purely logarithmic time overhead and constant space overhead. The two constructions identify two different ways in which a good qLTC can support logarithmic-time fault tolerance.  In our main construction, non-Clifford gates are implemented directly by the computational code's native transversal $\CCZ$ gates,   while the second construction implements non-Clifford gates through a separate state factory.  

A natural question is whether this achieves the optimal space and time overhead possible for quantum fault tolerance. Relevant lower
bounds have been discussed in~\cite{NguyenPattison2025,
FawziFawziRouze2022,BhartiHaugTanggara2026}, but further work is needed
to determine their precise implications for our setting. More generally, it would be interesting to further characterize the optimal space-time tradeoff, which would
clarify how much quantum space is needed to achieve a given time overhead, particularly below the logarithmic scale.
A related question is whether the width condition in our constructions is unavoidable, or whether more efficient resource-state preparation can extend the constant-space-overhead guarantee to narrower circuits.

Another valuable direction is to relax the assumptions on classical processing
and physical connectivity.
Determining the additional overhead required for noisy classical control
or geometrically restricted interactions would help relate these asymptotic constructions to more practical computational architectures.

%% file: references.bib
@inproceedings{AharonovBenOr1999,
  title={Fault-tolerant quantum computation with constant error},
  author={Aharonov, Dorit and Ben-Or, Michael},
  booktitle={Proceedings of the twenty-ninth annual ACM symposium on Theory of computing},
  pages={176--188},
  year={1997}
}

@article{Kitaev1997,
  title={Quantum computations: algorithms and error correction},
  author={Kitaev, A Yu},
  journal={Russian Mathematical Surveys},
  volume={52},
  number={6},
  pages={1191--1249},
  year={1997}
}

@article{KnillLaflammeZurek1998,
  title={Resilient quantum computation},
  author={Knill, Emanuel and Laflamme, Raymond and Zurek, Wojciech H},
  journal={Science},
  volume={279},
  number={5349},
  pages={342--345},
  year={1998},
  publisher={American Association for the Advancement of Science}
}

@article{TamiyaKoashiYamasaki2024,
  title={Fault-tolerant quantum computation with polylogarithmic time and constant space overheads},
  author={Tamiya, Shiro and Koashi, Masato and Yamasaki, Hayata},
  journal={Nature Physics},
  volume={22},
  number={1},
  pages={27--32},
  year={2026},
  publisher={Nature Publishing Group UK London}
}

@inproceedings{NguyenPattison2025,
  title={Quantum fault tolerance with constant-space and logarithmic-time overheads},
  author={Nguyen, Quynh T and Pattison, Christopher A},
  booktitle={Proceedings of the 57th Annual ACM Symposium on Theory of Computing},
  pages={730--737},
  year={2025}
}

@misc{GoodQLTCTransversal2026,
  author = {Yiming Li and Zimu Li and Zi-Wen Liu},
  title = {Transversal non-{Clifford} gates on good quantum locally testable codes},
  year = {2026},
  eprint = {2609.26691},
  archivePrefix = {arXiv},
  primaryClass = {quant-ph},
  note = {arXiv:2609.26691},
  url = {https://arxiv.org/abs/2609.26691}
}

@article{Gottesman2013,
author = {Gottesman, Daniel},
title = {Fault-tolerant quantum computation with constant overhead},
year = {2014},
issue_date = {November 2014},
publisher = {Rinton Press, Incorporated},
address = {Paramus, NJ},
volume = {14},
number = {15–16},
issn = {1533-7146},
journal = {Quantum Info. Comput.},
month = nov,
pages = {1338–1372},
numpages = {35}
}

@article{Fawzi2018PolyTime,
  title={Constant overhead quantum fault tolerance with quantum expander codes},
  author={Fawzi, Omar and Grospellier, Antoine and Leverrier, Anthony},
  journal={Communications of the ACM},
  volume={64},
  number={1},
  pages={106--114},
  year={2020},
}

@article{YamasakiKoashi2024,
  title={Time-efficient constant-space-overhead fault-tolerant quantum computation},
  author={Yamasaki, Hayata and Koashi, Masato},
  journal={Nature Physics},
  volume={20},
  number={2},
  pages={247--253},
  year={2024},
  publisher={Nature Publishing Group UK London}
}

@article{Prasad1977,
  title={Strong approximation for semi-simple groups over function fields},
  author={Prasad, Gopal},
  journal={Annals of Mathematics},
  volume={105},
  number={3},
  pages={553--572},
  year={1977},
  publisher={JSTOR}
}

@article{Berkowitz1984,
  title={On computing the determinant in small parallel time using a small number of processors},
  author={Berkowitz, Stuart J},
  journal={Information processing letters},
  volume={18},
  number={3},
  pages={147--150},
  year={1984},
  publisher={Elsevier}
}

@InProceedings{FawziFawziRouze2022,
  author =	{Fawzi, Omar and M\"{u}ller-Hermes, Alexander and Shayeghi, Ala},
  title =	{{A Lower Bound on the Space Overhead of Fault-Tolerant Quantum Computation}},
  booktitle =	{13th Innovations in Theoretical Computer Science Conference (ITCS 2022)},
  pages =	{68:1--68:20},
  series =	{Leibniz International Proceedings in Informatics (LIPIcs)},
  ISBN =	{978-3-95977-217-4},
  ISSN =	{1868-8969},
  year =	{2022},
  volume =	{215},
  URL =		{https://drops.dagstuhl.de/entities/document/10.4230/LIPIcs.ITCS.2022.68},
  URN =		{urn:nbn:de:0030-drops-156649},
  doi =		{10.4230/LIPIcs.ITCS.2022.68}
}

@misc{BhartiHaugTanggara2026,
      title={Fault-tolerant quantum computation cannot be achieved with constant spacetime overhead}, 
      author={Kishor Bharti and Tobias Haug and Andrew Tanggara},
      year={2026},
      eprint={2608.26272},
      archivePrefix={arXiv},
      primaryClass={quant-ph},
      note = {arXiv:2608.26272},
      url={https://arxiv.org/abs/2608.26272}, 
}
